\documentclass[11pt,USenglish,thm-restate]{article}

\usepackage{setspace}
\usepackage{subcaption}

\usepackage{authblk}

\usepackage[margin=1.2in,letterpaper]{geometry}

\usepackage[utf8]{inputenc}

\usepackage{anyfontsize}

\usepackage{amssymb}
\usepackage{amsthm}
\usepackage{amsmath}
\usepackage{amsfonts}
\usepackage{nccmath}
\usepackage{bbm}
\usepackage[b]{esvect}
\usepackage{mathrsfs}
\usepackage{mathtools}
\usepackage{centernot}
\usepackage{xspace}
\usepackage{stmaryrd}

\usepackage[dvipsnames]{xcolor}
\usepackage{verbatim}

\usepackage{graphicx}
\usepackage{lscape}
\usepackage{float}
\usepackage{supertabular}
\usepackage{multirow}
\usepackage{longtable}
\usepackage{enumerate}
\usepackage[inline]{enumitem}
\usepackage[framemethod=tikz]{mdframed} %and thus tikz
\usepackage{algpseudocode}

\usepackage{hyperref} %Load last if possible
\usepackage[capitalize]{cleveref}

\usepackage{makecell}

\usepackage{thm-restate}

\usepackage{caption}
\usepackage{subcaption}

\usepackage{textcomp}
\usepackage{pgfplots}
\pgfplotsset{compat=1.18}
\usepgfplotslibrary{fillbetween}
\newtheorem{theorem}{Theorem}
\newtheorem{lemma}{Lemma}
\newtheorem{corollary}{Corollary}
\newtheorem{definition}{Definition}

\makeatletter
\DeclareRobustCommand*{\bfseries}{%
  \not@math@alphabet\bfseries\mathbf
  \fontseries\bfdefault\selectfont
  \boldmath
}
\makeatother

\renewcommand{\paragraph}[1]{\noindent{\bf #1}}
\RequirePackage{totcount}
\RequirePackage{color}
\RequirePackage[colorinlistoftodos]{todonotes}

\newtotcounter{notecount}
\newcommand{\notewarning}{%
\ifnum\totvalue{notecount}>0%
 \vspace{1ex}
\begin{center}
 \begin{tikzpicture}[baseline=(A.south)]
    \node (A) [] at (0,0){};
    \node [rounded corners=1pt,rectangle, draw=red, fill=red!20,text=black](B) at (0.1ex,0ex){
        \Large \raggedright {\bf Warning:} There are still some notes left!
    };
 \end{tikzpicture}
\end{center}
 \vspace{1ex}
\fi
}
\makeatletter
\def\myaddcontentsline#1#2#3{%
  \addtocontents{#1}{\protect\contentsline{#2}{#3}{Section \thesubsection\ at p. \thepage}{}}}
\renewcommand{\@todonotes@addElementToListOfTodos}{%
    \if@todonotes@colorinlistoftodos%
        \myaddcontentsline{tdo}{todo}{{%
            \colorbox{\@todonotes@currentbackgroundcolor}%
                {\textcolor{\@todonotes@currentbackgroundcolor}{o}}%
            \ \@todonotes@caption}}%
    \else%
        \myaddcontentsline{tdo}{todo}{{\@todonotes@caption}}%
   \fi}%
\newcommand*\mylistoftodos{%
  \begingroup
       \setbox\@tempboxa\hbox{Section 9.9 at p. 99}%
       \renewcommand*\@tocrmarg{\the\wd\@tempboxa}%
       \renewcommand*\@pnumwidth{\the\wd\@tempboxa}%
       \listoftodos%
  \endgroup
}
\makeatother
\definecolor{lightgreen}{rgb}{0.86, 0.93, 0.78}
\definecolor{bordergreen}{rgb}{0.55, 0.76, 0.74}
\definecolor{lightblue}{rgb}{0.70, 0.90, 0.99}
\definecolor{borderblue}{rgb}{0.01, 0.66, 0.96}
\definecolor{lightamber}{rgb}{1, 0.93, 0.70}
\definecolor{borderamber}{rgb}{1, 0.76, 0.03}
\definecolor{lightcolor4}{rgb}{ 0.93, 0.70, 1}
\definecolor{bordercolor4}{rgb}{0.76, 0.03, 1}
\definecolor{lightcolor5}{rgb}{0.78,0.86,0.93}
\definecolor{bordercolor5}{rgb}{0.74,0.55,0.76}

\RequirePackage{algpseudocode}
\newcommand{\algoHead}[1]{\underline{\textbf{#1}}}

\makeatletter
\algnewcommand{\ExtendedState}[1]{\State
\parbox[t]{\dimexpr\linewidth-\ALG@thistlm}{\hangindent=\algorithmicindent\strut\hangafter=3#1\strut}}
\makeatother
\algnewcommand\algorithmicinput{\textbf{Input:}}
\algnewcommand\Input{\item[\algorithmicinput]}

\algrenewcommand{\algorithmiccomment}[1]{{\color{gray}// #1}}
\algnewcommand{\IIf}[1]{\State\algorithmicif\ #1\ \algorithmicthen}
\algnewcommand{\EndIIf}{\unskip\ \algorithmicend\ \algorithmicif}

 \RequirePackage{varwidth}
 \RequirePackage{color}
 \RequirePackage[most]{tcolorbox}%with most option

\newtcolorbox{titlebox}[5]{enhanced,center,colframe=black,colback=white,boxrule={#3},arc={#2},auto outer arc,%
 breakable,
 pad at break*=5pt,vfill before first,before={},after={},top=12pt,left=4pt,%
 enlarge top by=2pt,%enlarge bottom by=7pt,%
 fontupper=\small,
 title={\rule[-.3\baselineskip]{0pt}{\baselineskip}\normalsize\sffamily\bfseries #1}, varwidth boxed title*=-30pt, 
 attach boxed title to top left={yshift=-10pt,xshift=10pt}, coltitle=black,
 boxed title style={colback=white,boxrule={#5},arc={#4},auto outer arc}
}

 \newenvironment{dianabox}[1]
 {\begin{titlebox}{\normalfont #1}{0.5pt}{0.5pt}{1pt}{0.75pt}}
 {\end{titlebox}}

\makeatletter
\let\orgdescriptionlabel\descriptionlabel
\renewcommand*{\descriptionlabel}[1]{%
  \let\orglabel\label
  \let\label\@gobble
  \phantomsection
  \edef\@currentlabel{#1}%
  \let\label\orglabel
  \orgdescriptionlabel{#1}%
}
\makeatother

\newcommand{\xmath}[1]{\ensuremath{#1}\xspace}

\newcommand{\parameter}[1]{\xmath{\texttt{#1}}}

\let\emptyset\varnothing

\DeclarePairedDelimiter\abs{\big\lvert}{\big\rvert}% |#1|
\newcommand{\party}{P\xspace}
\newcommand{\PS}{\mathcal{P}} 
\newcommand{\false}{\parameter{false}}
\newcommand{\true}{\parameter{true}}

\newcommand{\poly}{\mathsf{poly}}

\newcommand{\msg}{\mathtt{msg}}

\newcommand{\msgset}{\mathcal{M}}

\newcommand{\gcc}[2]{\Pi_{\text{GC}}}

\newcommand{\inputt}[0]{\mathsf{in}}
\newcommand{\outputt}[0]{\mathsf{out}}
\newcommand{\ba}[0]{\textsc{BA}}
\newcommand{\approximateagreement}[0]{\textsc{AA}}

\newcommand{\concat}[0]{\mathbin\Vert}

\newcommand{\gba}[0]{\textsc{gBA}}

\newcommand{\ca}[0]{\textsc{CA}}
\newcommand{\plainbc}[0]{\textsc{BC}}

\newcommand{\rs}[0]{\textsc{RS}}
\newcommand{\mt}[0]{\textsc{MT}}

\newcommand{\leftt}[0]{\textsc{left}}
\newcommand{\rightt}[0]{\textsc{right}}

\newcommand{\bitcost}[0]{\textsc{bits}}
\newcommand{\roundcost}[0]{\textsc{rounds}}

\newcommand{\naturalnumbers}[0]{\mathbb{N}}

\newcommand{\encode}[0]{\textsc{encode}}
\newcommand{\decode}[0]{\textsc{decode}}

\newcommand{\share}[0]{\textsc{s}}

\newcommand{\buildd}[0]{\textsc{Build}}
\newcommand{\verifyy}[0]{\textsc{Verify}}

\newcommand{\enc}[0]{\mathsf{enc}}

\newcommand{\largebitcostca}[0]{\mathsf{HighCostCA}}

\newcommand{\supernodefactor}[0]{\gamma}

\newcommand{\initSupernodes}[0]{\mathsf{initSupernodes}}

\newcommand{\supernodeCA}[0]{\mathsf{supernodeCA}}

\newcommand{\reduceSupernodes}[0]{\mathsf{reduceSupernodes}}

\newcommand{\supersend}[0]{\mathsf{supersend}}

\newcommand{\sendShares}[0]{\mathsf{sendShares}}

\newcommand{\reconstructFromShares}[0]{\mathsf{recFromShares}}

\newcommand{\supernodesCA}[0]{\mathsf{committeeCA}}

\newcommand{\botCommitteeCA}[0]{\bot\text{-}\mathsf{committeeCA}^{\widetilde{L}}}

\newcommand{\hull}[1]{\langle{#1}\rangle}

\newcommand{\safe}[0]{\mathsf{safe}}
\newcommand{\restrict}[0]{\mathsf{restrict}}
\newcommand{\threshbad}[0]{F_{\text{bad}}}

\newcommand{\bigO}{\mathcal{O}}

\newcommand{\exponentialSearch}[0]{\mathsf{naiveApproximation}}

\usepackage{cite}

\title{General Convex Agreement with Near-Optimal Communication}

\author[1]{Marc Dufay}
\author[2]{Diana Ghinea}
\author[1]{Anton Paramonov}

\affil[1]{ETH Zurich\\
\texttt{\{mdufay,aparamonov\}@ethz.ch}}
\affil[2]{Lucerne University of Applied Sciences and Arts\\
\texttt{diana.ghinea@hslu.ch}}

\date{}

\begin{document}
\maketitle

\fontsize{11pt}{13pt}\selectfont 
\singlespacing
\begin{abstract}
\normalsize
Byzantine Agreement (BA) considers a setting of $n$ parties out of which up to $t$ can be byzantine (malicious), and requires the honest parties to agree on an input subject to a condition called \emph{validity}: if all honest parties have input $v$, the output agreed upon must be $v$. 
Convex Agreement (CA) strengthens BA by requiring the output agreed upon to lie in the convex hull of the honest parties' inputs. This validity condition captures aggregation tasks, such as robust learning and sensor fusion, where honest inputs may differ but should still constrain the final decision. Existing protocols for CA over general convexity spaces require at least $O(L \cdot n^2)$ bits of communication for $L$-bit inputs,
leaving a gap with BA's $\Omega(L \cdot n)$ lower bound.

We investigate this gap, and we present deterministic synchronous CA protocols with near-optimal communication complexity in the long-message regime. When $L=\Omega(n\cdot\kappa)$, where $\kappa$ is a security parameter, our protocols use $\mathcal{O}(L\cdot n\log n)$ bits of communication for finite convexity spaces and $\mathcal{O}(L\cdot n^{1+o(1)})$ communication for Euclidean spaces $\mathbb{R}^d$. Our protocols also have asymptotically optimal round complexity $\mathcal{O}(n)$. If an
upper bound  $L$ on the honest inputs' length in bits is known in advance, we achieve near-optimal resilience
$t<n/(\omega+\varepsilon)$ for any constant $\varepsilon>0$, where $\omega$ is the Helly number
of the convexity space. When no such bound is known, we achieve resilience
$t<n/(\omega+\varepsilon+1)$.

As a sample application, we show how our protocols can be used to obtain efficient solutions for parallel instances
of BA.

Our main technical contribution is the use of extractor graphs to obtain a deterministic assignment of parties to committees, which is robust against adaptive adversaries.
\end{abstract}

\thispagestyle{empty} 

\newpage
\pagenumbering{arabic}

\section{Introduction}

Byzantine Agreement ($\ba$) \cite{LaShPe82} considers a setting of $n$ parties, of which up to
$t$ may be byzantine. Informally, $\ba$ asks the honest parties to ``decide the same thing''.
However, agreement is useful only if the decision is also meaningful, or
credible relative to the honest parties' inputs. This is handled by a property of $\ba$ called \emph{validity}. The standard definition only imposes restrictions on the output agreed upon when the honest parties hold the same input -- otherwise, any output is allowed. This guarantee is somewhat weak especially in scenarios where honest inputs are different, but still ``close''. Examples include: blockchain oracle mechanisms that aggregate multiple price reports into a single outcome \cite{delphi24}, aggregating model updates in distributed or federated learning \cite{SuVai16,federated20,federated21,SPAA:CMMS25}, and robust multi-sensor fusion \cite{sensor24}.

Convex Agreement ($\ca$) \cite{PODC:VaiGar13, DISC:NoRy19, JACM:DLPSW86} strengthens the validity condition of $\ba$ by requiring that the common output lies in the honest inputs' convex hull (e.g., if the inputs are real values, inside the range spanned by the honest inputs), which captures the principle that the decision should be explainable as a ``mixture'' of honest views.

$\ca$ immediately inherits the lower bounds known for $\ba$ regarding efficiency. Matching $\ba$'s lower bound regarding round complexity in a synchronous network (where parties hold synchronized clocks and messages get delivered within a predefined amount of time) is immediate via Byzantine Broadcast ($\plainbc$) \cite{PSL80, DolStr83}: this provides the parties with an identical view over the inputs, after which every party can locally compute a valid output using a deterministic ``safe area'' rule.  Such an approach also achieves optimal resilience \cite{PODC:VaiGar13, DISC:NoRy19, eprint:ConvexWorld}.

On the other hand, approaching $\ba$'s asymptotically tight lower bound for communication complexity becomes challenging. If the honest parties' inputs can be represented using $L$ bits, $\plainbc$-based approaches, such as the one described above, incur at least $\bigO(L n^2)$ bits of communication. In fact, such approaches require honest parties to relay messages from byzantine parties, leading to an adversarially-chosen communication complexity. This stands in contrast to $\ba$, where \emph{extension protocol} techniques allow $\bigO(L n + \poly(n,\kappa))$ bits of communication to be sufficient even via deterministic protocols, where $\kappa$ denotes a security parameter, matching the lower bound of $\Omega(L n)$ bits \cite{PODC:FitHir06,PODC:LiaVai11,PODC:GanPat16,DISC:NRSVX20,AC:BLLN22}. A natural hope is then to port this communication efficiency from $\ba$ to $\ca$. Recent work makes concrete progress in this direction: \cite{PODC:GhiLiuWat25} proposes a deterministic protocol that, assuming a hash function, achieves $\ca$ for integer inputs with communication complexity $\bigO(L n + \poly(n, \kappa))$, breaking the $\Theta(L n^2)$ of $\plainbc$-based approaches. As pointed out by \cite{PODC:GhiLiuWat25}, matching $\ba$-style communication in $\ca$ is intrinsically more delicate: techniques from  $\ba$ extension protocols involve compressing values and hence lose information about the honest inputs' convex hull.
In this paper, we investigate this communication gap beyond integers, for \emph{general} convexity spaces, where prior protocols achieve at least $\bigO(L n^2)$ bits of communication. We therefore ask  the following:

\begin{center}
    \fbox{%
        \parbox{0.95\linewidth}{%
            \centering
            Given that the parties' inputs can be represented using up to $L$ bits,
            is there a \emph{deterministic} protocol achieving $\ca$ on \emph{any} convexity space $\mathcal{C}$ with communication complexity $o(L n^2)$?
        }
    }
\end{center}

\subsection{Our Contribution}
We answer the question affirmatively in the synchronous model, i.e., where messages get delivered within a predefined amount of time $\Delta$ and parties have synchronized clocks: we show that $\ca$ on convexity space $\mathcal{C}$ can be achieved deterministically with near-optimal communication in the long-message regime, with asymptotically optimal round complexity, and near-optimal resilience.
Concretely, we present deterministic protocols achieving communication complexity $\bigO(L \cdot n \log n + n^2 \kappa \log n)$ when, informally, all elements of the convexity space can be represented using $L$ bits (e.g., in finite spaces), and $\bigO(L \cdot n^{1 + o(1)} + n^2 \kappa \log^{1 + o(1)} n)$ otherwise. Here $\kappa$ is a security parameter that comes from using hashes and signatures. Note that $\Omega(L\cdot n + n^2)$ is a communication complexity lower bound, since $n$ parties need to learn the $L$-bit output and $\Omega(n^2)$ messages are needed even for binary BA \cite{dolev1985bounds}.

At a high level, our approach avoids ``compressing'' values and instead \emph{compresses the system}. We arrange the parties into \emph{supernodes}, where each supernode is a group of parties meant to be \emph{good}, that is, acting as a single identity which holds a value in the honest inputs' convex hull. We start with $n$ supernodes, and we \emph{iteratively} reduce these to a small number of supernodes. Once done, the safe-area approach suffices since supernodes can distribute their values efficiently via prior techniques for extension protocols on $\ba$.
This way, the main challenge in our approach consists of assigning the parties into supernodes in a \emph{deterministic} manner to ensure security against an \emph{adaptive adversary}. In particular, the following holds in each iteration: a large fraction of the supernodes (i) contain sufficiently many honest parties to act as a single identity and (ii) hold feasible values, that is, values in the honest inputs' convex hull.

We achieve both of these properties using \emph{extractors}: roughly, these are bipartite graphs that can be constructed efficiently deterministically but provide (pseudo)random properties. We use an extractor in each iteration to assign parties to new supernodes: parties correspond to the left part, supernodes to the right, and the edge signifies an assignment. Extractor's structure then ensures that property (i) holds. To ensure property (ii), in our protocol, each supernode receives a value decided by a committee of (old) supernodes. If properties (i) and (ii) hold for most of the supernodes in the committee, the value decided in a committee is guaranteed to be feasible. 
It turns out that if the supernodes to committees assignment is also done via an extractor, then properties (i) and (ii) indeed hold for most of the supernodes in most of the committees. Moreover, extractor-based assignment yields small committees. This enables us to rely on the safe-area approach (together with our mechanism for distributing values efficiently among supernodes) to decide a value within the committee.  

We also highlight that our protocols only perform safe area computations on a small number of supernodes/values: $\bigO(1)$ in finite spaces, and $\bigO(\log\log n)$ in spaces such as $\mathbb{R}^d$. This substantially reduces the amount of local computation compared with prior approaches, where parties compute safe areas over views of size $\Theta(n)$.

We present a protocol that, following this approach, achieves $\ca$ whenever an upper bound on the honest inputs' bit-length is known in \Cref{section:general-protocol}. Afterwards, \Cref{section:unknown-L} removes this assumption. Finally, \Cref{section:extractor-magic} is concerned with constructing the extractor-based assignments that our protocols rely on.

We now state our main results, starting with $\ca$ on finite convexity spaces.

\begin{restatable}{theorem}{MainResultFiniteSpace}\label{theo:main-theo-finite}
    Consider a finite convexity space $\mathcal{C}$ with constant Helly number $\omega \geq 2$ where all elements have $L$-bit encodings. Then, for any constant $\varepsilon > 0$, there is a protocol achieving $\ca$ on $\mathcal{C}$ secure against $t < n / (\omega + \varepsilon)$ byzantine corruptions in authenticated settings and $t < n / (\max(3, \omega) + \varepsilon)$ in unauthenticated settings, communication complexity $\mathcal{O}(n L \log n + n^2 \log n  \cdot \kappa)$ and round complexity $\bigO(n)$.
\end{restatable}

The next result considers Euclidean spaces. We formally define later what it means for a value in $\mathbb{R}^d$ to be encodable with $L$ bits.
\begin{restatable}{theorem}{MainResultRd}\label{theo:main-R-d}
    For any constants $\varepsilon > 0$ and $d \geq 2$, there is a protocol achieving $\ca$ on $\mathbb{R}^d$ with straight-line convexity assuming that honest parties' inputs can be encoded using $L$ bits, even when up to $t$ of the parties involved are byzantine, where $t < n / (d + 1 + \varepsilon)$ if $L$ is known a priori, and $t < n / (d + 2 + \varepsilon)$ otherwise. The protocol has communication complexity $\bigO(n^{1 + o(1)} L + n^2 \log^{1+o(1)} n  \cdot \kappa)$  and round complexity $\bigO(n)$.
\end{restatable}

The finite-space result also enables us to batch several $\ca$ instances at the same time in an efficient manner. The parties can bundle their inputs into one tuple and run the protocol once, instead of running a separate protocol for each instance. The communication then depends on the total input length $L=\sum_i L_i$, while the $\poly(n, \kappa)$ overhead is paid only once.

\begin{restatable}{theorem}{MainResultParallel}\label{theo:main-theo-parallel}
    Consider $q$ finite convexity spaces $\mathcal{C}_1, \mathcal{C}_2, \ldots, \mathcal{C}_q$, each of the elements in $\mathcal{C}_i$ has an $L_i$-bit encoding, and let $\omega_i \in \bigO(1)$ be the Helly number of $\mathcal{C}_i$. Additionally, let $\omega_{\max} = \max \omega_i$ and $L = \sum L_i$.
    Then, for any constant $\varepsilon > 0$, there is a protocol achieving $\ca$ simultaneously in all of $\mathcal{C}_1, \mathcal{C}_2, \ldots, \mathcal{C}_q$ even when up to $t < n / (\omega_{\max} + \varepsilon)$ of the $n$ parties involved are byzantine in authenticated settings and up to $t < n / (\max(3, \omega_{\max}) + \varepsilon)$ in unauthenticated settings, with communication complexity $\bigO(n L \log n + n^2 \log n  \cdot \kappa)$ and round complexity $\bigO(n)$.
\end{restatable}

Our approach also leads to results beyond $\ca$, most notably for parallel instances of $\ba$. As $\ba$ is a particular case of $\ca$, we obtain the following corollaries for instances with various input lengths, and for instances with binary inputs.

\begin{corollary}
    For any constant $\varepsilon > 0$,
    it is possible to solve $q$ parallel instances of $\ba$ with inputs of size $L_1,L_2, \ldots,L_q$ with near-optimal resilience $t < n/(2 + \varepsilon)$, near-optimal communication complexity $\mathcal{O}(n L \log n)$, and optimal round complexity $\bigO(n)$, assuming $L = \sum L_i = \Omega(n \kappa)$.
\end{corollary}

\begin{corollary}
    For any $\varepsilon > 0$, it is possible to solve $n \cdot \kappa$ parallel instances of binary $\ba$ with near-optimal resilience $t < n/(2 + \varepsilon)$, near-optimal communication complexity $\mathcal{O}(\kappa \cdot n^2 \log n)$, and optimal round complexity $\mathcal{O}(n)$.
\end{corollary}

\subsection{Related Work}
\paragraph{Convex Agreement.}
The strengthened validity condition requiring honest parties' outputs to lie in the honest inputs' convex hull was introduced in the context of Approximate Agreement ($\approximateagreement$) on real numbers \cite{JACM:DLPSW86}. $\approximateagreement$ relaxes the agreement requirement of $\ba$: parties are allowed to output values that differ by at most a predefined error $\varepsilon>0$.
$\approximateagreement$ has prompted a long line of work focusing, e.g., on optimal round complexity~\cite{Fekete90, PODC:Fekete87, BenDoHo10, ARXIV:FuGhPa25}, improved resilience bounds~\cite{OPODIS:AAD04, PODC:GhLiWa22, PODC:LenLos22}, reducing message complexity \cite{MoseArxivNew}. Moreover, variants beyond real values have been considered: $\mathbb{R}^d$~\cite{PODC:VaiGar13,STOC:MenHer13,SPAA:GhLiWa23,PODC:GhMeMi26} as well as from discrete spaces such as graphs ~\cite{DISC:NoRy19, SIROCCO:Alistarh21, PODC:Ledent21, eprint:ConvexWorld, DISC:ACFR19, ARXIV:FuGhPa25}. For a broader overview on $\approximateagreement$, see \cite{SokDiana}.

Convex Agreement ($\ca$) was first defined for $\mathbb{R}^d$ by Vaidya and Garg~\cite{PODC:VaiGar13,DIST:MHVG15}, and was later defined for abstract convexity spaces~\cite{DISC:NoRy19, eprint:ConvexWorld}. A related direction on $\mathbb{R}$ asks for outputs that satisfy a stronger requirement, namely being close to the median of the honest inputs~\cite{OPODIS:StolWat15, OPODIS:CGHWW23}, or to the $k$-th smallest honest input~\cite{IEEE:MelWat18}. 

Each of these prior protocols involves a step where the parties distribute their (input) values, leading to $\Omega(L\cdot n^2)$ bits of communication for $L$-bit inputs. Recently, \cite{PODC:GhiLiuWat25} has proposed a $\ca$ protocol for integers (assuming a hash function), tolerating $t<n/3$ corruptions, with communication $\bigO(Ln+\kappa\cdot n^2\log^2 n)$. Their approach deviates from prior works on $\ca$ and $\approximateagreement$ by relying on a byzantine version of the \emph{longest common prefix} problem: viewing inputs in $\mathbb{Z}$ as bitstrings, the longest common prefix among honest parties pinpoints a subset of the honest input range, which can then be leveraged to realize $\ca$. For general convexity spaces, however, it is unclear whether prefixes/small parts of the inputs' binary encodings can reveal such information, which causes our work to rely on a different approach.

\paragraph{Other extension protocols.} Protocols for long messages have been the subject of a long line of works. Turpin and Coan~\cite{IPL:TurCoa84} have proposed a reduction from long-messages $\ba$ to short-messages $\ba$ with a communication cost of $\bigO(L n^2)$ bits, given $t < n/3$. In the honest-majority regime, Fitzi and Hirt~\cite{PODC:FitHir06} achieved $\ba$ with asymptotically optimal communication $\bigO(Ln+\poly(n,\kappa))$ (given a universal hash function).
Further works have focused on error-free solutions, aiming to reduce the additive $\poly(n,\kappa)$ overhead, both for the $t<n/3$ settings ~\cite{PODC:LiaVai11,PODC:GanPat16, DISC:NRSVX20} and honest-majority settings~\cite{PODC:GanPat16, DISC:NRSVX20, AC:BLLN22}.
The recent work of \cite{PODC:CDGGKV24} has focused on long-message BA with external validity (i.e., values can be verified via a public predicate) whose communication adapts to the actual number of corruptions $f \leq t$, while remaining asymptotically optimal.
Complementarily, \cite{PODC:CDGGKV25} studies multi-shot $\ba$ and related tasks under classical and external validity, and presents protocols with asymptotically optimal amortized communication for long inputs.
Extension protocols have also been studied for closely related primitives, including $\plainbc$ with $t<n$~\cite{AC:HirRay14, TCC:ChoOst18} and asynchronous Reliable Broadcast~\cite{SRDS:CacTes05, DISC:NRSVX20}.

\paragraph{Protocols for short messages.} 
Reducing communication is an interesting topic for short inputs (i.e., a constant number of bits) as well. Dolev and Reischuk~\cite{PODC:DolRei82} proved that any deterministic $\ba$ protocol must exchange $\Omega(t^2)$ messages/bits in the presence of $t$ byzantine faults, and this is tight~\cite{CoaWel92, DISC:MoRe21}. In contrast, randomized $\ba$ protocols may achieve subquadratic communication ~\cite{SODA:KSSV06,king09,PODC:KinSai10,STOC:GelKom24,TCC:BKLL20,PODC:ACDNPR19,DISC:CohKeiSpi20}.

\paragraph{Pseudorandomness.}
Our approach relies heavily on the use of extractors. Extractors are part of what are called \emph{pseudorandom objects} \cite{pseudorandomness} -- these keep properties about random structures while being fully deterministic. Other common pseudorandom graph structures are, for example, expanders, samplers or condensers. Expanders in particular have been already used multiple times to disseminate information with sparse communication \cite{DISC:MoRe21,PODC:CDGGKV24,DISC25:AdaptiveBA}. As far as we are aware, extractors in distributed computing have only been used to improve randomness when only a weak random source is available \cite{FOCS:KLRZ08}.

\paragraph{Parallel instances of $\ba$.}
As a corollary of our main results, we obtain a near-optimal protocol for parallel $\ba$ on $\Omega(n)$ instances. This has already been studied under the name Multidimensional $\ba$ by \cite{MultiDimBA}, who propose a randomized protocol, but do not specify its communication complexity. A closely related problem is parallel Broadcast, also called Interactive Consistency, which can be solved with optimal bit complexity \cite{PODC:CDGGKV24}. We note that parallel $\ba$ implies parallel Broadcast with the same complexity but the opposite is not true. Another similar area is Multishot Agreement, where instances no longer have to be processed at the same time and can depend on each other. Civit et al. \cite{PODC:CDGGKV25} present a protocol for this problem, but it requires $\Omega(n^3)$ instances to become optimal.

\section{Preliminaries}

We describe our model, and present a few key concepts and definitions.

\paragraph{Model.}
We consider a setting with $n$ parties $\party_1, \party_2, \dots, \party_n$ in a fully-connected network where links model authenticated channels. 
The network is synchronous: the parties hold synchronized clocks, and messages are delivered within a publicly known amount of time $\Delta$.

\paragraph{Adversary.}
We consider a computationally-bounded adaptive adversary that can corrupt up to $t$ parties at any point in the protocol's execution, causing them to become byzantine: corrupted parties may deviate arbitrarily from the protocol.

\paragraph{Cryptographic assumptions.}
We consider a security parameter $\kappa$, a collision-resistant hash function $H_{\kappa} : \{0, 1\}^{\star} \rightarrow \{0, 1\}^\kappa$. Informally, $H_\kappa : \{0, 1\}^{\star} \rightarrow \{0, 1\}^\kappa$ is collision-resistant if, for any computationally-bounded adversary $\mathcal{A}$, the probability that $\mathcal{A}(1^{\kappa})$ outputs two values $x \neq y$ such that $H_{\kappa}(x) = H_{\kappa}(y)$ is negligible in $\kappa$ \cite{Rog06}.

We consider both \emph{unauthenticated} settings and \emph{authenticated} settings. The term \emph{unauthenticated} refers to a setting where the only cryptographic assumption is the hash function. In contrast, \emph{authenticated settings}, on top of the hash function, assume a public key infrastructure (PKI) and a secure signature scheme, with security parameter $\kappa$. 

For simplicity of presentation, our proofs assume that the adversary is unable to find any collisions for $H_{\kappa}$ or to forge signatures. When these primitives are replaced with real-world instantiations, our protocols' guarantees still hold except with probability negligible in $\kappa$.

\paragraph{Byzantine Agreement.} Our protocols rely on $\ba$ as a building block, defined below \cite{LaShPe82}.

\begin{definition}[Byzantine Agreement]
Let $\Pi$ be an $n$-party protocol where each party holds a value $v_{\inputt}$ as input, and parties terminate upon generating an output $v_{\outputt}$. $\Pi$ achieves $\ba$ if the following properties hold even when $t$ of the $n$ parties are corrupted:
\textbf{(Termination)} All honest parties terminate;
\textbf{(Validity)} If all honest parties hold the same input value $v$, they output $v_{\outputt} = v$;
\textbf{(Agreement)} All honest parties output the same value.
\end{definition}

\paragraph{Basic convexity notions and $\ca$.}
Before defining $\ca$, we recall a few notions on convexity spaces.
Let $\mathcal{C}$ be a non-empty set of values.
A \emph{convexity notion} on $\mathcal{C}$ is a family $\mathcal{K} \subseteq 2^{\mathcal{C}}$
such that $\varnothing, \mathcal{C} \in \mathcal{K}$ and $\mathcal{K}$ is closed under arbitrary intersections.
Members of $\mathcal{K}$ are called \emph{convex sets}.
For simplicity, we will often refer to $\mathcal{C}$ itself as a \emph{convexity space}, leaving the underlying family $\mathcal{K}$ implicit when it is clear from context.
This framework captures \emph{straight-line convexity} on $\mathcal{C}=\mathbb{R}^d$ by taking $\mathcal{K}$ to be the sets closed under line segments, but it also allows other choices, e.g., \emph{box convexity}, where $\mathcal{K}$ consists of axis-aligned boxes $I_1 \times \cdots \times I_D$ with each $I_i$ a closed interval.
For $S \subseteq \mathcal{C}$, the \emph{convex hull} of $S$, denoted by $\hull{S}$, is the intersection of all convex sets containing $S$, i.e.
$\hull{S} = \bigcap \{ K \in \mathcal{K} : S \subseteq K \}$.
Note that $\hull{S}$ is convex by intersection closure.
We may then state the definition of $\ca$ \cite{PODC:VaiGar13, DISC:NoRy19}. Throughout the paper, \emph{valid value} refers to a value satisfying Convex Validity, as defined below.
\begin{definition}[Convex Agreement]
Let $\Pi$ be an $n$-party protocol where each party holds a value $v_{\inputt}$ as input, and parties terminate upon generating an output $v_{\outputt}$. $\Pi$ achieves $\ca$ if the following properties hold even when $t$ of the $n$ parties are corrupted:
\textbf{(Termination)} All honest parties terminate; \textbf{(Convex Validity)} Honest parties' outputs lie in the honest inputs' convex hull; \textbf{(Agreement)} All honest parties output the same value.
\end{definition}

Optimal resilience thresholds for $\ca$ on a convexity space depend on a feature of the space called the \emph{Helly number}, denoted by $\omega$. In $\mathbb{R}^d$ with standard (line-segment) convexity, Helly's theorem states that if every subcollection of size $D+1$ of a finite family of convex sets intersects, then the whole family intersects \cite{helly}. 
Motivated by this, the \emph{Helly number} $\omega$ of a convexity space $\mathcal{C}$ is the smallest $h$ such that for every finite family $\{C_i\}\subseteq\mathcal{K}$, if every $h$ of the sets intersect, then $\bigcap_i C_i\neq\varnothing$. For example, $\omega=D+1$ for standard convexity in $\mathbb{R}^d$, while $\omega=2$ for box convexity. We note that, if a space has Helly number $\omega = 1$, then there is an element which appears in the convex hull of every set. On such space, $\ca$ can be solved with no communication. Therefore, our paper assumes $\omega \geq 2$.

\paragraph{Resilience and notations.}
We recall that, in the synchronous model, $\ba$ can be achieved up to $t < n/3$ corruptions in unauthenticated settings, and $t < n/2$ corruptions in authenticated settings.
For $\ca$ on a convexity space with Helly number $\omega$, the optimal resilience thresholds are $t < n / \max(3, \omega)$ in unauthenticated settings, and $n / \omega$ in authenticated settings \cite{PODC:VaiGar13, DISC:NoRy19, eprint:ConvexWorld}. $t < n / 3$ is necessary in unauthenticated settings due to \emph{Agreement} for $\ba$ and hence also for $\ca$, while $t < n / \omega$ comes from the stronger validity condition in $\ca$.
For simplicity, as our work considers both authenticated and unauthenticated settings, we define the variable $\eta$, which refers to the denominators in the $\ba$ optimal thresholds: $\eta = 2$ in authenticated settings, and $\eta = 3$ in unauthenticated settings. Then, we define $\omega' = \max(\omega, \eta)$ so that we can focus on providing a protocol solving $\ca$ if $t < n/(\omega' + \varepsilon)$ regardless of the setting.

\paragraph{Bit representations.}
Throughout the paper, even though parties hold values from a convexity space, they must also be manipulated as bitstrings for communication purposes. Since not every convexity space admits an encoding of all its elements, we may have to distinguish between the original convexity space and the subset of values that may appear as inputs.

As such, we always assume that every party is given the same encoding function $\enc : X \to \{0,1\}^\star$, where $X \subseteq \mathcal{C}$ and $\enc$ is injective. The encoding of an element $x \in X$ is defined as $\enc(x)$. We note that this implies that $X$ is at most countably infinite, while $\mathcal{C}$ may not be. Honest parties' inputs are said to be encoded using $L$ bits if they are all in $X$ and for every such input $v \in X$, $|\enc(v)| \leq L$. We call $X$ the input domain of $\mathcal{C}$. 

We note that, in our $\ca$ protocols, parties have to compute intermediate values: these may have encodings that are longer than the honest inputs' bit-length, or may even fall outside of $X$.
Given a convexity space $\mathcal{C}$ with encoding $\enc$ and input domain $X$, we say $\enc^\star : X^\star \to \{0,1\}^\star$ is an \emph{extension} of $\enc$ if $X \subseteq X^\star \subseteq \mathcal{C}$ and, for every $v \in X$, $|\enc^\star(v)| \leq |\enc(v)|+1$. The set $X^\star$ represents possible intermediate outputs during the protocol. Note that $X^\star$ may be strictly larger than $X$.
In our protocols, all elements of $\mathcal{C}$ exchanged between parties will be encoded using such an \emph{extension} of $\enc$ which contains all possible intermediate values. To ensure that such an extension exists and to control the growth of the intermediate values' representations, we introduce the following notion:
\begin{definition}[Dilation factor]\label{def:dilation-factor}
Let $\mathcal{C}$ be an encoded convexity space with Helly number $\omega$, encoding $\enc$ and input domain $X$. For a given $\delta > 0$, we say that $\mathcal{C}$ has dilation factor $\delta$ if there exists a set $X^\star$ and an extension $\enc^\star : X^\star \to \{0,1\}^\star$ of $\enc$ such that $X \subseteq X^\star \subseteq \mathcal{C}$, and for any $L \geq \log(\omega)$ and any non-empty finite multiset $S \subseteq X^\star$ whose elements have $\enc^\star$-encodings of length at most $L$, there is an element $v \in \bigcap_{A \subseteq S,\ |S \setminus A| < |S|/\omega} \hull{A}$ such that $v \in X^\star$ and $|\enc^\star(v)| \leq \delta \cdot L$.
\end{definition}

We say that a convexity space $\mathcal{C}$ with a given encoding has finite dilation number if there is a $\delta > 0$ such that $\mathcal{C}$ has dilation factor $\delta$. We analyze the dilation factor for finite convexity spaces, where the input domain is the whole space, and for straight-line convexity in $\mathbb{R}^d$, where honest inputs come from an encoded input domain $X \subseteq \mathbb{R}^d$.

The next theorem is a direct consequence of the definition of the Helly number and can be obtained using $X^\star=X$ and $\enc^\star=\enc$.
\begin{theorem}\label{theo:dilation-factor-finite}
Let $\mathcal{C}$ be a finite convexity space with Helly number $\omega$, where every element is encoded using $L$ bits. Then $\mathcal{C}$ has dilation factor $\delta := 1$.
\end{theorem}

We now consider Euclidean spaces. Note that $\mathbb{R}^d$ is uncountably infinite, which means that any encoding would restrict the parties' input to a strict subset of $\mathbb{R}^d$ which is at most countably infinite. We defer the proof of the theorem below to \Cref{appendix:dilation-factor}.
\begin{restatable}{theorem}{DilationFactorR}\label{theo:dilation-factor-euclidian}
    Let $d > 0$, then for any encoding, $\mathbb{R}^d$ has dilation factor $\delta := 4 d^2$.
\end{restatable}

\paragraph{Communication complexity and round complexity.} 
We denote the communication complexity for $L$-bit inputs of an $n$-party protocol $\Pi$ by $\bitcost_{L, n}(\Pi)$: this is the worst-case total number of bits sent by honest parties if they all hold inputs which can be encoded using at most $L$ bits.
In addition, $\roundcost_{n}(\Pi)$ denotes the worst-case round complexity of $\Pi$. We add that for some of our subprotocols we will use different subscripts for the $\bitcost$ notation, and the meaning will be clear in context.

\section{Protocol for Fixed Inputs' Bit-Length} \label{section:general-protocol}

In this section, we assume a publicly known upper bound $L$ on the honest inputs' bit-length and we present our core $\ca$ protocol $\Pi_{\ca}^L$.  Without this assumption, a byzantine party could follow the protocol correctly with an arbitrarily long input, and force honest parties to relay it, impacting the communication complexity. We remove this assumption in \Cref{section:unknown-L}, where we describe a mechanism that estimates $L$ and enables us to rely on $\Pi_{\ca}^L$, at the cost of  slightly weaker resilience guarantees.

At a high level, $\Pi_{\ca}^L$ relies on reducing the number of inputs from up to $n$ valid inputs to up to a small number $\supernodefactor$ of valid values, where $\supernodefactor$ is defined as $\supernodefactor := 2$ if the input space has dilation factor $\delta = 1$, and $\supernodefactor := \log \log n$ otherwise.
To reduce the number of inputs, we proceed in iterations: the parties are initially arranged into $N := n$ \emph{supernodes} such that each (\emph{good}) supernode holds a valid value. Each iteration rearranges these into $\lfloor N/\supernodefactor \rfloor$ supernodes, while maintaining the same fraction of good supernodes. Once $\supernodefactor$ supernodes/values are reached, we can rely on a safe-area approach to achieve $\ca$.

\paragraph{Supernodes.} Supernodes are (possibly intersecting) groups of parties of equal size -- these should be acting as a single identity. The assignment is known by all parties at all times. We distinguish between three types of supernodes: \emph{byzantine}, \emph{confused} and \emph{good}.

A supernode $S$ consisting of $n_S$ parties is \textbf{byzantine} if it contains at least $\frac{n_S}{\omega' + \varepsilon / 2}$ byzantine parties. This threshold is chosen so that $\ba$ can be achieved within $S$ and $\ca$ can be achieved within committees we will design using supernodes.

For a non-byzantine supernode $S$, we may define the \textbf{input} of $S$: this is a value that all honest parties in $S$ hold for $S$.
Non-byzantine supernodes may not always hold valid inputs, due to past interactions with byzantine parties: we say that a supernode is \textbf{confused} if it is non-byzantine but holds a non-valid input. 
If a supernode is either byzantine or confused, we call it \textbf{bad}.
Finally, if supernode $S$ is neither byzantine nor confused, then it is \textbf{good}.

An invariant we will preserve throughout our protocol is that the proportion of bad supernodes is at most $\threshbad = \min\left(\frac{1}{\omega' + \varepsilon},  \frac{1}{\eta + \varepsilon / 3} - \frac{1}{\omega' + \varepsilon/2} \right)$. Looking ahead, the first part of the threshold $\threshbad$ ensures that committees of supernodes have at most a $1/\omega$-proportion of bad supernodes and can therefore compute a valid output, and the second part is crafted so that $\ba$ can be achieved within our committees.

\paragraph{Initializing the supernodes.} We implement the initial assignment of the $n$ parties into $N := n$ supernodes satisfying the invariant above in a subprotocol $\initSupernodes^L$, which we describe in \Cref{subsection:init-supernodes}. 
$\initSupernodes^L$ prepares the initial assignment using an extractor. Then, to provide an input to each supernode $S$, the parties in supernode $S$ run a $\ca$ protocol.

\paragraph{Reducing the number of supernodes.} 
When reducing the number of supernodes from $N$ to $N' = \lfloor N/ \supernodefactor \rfloor$, we want to maintain the threshold $\threshbad$ of bad supernodes. This step is implemented in $\reduceSupernodes^L$, described in \Cref{subsection:reduce-supernodes}. Roughly, this proceeds in two parts: (1) it defines the new $N'$ supernodes using an extractor, and (2) for each (new) supernode $S_i$, it defines a committee $C_i$ of (old) supernodes that is responsible for computing a valid input for $S_i$. The committee assignment also relies on an extractor. Then, to compute $S_i$'s input, the supernodes in $C_i$ run a subprotocol  $\supernodesCA^L$ that achieves $\ca$ among these in a communication-efficient manner. We describe $\supernodesCA^L$ in \Cref{subsection:supernodes-ca}.

\paragraph{Final step.} Once we reach at most $\supernodefactor$ supernodes, we run the subprotocol $\supernodesCA^L$, which provides each honest party in these supernodes with a valid value in a communication-efficient manner. It is possible, however, that not all honest parties are assigned to supernodes. Hence, in this final step, the parties among the remaining up to $\gamma$ supernodes distribute the final value to all parties via a communication-efficient subprotocol called $\supersend^L$, which is roughly an all-to-all variant of $\plainbc$ that we discuss in \Cref{section:supersending}.

\paragraph{Main protocol.} We present the code of $\Pi_{\ca}^L$ below. We note that the $\supersend^L$ call mentions \emph{virtual} parties -- this will be discussed in the subsequent sections, and comes from the fact that supernodes are not necessarily disjoint.
We also recall that the constant $\delta$ is the dilation factor of the convexity space.

\begin{dianabox}{$\Pi_{\ca}^L$}
	\algoHead{Code for party $\party$ with input $v_{\inputt}$}
	\begin{algorithmic}[1]
    \Statex \textbf{Initial assignment:}
    \State Run $\initSupernodes^L$ to get the initial $N := n$ supernodes. 
    \Statex \textbf{Reducing the number of supernodes from $N$ to up to $\supernodefactor$}:
    \State Set $L := \delta \cdot L$.
    \While{$N \geq \supernodefactor$}
        \State Invoke $\reduceSupernodes^L$ to reassign the parties from $N$ to $N' := \lfloor N / \supernodefactor \rfloor$ supernodes.
        \State $N := N', L := \delta \cdot L$.
    \EndWhile
    \Statex \textbf{Final output:}
    \State If you are part of a final supernode, participate in $\supernodesCA^L$ and obtain $v_{\outputt}$.
    \State The virtual parties in the final supernodes invoke $\supersend^{\delta L}$ with input $v_{\outputt}$ as senders, and all $n$ parties participate as receivers. Output the value received from $\supersend^{\delta L}$.
	\end{algorithmic}
\end{dianabox}

$\Pi_{\ca}^L$ enables us to state the following results for spaces with dilation factor $\delta := 1$ and for spaces with finite dilation factor. We defer the proofs for these theorems, along with the formal analysis of $\Pi_{\ca}^L$, to  \Cref{subsection:ca-L-analysis}.

\begin{theorem}\label{theo:abstract-for-1-fixed-L}
Consider an arbitrary constant $\varepsilon > 0$ and a convexity space $\mathcal{C}$ with constant Helly number $\omega \geq 2$ and dilation factor $\delta = 1$. Then, there is a protocol $\Pi_{\ca}^L$ achieving $\ca$ on $\mathcal{C}$ whenever: honest parties hold $L$-bit inputs from $\mathcal{C}$, up to $t$ of the $n$ parties involved are byzantine, where $t < n / (\max(3, \omega) + \varepsilon)$ in unauthenticated settings and $t < n / (\omega  + \varepsilon)$ in authenticated settings. $\Pi_{\ca}^L$ has communication complexity $\bitcost_{L, n}(\Pi_{\ca}^L) =  \bigO(Ln \log n + \kappa \cdot n^2 \cdot \log n)$ and round complexity $\roundcost_{n}(\Pi_\ca^L) = \bigO(n)$.
\end{theorem}

\begin{theorem}\label{theo:abstract-for-finite-fixed-L}
Consider an arbitrary constant $\varepsilon > 0$ and a convexity space $\mathcal{C}$ with constant Helly number $\omega \geq 2$ and finite dilation factor $\delta$. Then, there is a protocol $\Pi_{\ca}^L$ achieving $\ca$ on $\mathcal{C}$ whenever: honest parties hold $L$-bit inputs from $\mathcal{C}$, up to $t$ of the $n$ parties involved are byzantine, where $t < n / (\max(3, \omega)  + \varepsilon)$ in unauthenticated settings and $t < n / (\omega  + \varepsilon)$ in authenticated settings. $\Pi_{\ca}^L$ has communication complexity $\bitcost_{L, n}(\Pi_{\ca}^L) =  \bigO(Ln^{1+o(1)} + \kappa \cdot n^2 \cdot \log^{1+o(1)} n)$ and round complexity $\roundcost_{n}(\Pi_\ca^L) = \bigO(n)$.
\end{theorem}

These results already enable us to discuss some of the claims in our paper's introduction. \Cref{theo:main-theo-finite} becomes a consequence of \Cref{theo:abstract-for-1-fixed-L} as, 
according to \Cref{theo:dilation-factor-finite}, spaces where each element has an $L$-bit label have dilation factor $\delta = 1$. 

Further, recall \Cref{theo:main-theo-parallel}, which states that multiple convex agreement instances can be efficiently run in parallel. This is because one can obtain a single large convex space as a product of convex spaces (see \Cref{theo:parallel-is-easy} below) and run a single convex agreement on this large space, ``automatically'' obtaining a convex agreement along each axis. 
\begin{restatable}{theorem}{ParallelIsEasy} \label{theo:parallel-is-easy}
    Let $q \geq 1$, and consider $q$ finite convexity spaces $\mathcal{C}_1, \mathcal{C}_2, \ldots, \mathcal{C}_q$. Assume that convexity space $\mathcal{C}_i$ has Helly number $\omega_i$, and that its elements can be represented using $L_i$ bits, and let $\omega_{\max} = \max \omega_i$ and $L = \sum L_i$.
    Then, there is a finite convexity space $\mathcal{C}'$ with Helly number at most $\omega_{\max}$ whose elements can be represented using $L$ bits such that: $\mathcal{C}'$ can be used to represent the convexity relation on all spaces $\mathcal{C}_1, \mathcal{C}_2, \ldots, \mathcal{C}_q$ at the same time.
\end{restatable}

We prove this theorem in \Cref{section:parallel-composition}.

\subsection{Building Block: $\ba$ within Supernodes}\label{section:ba-main}
We rely on a $\ba$ protocol $\Pi_{\ba}$ that achieves optimal communication for short messages: this will be run in various settings, e.g., on sets of \emph{supernodes}. As a party may belong to multiple supernodes, the same party may appear in the protocol multiple times, hence simulating multiple \emph{virtual parties}. We say that a virtual party is honest if it is simulated by an honest party, and that a virtual party is byzantine if it is simulated by a byzantine party. 
We then consider the following property:

\vspace{0.2em}
\noindent \textbf{(Multi-Participation Safety)} If $\Pi_{\ba}$ achieves $\ba$ in an $n$-party setting where up to $t$ parties  are byzantine, then it also achieves $\ba$ even when it runs in a setting of $n$ \emph{virtual} parties out of which up to $t$ may be simulated by byzantine parties.
\vspace{0.2em}

In addition, we will be running $\Pi_{\ba}$ in settings where the number of (virtual) byzantine parties may exceed the number of corruptions that $\Pi_{\ba}$ can tolerate. We need to ensure that it maintains its communication complexity and round complexity even in such settings.

\vspace{0.2em}
\noindent \textbf{(Extra-Corruptions Safety)} If $\Pi_{\ba}$ has round complexity $\roundcost_{n}(\Pi_{\ba})$ and the honest parties send $\bitcost_{L, n}(\Pi_{\ba})$ bits in total by round $\roundcost_{n}(\Pi_{\ba})$ in an $n$-party setting where up to $t$ parties are byzantine, then: in a setting of $n$ (virtual) parties where the number of \emph{(virtual)} byzantine parties exceeds $t$, honest parties still complete the execution of  $\Pi_{\ba}$ within $\roundcost_{n}(\Pi_{\ba})$ rounds, and still send  $\bitcost_{L, n}(\Pi_{\ba})$ bits in total by round $\roundcost_{n}(\Pi_{\ba})$.
\vspace{0.2em}

These properties naturally apply to many unauthenticated protocols, and also to authenticated protocols that rely solely on PKI and signatures (notably, without threshold signatures). 
As our resilience threshold is $t < n / \omega' + \varepsilon = n / \max(\omega, \eta) + \varepsilon$ for any $\varepsilon > 0$, we can use the protocol of \Cref{thm:magic-ba-no-pki} for $\omega \geq 3$ or in unauthenticated settings. In authenticated settings, for spaces with  $\omega = 2$, the resilience threshold exceeds $n/3$, and hence we cannot rely on \Cref{thm:magic-ba-no-pki} -- in this case, our resilience threshold enables the protocol of \Cref{thm:magic-ba-pki}. We discuss the details regarding the two additional properties for these protocols in \Cref{section:ba:extra-properties}. We also note that \cite{DISC:MoRe21} analyzes the communication complexity of their protocol assuming inputs of length $L = \bigO(\kappa)$, and we extend their analysis to general $L$ in \Cref{sec:appendix:ba communication}.

\begin{theorem}[Theorem 3 of \cite{UnauthenticatedBA}]\label{thm:magic-ba-no-pki}
    There is an unauthenticated $\ba$ protocol $\Pi_{\ba}$ resilient against $t < n/3$ corruptions that additionally achieves Multi-Participation Safety and Extra-Corruptions Safety. $\Pi_{\ba}$ has round complexity $\roundcost_{n}(\Pi_{\ba}) = \bigO(n)$ and communication complexity $\bitcost_{L, n}(\Pi_{\ba}) = O(L n^2)$.
\end{theorem}

\begin{theorem}[Theorem 2 of \cite{DISC:MoRe21}]\label{thm:magic-ba-pki}
    There is an authenticated $\ba$ protocol $\Pi_{\ba}$ resilient against $t < (1/2 - \varepsilon) \cdot n$ corruptions, where $\varepsilon > 0$ is a positive constant, that additionally achieves Multi-Participation Safety and Extra-Corruptions Safety. $\Pi_{\ba}$ has round complexity $\roundcost_{n}(\Pi_{\ba}) = \bigO(n)$ and communication complexity $\bitcost_{L, n}(\Pi_{\ba}) = \bigO((L + \kappa) \cdot n^2)$, where $\kappa$ denotes the security parameter.
\end{theorem}

For simplicity of presentation, we will be using the term \textbf{$\ba$-friendly} to refer to groups of supernodes/groups of virtual parties where the assumed $\ba$ protocol succeeds, i.e., where the number of byzantine (virtual) parties does not exceed the resilience of the assumed protocol $\Pi_{\ba}$. In settings that are not $\ba$-friendly, we only use Extra-Corruptions Safety to upper-bound the round/communication complexity, without relying on correctness. We note that, regardless of the setting, if the proportion of byzantine (virtual) parties in a group is at most $1/(\eta + \varepsilon/4)$, then this group is $\ba$-friendly.

\subsection{Supersending}
To achieve near-optimal communication complexity overall, we need to ensure that groups of (virtual) parties communicate efficiently. We implement a primitive called \emph{supersending}, described by \Cref{lemma:supersending}, stated below. This is a variant of $\plainbc$ \cite{PSL80,DolStr83} with multiple senders, and relies on prior techniques from extension protocols, e.g., \cite{DISC:NRSVX20,AC:BLLN22, PODC:GhiLiuWat25}.
Roughly, if a group $A$ of $n_A$ (virtual) parties wants to send a message $m$ to a group $B$ of $n_B$ (virtual) parties, each virtual party $\party$ in $A$ will first compute the Reed-Solomon encoding of $m$ -- this provides $\party$ with $n_B$ shares so that any $\lfloor n_B/2\rfloor$ + 1 of the shares reconstruct $m$ \cite{ReedSolomon}. These shares are then accumulated in an encoding $z$ as the root of the Merkle tree constructed from the shares \cite{MerkleTrees}. $\party$ sends a share and Merkle tree witness to every virtual party in $B$. The virtual parties in $B$ agree on an encoding to be reconstructed by using the assumed $\ba$ protocol $\Pi_{\ba}$, and then they attempt to reconstruct the initial message. We present the supersending primitive in detail, along with the proof of \Cref{lemma:supersending}, in \Cref{section:supersending}.

\begin{restatable}{lemma}{SupersendLemma} \label{lemma:supersending}
    Let $A$ and $B$ be two groups of (virtual) parties of size $n_A$ and $n_B$ respectively. Assume that the honest virtual parties in $A$ join $\supersend^L$ with an $L$-bit message $m$ as input. Then, every honest virtual party in $B$ completes $\supersend^L$ with either an $L$-bit output or $\bot$, within $\roundcost_{n_A, n_B}(\supersend^L) = \bigO(n_B)$ rounds, with communication complexity $\bitcost_{L, n_A, n_B}(\supersend^L) = \bigO(L \cdot (n_A + n_B) + \kappa n_B \log n_B \cdot (n_A + n_B))$.

     Moreover, if $B$ is $\ba$-friendly, the honest virtual parties in $B$ obtain the same output $m'$. If, in addition, $A$ contains more than $n_A/2$  honest virtual parties, $m' = m$.
\end{restatable}

\subsection{Initializing the Supernodes} \label{subsection:init-supernodes}

We describe the subprotocol $\initSupernodes^L$, which is responsible for defining the initial $N$ supernodes. To ensure a threshold of up to $\threshbad$ bad supernodes, we rely on the theorem below, which  we prove in \Cref{section:extractor-magic} using an extractor.  We define the function $D(\omega, \varepsilon)$ in \Cref{section:extractor-magic}: for now, as $\omega$ and $\varepsilon$ are constants, we can regard $D(\omega, \varepsilon)$ as a constant.
\begin{restatable}{theorem}{InitialExtract} \label{thm:initial-extractor}
    Assume $n$ parties, out of which up to $\frac{n}{\omega' + \varepsilon}$ can be byzantine. Then, there is a deterministic $\poly(n, 1/\varepsilon, \omega)$-time algorithm  that assigns the $n$ parties to $n$ supernodes such that: at most $\threshbad \cdot n$ supernodes are byzantine;
    every party belongs to at most $D(\omega, \varepsilon)$ many supernodes;
    each supernode contains $D(\omega, \varepsilon)$ parties.
\end{restatable}

Afterwards, we need to provide each supernode with a value. As supernodes consist of $\bigO(1)$ parties, we may rely on the safe-area approach described in our paper's introduction \cite{PODC:VaiGar13, DISC:NoRy19, eprint:ConvexWorld}. We implement this step in a separate subprotocol $\supernodeCA^L$. In this subprotocol, we provide all parties within supernode $S$ with an identical view over the inputs of the parties in $S$. We can achieve this using, e.g., a $\plainbc$ protocol, or the $\ba$ protocols described in \Cref{section:ba-main}. As we assume an upper bound $L$ on the inputs' length, and $S$ has size $\bigO(1)$, this step has communication complexity $\bigO(L + \kappa)$ per supernode.
We then denote the multiset of values in the common view by $\msgset$. If $S$ is a non-byzantine supernode of size $n_S$, it contains up to $t_S < n_S / (\omega' + \varepsilon / 2)$ corruptions -- hence, at least $n_S - t_S$ of the values in $\msgset$ come from honest parties. Then, each party computes a \emph{safe area} $S := \safe_{k}(\msgset)$ with $k := \abs{\msgset} - (n_S - t_S)$, as defined below.

\begin{restatable}[Safe area]{definition}{SafeAreaDef}\label{definition:safe-area}
    Let $\msgset$ be a multiset.
    For a given $k$, the safe area of $\msgset$ is
    $
        \safe_{k}(\msgset) = \bigcap_{M \in \restrict_{k}(\msgset)} \hull{M},
    $
    where $\restrict_{k}(\msgset) = \{M \subseteq \msgset : \abs{M} = 
    \abs{\msgset} - k\}$.
\end{restatable}

The parties apply a deterministic decision to output a value in $S$ (if any), which already implies Agreement and Termination. The lemma below (proven in \Cref{appendix:safe-area-proofs}) follows from \cite[Lemma 15]{eprint:ConvexWorld} and the safe area definition, and ensures that: if $S$ is a non-byzantine supernode, the resulting safe area is non-empty and included in the convex hull of the inputs of the honest parties in $S$, which implies that Convex Validity holds. 
\begin{restatable}{lemma}{NonEmptySafeArea} \label{lemma:safe-area}
    Let $\msgset$ be a multiset of $n - t + k$ values, with $0 \leq k \leq t$, where the values come from a convexity space with Helly number $\omega$. Then, $\safe_k(\msgset) \subseteq \hull{\mathcal{H}}$ for any multiset $\mathcal{H} \subseteq \msgset$ of size $n - t$. Moreover, if  $t < n / \omega$, $\safe_k(\msgset) \neq \emptyset$.
\end{restatable}

We present the lemma describing $\supernodeCA^L$ below. We defer the proof, together with the implementation of $\supernodeCA^L$, to \Cref{appendix:ca-within-supernode}.

\begin{restatable}{lemma}{CAWithinSupernodes}\label{lemma:ca-within-supernode}
    Consider a supernode $S$ containing $n_S$ parties. Then, all honest parties in $S$ complete the execution of $\supernodeCA^L$ within $\roundcost_{n_S}(\supernodeCA^L) = \bigO(n_S)$ rounds, and with $\bitcost_{L, n_S}(\supernodeCA^L) = \bigO((L + \kappa) \cdot n_S^3)$ bits of communication.
    
    If an honest party obtains an output in  $\supernodeCA^L$, then this is a $(\delta \cdot L)$-bit value. In addition, if $S$ is a non-byzantine supernode, i.e., contains up to $t_S < n_S / (\omega' + \varepsilon / 2)$ byzantine parties, then all honest parties in $S$ agree on a $(\delta \cdot L)$-bit value $v$ in their inputs' convex hull.
\end{restatable}

We may now present the code of $\initSupernodes^L$.

\begin{dianabox}{$\initSupernodes^L$}
	\begin{algorithmic}[1]
    \State Assign the $n$ parties into $N = n$ supernodes $S_1, \ldots, S_N$ using \Cref{thm:initial-extractor}.
    \State For each supernode $S_i$, in parallel: 
    \State \hspace{0.5cm} Parties in $S_i$ join $\supernodeCA^L$ to obtain the value $v_i$ assigned to supernode $S_i$.
	\end{algorithmic}
\end{dianabox}

\begin{lemma}\label{lemma:init-supernodes}
    After running $\initSupernodes^L$, the $n$ parties are assigned into $n$ supernodes of size $D(\omega, \varepsilon)$ each such that at most $\threshbad \cdot n$ supernodes are bad, and no honest party holds values longer than $(\delta \cdot L)$ bits for any supernode it is part of. In addition,
    $\initSupernodes^L$ achieves round complexity $\roundcost_n(\initSupernodes^L) = \bigO(n)$ and communication complexity $\bitcost_{L, n}(\initSupernodes^L) = \bigO((L + \kappa) \cdot n)$.
\end{lemma}
\begin{proof}
    \Cref{thm:initial-extractor} implies that the $N$ supernodes have size $\bigO(1)$ each, and that at most $\threshbad \cdot n$ of the supernodes are byzantine. By \Cref{lemma:ca-within-supernode}, $\supernodeCA^L$ provides each non-byzantine supernode $S_i$ with a $(\delta \cdot L)$-bit value $v_i$ that is in the convex hull of the inputs of the honest parties in $S_i$, and therefore that is valid. Hence, every non-byzantine supernode is a good supernode, so the proportion of bad supernodes $\threshbad$ is also respected. In addition, \Cref{lemma:ca-within-supernode} ensures that, even if an honest party is part of a byzantine supernode, it holds either no value or a $(\delta \cdot L)$-bit value for that supernode.

     $\initSupernodes^L$ runs a $\supernodeCA^L$ invocation for each of the $N = n$ supernodes in parallel, and each supernode has size $\bigO(1)$. Then, by \Cref{lemma:ca-within-supernode},  $\initSupernodes^L$ achieves round complexity $\roundcost_{\bigO(1)}(\supernodeCA^L) = \bigO(1)$ and  communication complexity $n \cdot \bitcost_{L, \bigO(1)}(\supernodeCA^L) = \bigO((L + \kappa) \cdot n)$, as claimed.
\end{proof}

\subsection{$\ca$ among Supernodes}\label{subsection:supernodes-ca}

Subprotocol $\supernodesCA^L$ assumes $n_C$ supernodes (in a committee), and aims to enable the honest parties within these $n_C$ supernodes to agree on a value $v$ within the convex hull of the good supernodes' inputs. We achieve this whenever $t_C < n_C / \omega$ of the $n_C$ supernodes are bad and the virtual parties within the $n_C$ supernodes form a $\ba$-friendly group.

Subprotocol $\supernodesCA^L$ proceeds similarly to $\supernodeCA^L$, with the exception that we obtain a common view over the supernodes' inputs using the supersending primitive described in  \Cref{lemma:supersending}. 
We proceed as follows: each supernode $S$ supersends its input value to the group of all $n_C$ supernodes. If we are in a $\ba$-friendly setting, the parties obtain a common view containing all of the good supernodes' values.
Then, if $\msgset$ is the multiset containing the values in this common view, each (party in each) supernode computes the safe area $S := \safe_{k}(\msgset)$, where $k := \abs{\msgset} - (n_C - t_C)$, and applies a deterministic decision to output a value from $S$. This already implies Agreement and Termination. Validity holds if the number of bad supernodes indeed satisfies $t_C < n_C / \omega$, as described by \Cref{lemma:safe-area}. 

We present the code of $\supernodesCA^L$ below, and \Cref{lemma:supernode-ca} describes its guarantees. 
\begin{dianabox}{$\supernodesCA^L$}
	\algoHead{Code for supernode $S$ with input $v_{\inputt}$}
	\begin{algorithmic}[1]
    \State $S$ sends $v_{\inputt}$ to  $B :=$ the group of all $n_C$ supernodes via $\supersend^L$.
    \State Let $\msgset$ be the multiset of $L$-bit values received, $t_C = \lceil n_C / \omega \rceil$ - 1, $k := \abs{\msgset} - (n_C - t_C)$. 
    \State If $k < 0$ or  $\safe_k(\msgset) = \emptyset$, output $\bot$.
    \State Otherwise, let $v$ be the value with the smallest binary encoding in $\safe_k(\msgset)$, breaking ties deterministically. If $v$ has bit-length up to $\delta \cdot L$, output $v$; otherwise output $\bot$.
	\end{algorithmic}
\end{dianabox}

\begin{lemma}\label{lemma:supernode-ca}
    Assume that each of the $n_C$ supernodes consists of up to $n_S$ parties.
    Then, all honest parties among the $n_C$ supernodes complete the execution of $\supernodesCA^L$ with either a $(\delta \cdot L)$-bit output or $\bot$, within $\roundcost_{n_C, n_S}(\supernodesCA^L) = n_C \cdot n_S$ rounds, with communication complexity $\bitcost_{L, n_C, n_S}(\supernodesCA^L) =
    \bigO(L \cdot n_C^2 \cdot n_S + \kappa \cdot n_C^3 \cdot  n_S^2 \cdot \log (n_C \cdot n_S))$.

    Moreover, if the group of virtual parties among the $n_C$ supernodes is $\ba$-friendly, the honest virtual parties in the $n_C$ supernodes output the same value. If, in addition, up to $t_C < n_C / \omega$ of the supernodes are bad and the good supernodes hold $L$-bit inputs, the output is a $(\delta \cdot L)$-bit value in the convex hull of the good supernodes' inputs.
\end{lemma}
\begin{proof}
    The communication complexity of the protocol is $n_C \cdot \bitcost_{L, n_S, n_C \cdot n_S}(\supersend^L)$, and the round complexity is $\roundcost_{n_S, n_C \cdot n_S}(\supersend^L)$. Then, the complexities written in the lemma statements follow from \Cref{lemma:supersending}. 

    Next, assume that the virtual parties in the $n_C$ supernodes form a $\ba$-friendly group. In this case, \Cref{lemma:supersending} ensures that these honest virtual parties obtain the same multiset $\msgset$, hence the same safe area $\safe_k(\msgset)$, and therefore the same output.
    If the good supernodes hold $L$-bit valid values, \Cref{lemma:supersending}  additionally ensures that the good supernodes' values are received. Moreover, if there are up to $t_C < n_C / \omega$ bad supernodes, 
    \Cref{lemma:safe-area} implies that $\safe_k(\msgset) \neq \emptyset$ and that $\safe_k(\msgset)$ is included in the convex hull of the values of the at least $n_C - t_C$ good supernodes. Then, the honest virtual parties in the $n_C$ supernodes obtain the same value $v$ in the convex hull of the good supernodes' inputs. As the good supernodes hold $L$-bit inputs from a convexity space with dilation factor $\delta$, $v$ has bit-length $\delta \cdot L$.
\end{proof}

\subsection{Reducing the Number of Supernodes} \label{subsection:reduce-supernodes}
We now present the subprotocol $\reduceSupernodes^L$: this is responsible for reducing the number of supernodes from $N$ to  $N' = \lfloor N/ \supernodefactor \rfloor$, while maintaining the fraction of good supernodes.
We first assign the $N$ supernodes to $N'$ committees, where each committee consists of $\bigO(\supernodefactor)$ supernodes. Each committee $C_i$, with $i = 1 \ldots N'$, is responsible for computing a valid input of the new supernode $S_i$ via $\supernodesCA^L$.

We say that a committee $C_i$ is \textbf{bad} if it does not meet the preconditions of $\supernodesCA^L$ described in \Cref{lemma:supernode-ca}: it has at least a $\frac{1}{\omega}$-fraction bad supernodes, or the fraction of byzantine (virtual) parties being part of it is at least $\frac{1}{\eta + \varepsilon/4}$. Otherwise, committee $C_i$ is \textbf{good}.
To maintain the fraction of good supernodes, we need to ensure that only few committees are bad. We achieve this by assigning the $N$ supernodes to the $N'$ committees with the theorem below (proven in \Cref{section:extractor-magic}), which relies on an extractor. The function $D = D(\omega, \varepsilon)$ will be defined in \Cref{section:extractor-magic}: as $\omega$ and $\varepsilon$ are constants, this is a constant.

\begin{restatable}{theorem}{NodesToComs}
    \label{cor:nodes-to-coms}
    Assume $N$ supernodes out of which at most $\threshbad \cdot N$ are bad. Then, there exists a deterministic $\poly(N, 1/\varepsilon, \omega)$-time algorithm that assigns the $N$ supernodes to $\lfloor N/\supernodefactor \rfloor$ committees such that:
    at most $\threshbad / 2 \cdot \lfloor N/\supernodefactor \rfloor$ committees are bad; 
     every supernode belongs to at most $D(\omega, \varepsilon)$ committees;
     the committees have equal size, consisting of between $D(\omega, \varepsilon) \cdot \lfloor \supernodefactor \rfloor$ and $2 D(\omega, \varepsilon) \cdot \supernodefactor$  supernodes.
\end{restatable}
For the new supernodes' assignment, we use a more general form of \Cref{thm:initial-extractor}, described below. We prove this theorem in \Cref{section:extractor-magic}, also using an extractor.  
\begin{restatable}{theorem}{NodesToSuperNodes}
    \label{cor:nodes-to-supernodes}
    Assume $n$ parties, out of which up to $\frac{n}{\omega' + \varepsilon}$ can be byzantine. Then, there is a deterministic $\poly(n, 1/\varepsilon, \omega)$-time algorithm  that assigns the $n$ parties to $N$ supernodes such that: at most $ \threshbad / 2 \cdot N$ supernodes are byzantine;
    every party belongs to at most $D(\omega, \varepsilon)$ many supernodes;
    each supernode contains $D(\omega, \varepsilon) \cdot \lfloor n / N \rfloor$ parties.
\end{restatable}
Finally, each committee $C_i$ supersends its value obtained in $\supernodesCA^L$ to the new supernode $S_i$. We present the code below, and \Cref{lemma:reduce-supernodes} describes its guarantees.

\begin{dianabox}{$\reduceSupernodes^L$}
	\begin{algorithmic}[1]
    \Statex \textbf{Assigning supernodes to committees:}
    \State Assign the $N$ supernodes into $N' = \lfloor N / \supernodefactor \rfloor$ committees $C_1, \ldots, C_{N'}$ using \Cref{cor:nodes-to-coms}.
    \Statex
    \Statex \textbf{Assigning a value to each committee:}
    \State In each committee $C_i$, in parallel: the supernodes in $C_i$ run $\supernodesCA^L$ and obtain $v_i$.
    \Statex
    \Statex \textbf{Assigning parties and values to the new supernodes:}
    \State Assign the $n$ parties into $N'$ supernodes $S_1, S_2, \ldots S_{N'}$ using \Cref{cor:nodes-to-supernodes}.
    \State For each committee $C_i$, in parallel: 
    \State \hspace{0.5cm} Parties in $C_i$ send $v_i$ to the new supernode $S_i$ via $\supersend^{\delta \cdot L}$.
	\end{algorithmic}
\end{dianabox}

\begin{restatable}{lemma}{ReduceSupernodes} \label{lemma:reduce-supernodes}
Assume that, at the start of $\reduceSupernodes^L$,  the parties are arranged into $N$ supernodes 
of equal size of up to $\bigO(1) \cdot n / N$ each such that: each party belongs to up to $\bigO(1)$ supernodes, at most $\threshbad \cdot N$ of the supernodes are bad, and each non-byzantine supernode holds an $L$-bit input.

Then, $\reduceSupernodes^L$ rearranges the parties into $N' = \lfloor N/\supernodefactor\rfloor$ supernodes of equal size of up to $\bigO(1) \cdot n / N'$ such that: each party belongs to at most $D = \bigO(1)$ supernodes, at most  $\threshbad \cdot N'$ are bad, and each non-byzantine supernode holds an $(\delta \cdot L)$-bit input.
 $\reduceSupernodes^L$ has round complexity $\roundcost_n(\reduceSupernodes^L) = \bigO(n \supernodefactor / N)$  and communication complexity $\bitcost_{L,n}(\reduceSupernodes^L) =  \bigO(L \cdot n \cdot \supernodefactor + \kappa\cdot n^2 / N \cdot \log n \cdot\supernodefactor^2)$.
\end{restatable}
\begin{proof}
We first establish the sizes of the initial supernodes, of the committees, and of the new supernodes.
We denote by $n_A = \bigO(n/N)$ the size of the largest initial supernode.
Denote by $n_C$ the maximum number of supernodes within a committee: these are obtained via \Cref{cor:nodes-to-coms} and, as $\omega$ and $\varepsilon$ are constants, we obtain that $n_C$ is between $D(\omega, \varepsilon) \cdot \lfloor \supernodefactor \rfloor$ and $2 \cdot D(\omega, \varepsilon) \cdot \supernodefactor$, hence $n_C \in \bigO(\supernodefactor)$. Finally, we denote by $n_B$ the largest size of a new supernode (obtained in line 3). As the new supernodes are obtained via \cref{cor:nodes-to-supernodes}, $n_B = \bigO(1)\cdot n\supernodefactor/N$. 

For the communication complexity of $\reduceSupernodes^L$, we need to recall that $\delta \in \bigO(1)$.
By \Cref{lemma:supernode-ca}, running $\supernodesCA^L$ within a single committee has communication complexity $\bitcost_{L, n_C, n_A}(\supernodesCA^L)$. 
Hence, step $2$ of $\reduceSupernodes^L$ has communication complexity $\bigO(L \cdot \supernodefactor \cdot n + \kappa \supernodefactor^2 \cdot (n^2/N) \cdot  \log (\supernodefactor n/N))$. 
In addition, step $2$ provides each committee with an $\bigO(\delta \cdot L)$-bit value.
This is followed in step 3 by $N'$ parallel executions of $\supersend^{\delta \cdot L}$ from groups of $n_C \cdot n_A$ virtual parties to groups of $n_B$ virtual parties. By \Cref{lemma:supersending}, this adds a total of $N' \cdot \bitcost_{\delta \cdot L, n_C \cdot n_A, n_B}(\supersend^{\delta \cdot L}) = \bigO(L \cdot n + \kappa \cdot \supernodefactor \cdot n^2 / N \cdot \log(\supernodefactor \cdot n / N))
$ bits.  Summing these up, we obtain a total of $\bigO(L \cdot n \cdot \supernodefactor + \kappa\cdot n^2 / N \cdot \log n \cdot\supernodefactor^2)$ bits of communication, as claimed in the lemma statement.

For the round complexity, $\reduceSupernodes^L$ runs $\supernodesCA^L$ within committees of up to $n_C$ supernodes of size up to $n_A$ each,
and afterwards $\supersend^{\delta \cdot L}$ from groups of up to $n_A \cdot n_C$ parties to groups of up to $n_B$ parties. Applying the round complexities stated in \Cref{lemma:supernode-ca} for $\supernodesCA^L$ and \Cref{lemma:supersending} for $\supersend^{\delta \cdot L}$ leads to a total of $\bigO(n \supernodefactor / N)$ rounds, as claimed in the lemma's statement. 

It remains to discuss the properties of the new supernodes. \Cref{cor:nodes-to-coms} ensures that at most $ \threshbad / 2 \cdot N'$ of the $N'$ committees are bad, while the remaining committees are good. Note that the requirements of \Cref{lemma:supernode-ca} are indeed satisfied by a good committee as at most $\frac{1}{\eta + \varepsilon/4}$ fraction of (virtual) parties in $C_i$ are byzantine, which ensures that $C_i$ is $\ba$-friendly, and less than a $\frac{1}{\omega}$ fraction of the supernodes in $C_i$ are bad. Hence, \Cref{lemma:supernode-ca}  ensures that honest parties within a good committee $C_i$ agree on a value $v_i$ that is within the honest convex hull of the good supernodes' inputs.

Next, \Cref{cor:nodes-to-supernodes} ensures that at most $\threshbad / 2 \cdot N'$ of the new supernodes are byzantine.
\Cref{lemma:supersending} ensures that, regardless of the nature of the committee $C_i$, if new supernode $S_i$ is non-byzantine, then all honest parties in $S_i$ agree on the same (possibly non-valid) $(\delta \cdot L)$-bit input $v_i$ for $S_i$. If, in addition, committee $C_i$ is good, $v_i$ is valid. Consequently, a new supernode can only be bad if it is byzantine, which there are at most $\threshbad / 2 \cdot N'$, or if it is confused because its value comes from a bad committee, which there are at most $\threshbad / 2 \cdot N'$. Hence there are at most $\threshbad \cdot N'$ bad supernodes, and the remaining supernodes are good and hold $(\delta \cdot L)$-bit valid inputs. In addition, the properties on the new supernodes' size and on the number of  new supernodes a party belongs to follow from \Cref{cor:nodes-to-supernodes}.
\end{proof}

\subsection{Putting It All Together}\label{subsection:ca-L-analysis}

We are now ready to discuss the guarantees of $\Pi_{\ca}^L$. The theorem below states these for a general convexity space with dilation factor $\delta$.
\begin{theorem} \label{theorem:ca-L-final}
    Consider an arbitrary constant $\varepsilon > 0$ and a convexity space $\mathcal{C}$ with constant Helly number $\omega \geq 2$ and constant dilation factor $\delta$. 
    $\Pi_{\ca}^L$ achieves $\ca$ on $\mathcal{C}$ whenever honest parties hold $L$-bit inputs from $\mathcal{C}$ and up to $t < n / (\omega'  + \varepsilon)$ of the $n$ parties involved are byzantine. $\Pi_{\ca}^L$ has round complexity  $\roundcost_{n}(\Pi_{\ca}^L) = \bigO(n)$ and communication complexity:
    $
        \bitcost_{L, n}(\Pi_{\ca}^L) =  \bigO(Ln \supernodefactor \cdot  \log_{\supernodefactor}(n) \cdot \delta^{\log_{\supernodefactor}(n)} + \kappa \cdot n^2 \log n \cdot \supernodefactor^2 ) .
    $
\end{theorem}
\begin{proof}
We first record the invariant maintained by the protocol after each reduction step. For $i \geq 1$, let $N_i$ denote the number of supernodes after the $i$-th execution of $\reduceSupernodes$: $N_i = \Theta(n/\supernodefactor^i)$, as the floors affect the bounds only by constant factors. One can show that the properties below hold at the end of every iteration $i \geq 1$ using induction: the base case follows from \Cref{lemma:init-supernodes} for the initial assignment followed by \Cref{lemma:reduce-supernodes} for the first iteration, and the induction step follows from \Cref{lemma:reduce-supernodes}.

\begin{enumerate}[label=(\Alph*)]
\item The parties are arranged into $N_i$ supernodes.
\item The fraction of bad supernodes is at most $\threshbad$.
\item Every non-byzantine supernode is assigned an $L_i$-bit input, where $L_i = \delta^{i+1} \cdot L$.
\item Each party belongs to at most $D = \bigO(1)$ supernodes.
\item Each supernode has size $\bigO(1) \cdot n / N_i$.
\item Iteration $i$ is completed within $\bigO(n\supernodefactor / N_i)$ rounds, with communication complexity $\bigO(L_i \cdot n \cdot \supernodefactor + \kappa \cdot n^2/N_i \cdot \log n \cdot \supernodefactor^2)$.
\end{enumerate}

We now bound the round complexity and communication complexity of $\Pi_{\ca}^L$. 
Before the loop, $\Pi_{\ca}^L$ runs $\initSupernodes^L$, which adds $\bigO(n)$ rounds and $\bigO((L+\kappa)n)$ bits of communication by \Cref{lemma:init-supernodes}. Then, let $i_{\max}$ be the number of executions of $\reduceSupernodes^L$: since the protocol exits the loop once the number of supernodes is below $\supernodefactor$, we have $i_{\max} \leq \lfloor \log_{\supernodefactor} n \rfloor$ by Property (A). By Property~(F), the number of rounds spent inside the loop is at most:
\begin{align*}
\sum_{i=1}^{i_{\max}} \bigO(n\supernodefactor/N_i)
\leq \bigO \left( \sum_{i =0}^{\lfloor \log_{\supernodefactor}(n) \rfloor - 1} \supernodefactor^{i+1} \right) 
\leq \bigO\left(  \sum_{i = 0}^{\lfloor \log_{\supernodefactor}(n) \rfloor} \supernodefactor^i \right)
= \bigO\left(  \frac{\supernodefactor^{\lfloor \log_{\supernodefactor}(n) \rfloor + 1} - 1}{\supernodefactor - 1}\right).
\end{align*}
Since $\supernodefactor > 1$, this is $\bigO(\supernodefactor^{\log_{\supernodefactor}(n)}) = \bigO(n)$.

Again using Property~(F), the communication complexity of the loop is bounded by
\begin{align*}
&\sum_{i=1}^{i_{\max}} 
\bigO\left(L_i \cdot n \cdot \supernodefactor
+ \kappa \cdot n^2/N_i \cdot \log n \cdot \supernodefactor^2\right) \\
&\qquad\leq
\bigO(Ln\supernodefactor)
\sum_{i=0}^{\lfloor \log_{\supernodefactor} n \rfloor}\delta^{i+1}
+
\bigO(\kappa n \log n \cdot \supernodefactor^2)
\sum_{i=0}^{\lfloor \log_{\supernodefactor} n \rfloor}\supernodefactor^i \\
&\qquad\leq
\bigO\left(Ln\supernodefactor \cdot \log_{\supernodefactor} n
\cdot \delta^{\log_{\supernodefactor} n}\right)
+
\bigO\left(\kappa n^2 \log n \cdot \supernodefactor^2\right).
\end{align*}

We still need to account for the final step, invoking $\supernodesCA^{L'}$ and $\supersend^{L'},$ where $L' := L_{i_{\max}} \leq  \delta^{\lfloor \log_{\supernodefactor} n \rfloor+1}\cdot L$. Let $N_\star$ be the number of supernodes when the loop terminates: by the loop exit condition, $N_\star \leq \supernodefactor$. The protocol runs $\supernodesCA^{L'}$ on these $N_\star$ supernodes. At this stage, by Property~(C), every non-byzantine supernode holds an $L'$-bit input, and by Property~(E), each remaining supernode has size $n_S = \bigO(n/N_\star)$.
Applying \Cref{lemma:supernode-ca} with $n_C = N_\star$ and $n_S = \bigO(n/N_\star)$, the final invocation of $\supernodesCA^{L'}$ adds $\bigO(n)$ rounds and  $\bigO\left( L  n\supernodefactor \cdot \delta^{\lfloor \log_{\supernodefactor} n \rfloor+1} + \kappa n^2\supernodefactor\log n \right)$ bits. Next, $\supersend^{\delta \cdot L'}$ is invoked with a group of up to $ \bigO(n)$ parties as senders with $(\delta \cdot L')$-bit inputs and a group of up to $n$ parties as recipients. By \Cref{lemma:supersending}, this adds $\bigO(n)$ rounds and $\bigO(L n \cdot \delta^{\lfloor \log_{\supernodefactor} n \rfloor+2} + \kappa n^2 \log n)$ bits.

Recalling that $\delta = \bigO(1)$ and summing up the costs of the initialization step, the loop, and the final, we obtain $\roundcost_n(\Pi_{\ca}^L)=\bigO(n)$ and $\bitcost_{L,n}(\Pi_{\ca}^L)=\bigO(Ln\supernodefactor \cdot \log_{\supernodefactor} n \cdot \delta^{\log_{\supernodefactor} n} + \kappa n^2\log n \cdot \supernodefactor^2)$, as claimed.

It remains to prove correctness. At the beginning of $\supernodesCA^{L'}$, fewer than a $\threshbad \leq 1/(\omega' + \varepsilon)$ fraction of the final $N_{\star}$ supernodes are bad by Property (B), and every good supernode holds an $L'$-bit input by Property (C). Since $\omega' \geq \omega$, the fraction of bad supernodes is strictly less than $1/\omega$, so the validity condition required by \Cref{lemma:supernode-ca} is satisfied.
We also need to verify that $\supernodesCA^{L'}$ runs in a $\ba$-friendly setting. The fraction of byzantine supernodes is at most $\threshbad$, and every non-byzantine supernode contains less than a $1/(\omega' + \varepsilon/2)$ fraction of byzantine parties. Hence the fraction of byzantine virtual parties in the final $\supernodesCA^{L'}$ invocation is at most $\threshbad + (1-\threshbad)\cdot 1/(\omega' + \varepsilon/2) \leq \threshbad + 1/(\omega' + \varepsilon/2) \leq 1/(\eta+\varepsilon/3) < 1/(\eta+\varepsilon/4)$, where the penultimate inequality follows from the definition of $\threshbad$. Therefore, the final invocation of $\supernodesCA^{L'}$ is $\ba$-friendly.
Hence, by \Cref{lemma:supernode-ca}, all honest parties appearing in the final supernodes obtain the same $(\delta \cdot L')$-bit value $v_{\outputt}$ in $\supernodesCA^{L'}$, and this value lies in the convex hull of the good supernodes' inputs. Since every good supernode holds a value in the convex hull of the honest parties' original inputs, $v_{\outputt}$ also lies in the convex hull of the original honest parties' inputs. 

Afterwards, the final $\supersend^{\delta \cdot L'}$ invocation runs in a setting where the group of receivers is the group of $n$ parties, and therefore this is $\ba$-friendly. Moreover, the group of senders is the group of (virtual) parties among the $N_{\star}$ final supernodes: we have already established that this is a $\ba$-friendly setting, and therefore this group has an honest (virtual) majority. In addition, all honest virtual parties in the group of senders join  $\supersend^{\delta \cdot L'}$  with the $(\delta \cdot L')$-bit value $v_{\outputt}$ as input. Then \Cref{lemma:supersending} ensures that all parties output $v_{\outputt}$ and, 
consequently, $\Pi_{\ca}^L$ achieves $\ca$.
\end{proof}

We may return to the theorems stated at the beginning of the section, namely \Cref{theo:abstract-for-1-fixed-L} and \Cref{theo:abstract-for-finite-fixed-L}: these follow from \Cref{theorem:ca-L-final} by instantiating $\Pi_{\ba}$ (depending on whether we are working in an authenticated setting) and the factor $\supernodefactor$. For  \Cref{theo:abstract-for-1-fixed-L}, we set $\supernodefactor := 2$, leading to  communication complexity $\bigO(Ln \cdot 2 \cdot  \log_{2}(n) 1^{\log_{2}(n)} + \kappa \cdot n^2 \log n \cdot 2^2) = \bigO(Ln \cdot  \log(n) + \kappa \cdot n^2 \log n)$. For \Cref{theo:abstract-for-finite-fixed-L}, setting $\supernodefactor := \log \log n$ leads to the communication complexity claimed in the theorem's statement:
\begin{align*}
    \bigO(Ln \cdot \log \log n \cdot  \log_{\log \log n}(n) \cdot  \delta^{\log_{\log \log n}(n)} + \kappa \cdot n^2 \log n \cdot \log^2 \log n  ) \\
     = \bigO(Ln^{1+o(1)} \cdot n^{\frac{\log \delta}{\log \log \log n}} + \kappa \cdot n^2 \log^{1+o(1)} n ) \\
     =  \bigO(Ln^{1+o(1)}  + \kappa \cdot n^2 \log^{1+o(1)} n ).
\end{align*}
\section{Protocol for Unknown Inputs' Bit-Lengths} \label{section:unknown-L}

In \Cref{section:general-protocol}, we have described our protocol $\Pi_{\ca}^{L}$ that achieves $\ca$ with near-optimal resilience  $t < n/ (\omega + \varepsilon)$, given a publicly known upper bound $L$ on the honest inputs' bit-lengths. We now remove this assumption for the lower resilience threshold $t < n/(\omega + 1 + \varepsilon)$.

\paragraph{On lower resilience.}
We first highlight the reasons behind the reduced resilience threshold: when $t \in [n/(\omega + 1), n / \omega)$, the $t$ byzantine parties may follow the protocol correctly with inputs of their own choice, impacting safe area computations. 
To be more specific, the safe area may be reduced to a single point whose length depends on the byzantine inputs' length. For example, in $\mathbb{R}^d$ (with $d \geq 2$ and straight-line convexity), one can find a set of $d+2$ values from
$\mathbb{N}^d$ where $d + 1$ of them have bit-length $\bigO(d)$, and the remaining value $v_b$ has larger bit-length $L_b$. The $\bigO(d)$-bit values are assigned as inputs to the honest parties, while the byzantine parties follow the protocol correctly, but with input $v_b$. In this case, the resulting safe area may consist of a single point of size $\Theta(L_b)$.
As a consequence, protocols where honest parties send intermediate values computed via safe area have unbounded bit complexity if $t \geq n/(\omega + 1)$ and $L$ is unknown.

\paragraph{Starting point.}
We present a mechanism that provides a reasonable estimation $\widetilde{L}$ on $L$. This generalizes the mechanism of \cite{PODC:GhiLiuWat25} from $\mathbb{N}$ to arbitrary input spaces. This mechanism of \cite{PODC:GhiLiuWat25} first verifies whether the inputs are short, i.e., up to $n^2$ bits, using a bit $\ba$ protocol, and uses different approaches depending on this check.
If the inputs are short, the parties obtain a potentially tighter estimation $\widetilde{L}$ by comparing their inputs' bit-lengths with powers of two, and otherwise they rely on a high-communication $\ca$ protocol. A key insight for $\naturalnumbers$ is that, if a party $\party$ holds an $L$-bit input $v_{\inputt} \in \naturalnumbers$ and the estimation agreed upon is $\widetilde{L} < L$, then $\party$ may replace its input with $2^{\widetilde{L}} - 1$: this is valid as there is an honest input $v_{\inputt}' \leq 2^{\widetilde{L}} - 1 < v_{\inputt}$. In an arbitrary convexity space, this property is lost.

\paragraph{Finding a \emph{large} estimation.}
To work around this, we first strengthen the mechanism to provide the honest parties with an upper bound $\widetilde{L}$ towards the \emph{highest} honest input bit-length. As discussed in \cite{IEEE:MelWat18}, the up to $t$ byzantine parties prevent us from identifying the concrete highest honest bit-length. However, we can agree on a value $\widetilde{L}$ between the $(n - 2t)$-th honest input bit-length and the highest honest input bit-length.  This only leaves up to $t$ honest parties with inputs of more than $\widetilde{L}$ bits, hence unable to join $\Pi_{\ca}^{\widetilde{L}}$ with their inputs.
Hence, we proceed as follows: the parties first exchange their input lengths, and each party denotes the $(n-t)$-th lowest value received by $L'$. Then, similarly to \cite{PODC:GhiLiuWat25}, we use different approaches depending on whether $L' \leq n^2$ or not: this is decided with a $\ba$ protocol $\Pi_{\ba}$.
We remark that, as $\omega \geq 2$, our lower resilience threshold $t < n/(\omega + 1+\varepsilon) < n/3$ enables us to rely on the unauthenticated protocol $\Pi_{\ba}$ described by \Cref{thm:magic-ba-no-pki}.
Then, the parties obtain the upper bound $\widetilde{L}$ by running a subprotocol $\exponentialSearch$ if they agreed on $L' \leq n^2$, and by running a high-communication $\ca$ protocol on $\mathbb{N}$, denoted by $\largebitcostca$, otherwise. We describe these in \Cref{subsection:naive-approximation} and \Cref{subsection:high-communication-CA} respectively.
Either way, the parties obtain a value $\widetilde{L}$ between the $(n - 2t)$-th lowest honest input bit-length and at most $2 \cdot L$, where $L$ is the highest honest input bit-length. The up to $t$ honest parties holding inputs of more than $\widetilde{L}$ bits replace their inputs with (some default $\widetilde{L}$-bit value) $\bot$. 

\paragraph{Providing more honest parties with inputs.}
At this point, we have up to $t < n / (1 + \omega + \varepsilon)$ byzantine parties, plus up to $t$ honest parties with input $\bot$ -- this exceeds the number of corruptions that $\Pi_{\ca}^{\widetilde{L}}$ can tolerate. We decrease the number of honest parties holding $\bot$ by relying on committees responsible for providing each party with a valid $(\delta \cdot \widetilde{L})$-bit input, where $\delta$ is the dilation factor of the input space. For the committee assignment, we use the theorem below, which we prove in \Cref{section:extractor-magic}. The fraction $\mu$ is chosen so that, even if a $\mu$-fraction of the honest parties join $\Pi_{\ca}^{\delta \cdot \widetilde{L}}$  with non-valid inputs, $\Pi_{\ca}^{\delta \cdot \widetilde{L}}$ still succeeds. We add that $D := D(\omega, \varepsilon) \in \bigO(1)$, and this will be defined in \Cref{section:extractor-magic}.

\begin{restatable}{theorem}{UnknownLCommitteeAssignment}\label{lemma:unkown-committee-assigment}
    Assume $n \geq 1$, $\varepsilon > 0$, $t < n/ (\omega + 1 + \varepsilon)$.
    Let $\mu := \frac{1}{\omega+1+\varepsilon/2}-\frac{1}{\omega+1+\varepsilon}$.
    Then, there is a deterministic $\poly(n, 1/\varepsilon, \omega)$-time algorithm that assigns $n$ parties into $n$ committees such that each committee has size $D(\omega, \varepsilon)$. Moreover, if at most $t$ parties are byzantine and at most $t$ are honest with input $\bot$, then at most $\mu \cdot n$ committees have more than a $1/(\omega + 1)$-fraction of byzantine parties or more than a $1/(\omega+1)$-fraction of honest parties with input $\bot$.
\end{restatable} 
\Cref{lemma:unkown-committee-assigment} provides us with $n$ committees $C_1, \ldots, C_n$, where committee $C_i$ is responsible for obtaining a valid value for party $\party_i$. To obtain these values, the parties inside each committee $C_i$ run a subprotocol $\botCommitteeCA$. This will be a variant of $\supernodesCA^L$ that allows honest parties to join with input $\bot$, which we describe in \Cref{subsection:ca-bot}.

After this step, we are working in a setting where up to $t' < n / (\max(\omega, 3) + \varepsilon / 2)$ parties are either byzantine, or honest but hold a non-valid input. This allows us to run (the unauthenticated protocol) $\Pi_{\ca}^{\delta \cdot \widetilde{L}}$ for corruption threshold $t'$. Then, each party joins $\Pi_{\ca}^{\delta \cdot \widetilde{L}}$ either with the input obtained from its committee (if any) or $\bot$, and $\Pi_{\ca}^{\delta \cdot \widetilde{L}}$  provides the parties with their final outputs.

\begin{dianabox}{$\Pi_{\ca}$}
	\algoHead{Code for party $\party$ with input $v_{\inputt}$}
	\begin{algorithmic}[1]
    \Statex \textbf{Obtaining an upper bound $\widetilde{L}$ on the inputs' length:}
    \State Send  $|v_{\inputt}|_b$ to all parties, and let $L'$ be the $(n-t)$-th lowest value among the ones received.
    \State Set $b_{\inputt} := 1$ if $L' > n^2$ and $b_{\inputt} := 0$ otherwise.  Join $\Pi_{\ba}$ with input $b_{\inputt}$, obtain output $b_{\outputt}$.
    \State If $b_{\outputt} = 1$,  set $\widetilde{L} := \largebitcostca(\lceil L' /n^2\rceil) \cdot n^2$.
    \State Otherwise, if $b_{\outputt} = 0$, set $\widetilde{L} := \exponentialSearch(\min(L', n^2))$.
    \Statex
    \Statex \textbf{Providing each party with an $(\delta \cdot \widetilde{L})$-bit valid input:}
    \State If $|v_{\inputt}|_b > \widetilde{L}$, set $v_{\inputt} := \bot$.
    \State Find $n$ committees $C_1, C_2, \ldots, C_n$ of size $\bigO(1)$ using \Cref{lemma:unkown-committee-assigment}.
    \State For every committee $C_i$, in parallel:
    \State \hspace{0.5cm} Parties in committee $C_i$ run $\botCommitteeCA$ and obtain output $v_i$.
    \State \hspace{0.5cm} Every party in committee $C_i$ sends $v_i$ to party $\party_i$.
    \State Every party $\party_i$ sets $v_{\inputt}' := $ the value $v_i$ received most often, or $\bot$ if no value was received.
    \Statex 
    \Statex \textbf{Achieving $\ca$:}
    \State Every party joins $\Pi_{\ca}^{\delta \cdot \widetilde{L}}$ (for corruption threshold $t' < n / (\max(\omega, 3) + \varepsilon / 2)$) with input $v_{\inputt}'$. When obtaining output $v_{\outputt}$, it outputs $v_{\outputt}$.
	\end{algorithmic}
\end{dianabox}

This leads us to the theorems stated below. We defer the formal guarantees of $\Pi_{\ca}$, together with the proofs of \Cref{theo:abstract-for-1-unknown-L} and \Cref{theo:abstract-for-finite-unknown-L}, to \Cref{subsection:unknown-L-all-done}.

\begin{theorem}\label{theo:abstract-for-1-unknown-L}
Consider an arbitrary constant $\varepsilon > 0$ and a convexity space $\mathcal{C}$ with constant Helly number $\omega \geq 2$ and dilation factor $\delta = 1$. 
Then, there is an unauthenticated protocol $\Pi_{\ca}$ achieving $\ca$ on $\mathcal{C}$ whenever up to $t < n / (\omega  + \varepsilon + 1)$ of the $n$ parties involved are byzantine. $\Pi_{\ca}$ has communication complexity $\bitcost_{L, n}(\Pi_{\ca}) =  \bigO(Ln \log n + \kappa \cdot n^2 \cdot \log n)$ and round complexity $\roundcost_{n}(\Pi_\ca) = \bigO(n)$.
\end{theorem}

\begin{theorem}\label{theo:abstract-for-finite-unknown-L}
Consider an arbitrary constant $\varepsilon > 0$ and a convexity space $\mathcal{C}$ with constant Helly number $\omega \geq 2$ and finite dilation factor $\delta$. Then, 
there is an unauthenticated protocol $\Pi_{\ca}$ achieving $\ca$ on $\mathcal{C}$ whenever up to $t < n / (\omega  + \varepsilon + 1)$ of the $n$ parties involved are byzantine. $\Pi_{\ca}$ has communication complexity $\bitcost_{L, n}(\Pi_{\ca}) =  \bigO(Ln^{1 + o(1)} + \kappa \cdot n^2 \cdot \log^{1 + o(1)} n)$ and round complexity $\roundcost_{n}(\Pi_\ca) = \bigO(n)$.
\end{theorem}

These results enable us to discuss \Cref{theo:main-R-d}, stated in our paper's introduction. $\mathbb{R}^d$ has Helly number $d+1$, and, by \Cref{theo:dilation-factor-euclidian}, it has dilation factor $\delta := 4d^2$. Then, \Cref{theo:main-R-d} follows from \Cref{theo:abstract-for-finite-fixed-L} in the case where an upper bound $L$ on the honest inputs' bit-length is known, and from \Cref{theo:abstract-for-finite-unknown-L} when no such upper bound $L$ is known.
\section{Extractors} \label{section:extractor-magic}
It remains to present the proofs behind the extractor-based theorems we have used in  \Cref{section:general-protocol} and \Cref{section:unknown-L} to assign parties to supernodes, supernodes to committees, and parties to committees.
At core, our assignment theorems address the following task: starting from a set of, say, $n$ parties, assign them to committees so that only an $\varepsilon$-fraction of the committees are bad, where a committee is bad if more than a $1/\omega'$-fraction of its parties are byzantine. Intuitively, since the original fraction of byzantine parties is at most $1/(\omega' + \varepsilon)$, choosing sufficiently large committees uniformly at random would achieve this guarantee with high probability. However, such an approach would require some form of public randomness, while our protocols are deterministic. Moreover, relying on randomness in this way is not immediately secure against an adaptive adversary. We therefore focus on a deterministic substitute for random sampling instead. We obtain this by assigning parties to committees pseudo-uniformly using an object called an \emph{extractor} \cite[Chapter~6]{pseudorandomness}.

\paragraph{Extractors.} Historically, extractors were used to extract random bits from a ``weak'' random source, and their main application was in error-correcting codes \cite{ShaExtractor}. We give a definition of an extractor in terms of a bipartite graph below.

\begin{definition}\label{definition:extractor-weak}
    A bipartite multigraph $G = (L \cup R, E)$ with $\abs{L} = n$, $\abs{R} = m$ is left-regular of degree $d$ on the left. This graph is a $(k, \varepsilon)$-extractor if, for every subset $S \subseteq L$ of size at least $k$ and for every subset $T \subseteq R$, it holds that: $\left| \frac{\abs{E(S, T)}}{d \cdot \abs{S}} - \frac{\abs{T}}{\abs{R}} \right| < \varepsilon$. That is, the probability that an edge leaving $S$ goes to $T$ is $\varepsilon$-close to the probability that a random vertex on the right is in $T$.
\end{definition}

We note that this definition is slightly different compared to the regular one. First, the regular definition uses the concept of min-entropy while, for simplicity, we only consider subsets of size at least $k$, which leads to a weaker notion. Moreover, when building such a graph from an extractor, the values $n$, $m$ and $d$ are usually required to be powers of $2$, which is not the case here: the extractor construction we use works for arbitrary $n$. 

Random bipartite graphs are good extractors with high probability.  However, our focus is on deterministic protocols. In addition, generating a random bipartite graph and checking that it is an extractor would take exponential time. We want to avoid exponential-time approaches: for practical considerations, and as our $\ca$ protocols make use of a hash function (which requires the adversary to be computationally-bounded). We thus turn our focus towards \emph{explicit} extractors. A graph $G$ is called \textbf{explicit} if there exists a deterministic algorithm that builds it within time polynomial in $G$'s size.

The last point of interest is that the existing literature for extractors mainly focuses on the case where $k \ll n$. For our use case, we want $k \geq t$ and hence our interest lies in what are called \textit{high min-entropy} extractors, which have not been explored as much.

We therefore give an explicit extractor with near-optimal properties for our use case in \Cref{thm:finding-the-extractor}. 
We use the proof of \cite[Theorem 6.22]{pseudorandomness} as a starting point, which gives a similar result but, in contrast to \Cref{thm:finding-the-extractor}, does not require $R$ to be regular. We note that making $R$ regular is, in fact, not strictly required by our approach, but this property makes our proofs easier.
\begin{theorem} \label{thm:finding-the-extractor}
    Let $\varepsilon > 0$ and $0 < \alpha < 1$. Then, there is $D(\varepsilon, \alpha) = \poly(1/\varepsilon, 1/\alpha) > 0$ such that, for any $n$ and $k \geq \max(\lfloor \alpha \cdot n \rfloor, 1)$, there is an explicit $(k, \varepsilon)$-extractor, which has left and right sides of size $n$ each and is biregular with degree $D(\varepsilon, \alpha)$.
\end{theorem}
\begin{proof}
    We use a connection between \emph{spectral expansion} and extractors, which is shown in the first part of the proof of \cite[Theorem~6.22]{pseudorandomness}. Let $G=(V,E)$ be a $D$-regular graph on $n$ vertices, and let $\lambda$ be its \emph{spectral expansion}: this is the largest absolute value of any eigenvalue of the normalized adjacency matrix (where the entries are divided by $D$) other than the trivial eigenvalue $1$.
    Let $H=(L \cup R,E_H)$ be the incidence graph of $G$: $L = V \times \{0\}$, $R = V \times \{1\}$ and for every edge $(u,v) \in E$ we add edges $((u,0), (v,1))$ and $((v,0),(u,1))$ to $E_H$. Then, $H$ is a $(k,\varepsilon)$-extractor whenever $\lambda \leq \varepsilon \cdot \sqrt{k/n}$.
    In addition, as $H$ is the incidence graph of a $D$-regular graph with $n$ vertices, both sides of $H$ have size $n$ and are $D$-regular.
    As $k \geq \max(\lfloor \alpha \cdot n \rfloor, 1) \geq \alpha \cdot n / 2$, it is enough to ensure $\lambda \leq \varepsilon \cdot \sqrt{\alpha/2}$.   
    
    Then, it remains to build the $D$-regular graph $G$ on $n$ vertices, with $D = D(\varepsilon, \alpha)$ and spectral expansion  $\lambda \leq \varepsilon \cdot \sqrt{\alpha/2}$, within polynomial time. A $D$-regular Ramanujan graph will suffice:
    a $D$-regular graph is called \emph{Ramanujan} if all of its nontrivial eigenvalues are as small as possible up to constants; i.e., its normalized second eigenvalue is at most $2\sqrt{D-1}/D$. Hence, such graphs have spectral expansion $\lambda \leq 2 \sqrt{D - 1}/D \leq 2/ \sqrt{D}$.
    By using \cite{FOCS:Cohen16b}, one can build a $D$-regular Ramanujan (multi)graph for arbitrary $n$ and $D$ in polynomial time.
    Therefore, by taking $D = \lceil 8/(\varepsilon^2 \alpha)\rceil = \poly(1/\varepsilon,1/\alpha)$
    we get a low enough spectral expansion $\lambda \leq \varepsilon \cdot \sqrt{\alpha/2}$, and thus an appropriate extractor $H$.
\end{proof}

\paragraph{Core theorem.}
Using \cref{thm:finding-the-extractor}, we may present a generic committee assignment scheme, which will be the core of the theorems used in our $\ca$ protocols. This is described by the theorem below: this assignment scheme is represented as a bipartite graph in which the left vertices correspond to parties, and the right vertices to committees. The expander structure of the graph yields several useful properties. In particular, each party is assigned to only constantly many committees, yet no matter which parties the adversary corrupts after seeing the assignment, it can compromise only a small fraction of committees. 

\begin{theorem} \label{thm:big-one}
    Consider $0 < \alpha < \beta < 1, \mu > 0$. Then, there exists $D = poly(1/\alpha, 1/(\beta - \alpha), 1/\mu) > 0$ such that, for any $1 \leq m \leq n$, the following holds: there is an explicit bipartite multigraph $G = (L \cup R, E)$ such that $|L| = n$, $|R| = m$, vertices in $L$ have degree at most $D$, vertices in $R$ are $(D\lfloor n/m \rfloor)$-regular. Moreover, for any $S \subset L$ such that $|S| \leq \alpha |L|$, the set $Q = \{ v \in R, \text{at least a $\beta$-fraction of $v$'s neighbors comes from $S$} \}$ has size $|Q| < \mu |R|$.
\end{theorem}
\begin{proof}
    Let $\varepsilon = \varepsilon(\alpha,\beta,\mu)$, which we will define later.
    Let $1 \leq m \leq n$, and $s=\lfloor n/m\rfloor \geq 1$.

    We first note that, if $\lfloor \alpha n\rfloor=0$, then all appropriate sets $S$ are empty and the property is immediate.
    The remainder of the proof is therefore concerned with the case $\lfloor \alpha n\rfloor\geq 1$, and we set $k:=\lfloor \alpha n\rfloor$.

    By \cref{thm:finding-the-extractor}, there is a $D = \poly(1/\varepsilon, 1/\alpha)$ such that, for this value of $n$ and $k = \lfloor \alpha n\rfloor$, we can find a $(k, \varepsilon)$-extractor $G'=(L' \cup R', E')$.
    
    We construct the multigraph $G = (L \cup R, E)$ from the $(k, \varepsilon)$-extractor $G' = (L' \cup R', E')$: we keep the left part intact, i.e. $L = L'$, and $R$ groups the vertices of $R'$ in disjoint groups of size $s = \lfloor n/m \rfloor$ each so that $R$ has size exactly $m$. The last $[n \mod m]$ vertices of $R'$ are discarded. 
    For every edge $(u,v')\in E'$ with $u\in L'$ and $v'\in R'$ lying in one of the $m$ groups, we add an edge from $u$ to the corresponding group vertex in $R$. Edges incident to the unused vertices of $R'$ are discarded.
    This way, we obtain $G = (L \cup R, E)$ with left side vertices of degree at most $D$, and right side of degree $D \cdot s = D \cdot \lfloor n/m \rfloor$.

    We still need to define $D$, or more specifically, an $\varepsilon$, so that: for any  $S \subset L$ such that $|S| \leq \alpha |L|$, the set $Q = \{ v \in R, \text{at least a $\beta$-fraction of $v$'s neighbors comes from $S$} \}$ has size $|Q| < \mu |R|$. 
    We choose $\varepsilon$ small enough to ensure that $Q$ cannot contain a $\mu$-fraction of the vertices in $R$. In the following, we show that  $\varepsilon = \mu \cdot (\beta/\alpha - 1)/2$ suffices by providing a lower bound of $\varepsilon$ in terms of $|Q|/|R|$.

    We start from the properties of the $(k, \varepsilon)$-extractor graph $G'$.
    First, we remark that adding elements to $S$ can only increase the size of $Q$. As such, without loss of generality, we can assume that $|S| = \lfloor \alpha \cdot |L| \rfloor$.
    We denote by $T$ the vertices in $G'$ contained by the group vertices forming $Q$ in $G$. As $G'$ is a $(k, \varepsilon)$-extractor graph, and $|S| = \lfloor \alpha \cdot |L| \rfloor \leq \lfloor \alpha \cdot n \rfloor = k$, we obtain that 
    $\left| \frac{|E'(S, T)|}{D \cdot |S|} - \frac{|T|}{|R'|} \right| < \varepsilon$, hence
    $\varepsilon >  \frac{|E'(S, T)|}{D \cdot |S|} - \frac{|T|}{|R'|}$.

    The definition of $T$ implies that $|E'(S,T)| = |E(S,Q)|$. Then, as $Q$ is defined as a set of vertices $v$ in $R$ such that at least a $\beta$-fraction of  $v$'s neighbors come from $S$, $|E(S,Q)| \geq \beta \cdot |E(L, Q)|$. Afterwards, using the definition of $T$ (and as $G$ is a multigraph), we get that $|E(L, Q)| = |E'(L, T)|$. Finally, we recall that vertices in $T$ have degree $D$ in $G'$, therefore $|E'(L, T)| = D \cdot \abs{T}$. Consequently, we have obtained that $|E'(S,T)| \geq \beta \cdot D \cdot |T|$. Using this in our previous inequality on $\varepsilon$, we obtain that $\varepsilon > \frac{\beta \cdot |T|}{|S|}- \frac{|T|}{|R'|}$.
    Next, as $|S| \leq \alpha \cdot |L|$ and $|L| = |R'| = n$, we obtain the following.  We note that $\beta / \alpha - 1 > 0$ since $\beta > \alpha$. 
     \begin{align*}
     \varepsilon > \frac{\beta \cdot |T|}{|S|}- \frac{|T|}{|R'|} \geq \frac{\beta \cdot |T|}{\alpha \cdot |L|}- \frac{|T|}{|R'|} \geq \frac{\beta \cdot |T|}{\alpha \cdot  n}-\frac{|T|}{n} = (\beta / \alpha - 1) \cdot \frac{|T|}{n}.
     \end{align*}

    We next show that $|T|/n \geq |Q|/(2|R|)$. Since $T$ is the union of the $s$ vertices in $R'$ corresponding to each vertex of $Q$, we have $|T|=s|Q|$. Moreover, since $s=\lfloor n/m\rfloor$ and $m\leq n$, we have $s\geq n/(2m)$. As $m=|R|$, it follows that $s/n\geq 1/(2|R|)$, and therefore $|T| / n =  \abs{Q} \cdot s/n \geq \abs{Q} / \left(2 \cdot \abs{R} \right)$. Consequently, we have obtained that:

    \begin{align*}
    \varepsilon > (\beta / \alpha - 1) \cdot \frac{|T|}{n}
         \geq \frac{\beta / \alpha - 1}{2} \cdot \frac{|Q|}{\abs{R}}.
    \end{align*}

    As such, we set $\varepsilon = \mu \cdot (\beta/\alpha - 1)/2$. Since $\beta>\alpha$, the lower bound above implies $|Q|/|R| < \mu$, and hence $|Q|<\mu |R|$. 
    Finally, our choice of $\varepsilon$ also implies that the degree $D$ is  $D=\poly(1/\alpha,1/(\beta-\alpha),1/\mu)$, as claimed.
    Thus $G$ satisfies the required properties.
\end{proof}

\paragraph{Specific assignments.} \Cref{thm:big-one} enables us to return to the theorems used in \Cref{section:general-protocol} and \Cref{section:unknown-L} and present their proofs.
We start with \Cref{cor:nodes-to-coms}, which we have used in $\reduceSupernodes^L$ to assign supernodes to committees responsible for computing inputs. We recall that a committee is \emph{bad} if it has more than $\frac{1}{\omega}$-fraction bad supernodes, or it is not $\ba$-friendly, i.e., the fraction of byzantine parties being part of it is more than $\frac{1}{\eta + \varepsilon/4}$.

\NodesToComs*
\begin{proof}
    We apply  \Cref{thm:big-one} using the following parameters: $n = N$, $m = \lfloor N/\supernodefactor \rfloor$, $\alpha = \threshbad$, $\beta = \min \left( \frac{1}{\omega}, \frac{1}{\eta + \varepsilon/4 } - \frac{1}{\omega' + \varepsilon/2} \right)$, $\mu = \threshbad / 2$.
    As $\alpha = \threshbad =  \min \left( \frac{1}{\omega' + \varepsilon}, \frac{1}{\eta + \varepsilon/3 } - \frac{1}{\omega' + \varepsilon/2} \right)$ we note that $\alpha < \beta$. Within $\poly(N, 1/\varepsilon, \omega)$ time, this gives us a graph $G$ assigning the $N$ supernodes into $\lfloor N / \supernodefactor \rfloor$ committees, such that each supernode is assigned to up to $D$ committees, and each committee consists of $D \cdot \lfloor n / m \rfloor$ supernodes. Note that $\lfloor n / m \rfloor = \left \lfloor n  / \lfloor n / \supernodefactor \rfloor \right\rfloor$, which gives us the bounds $\lfloor \supernodefactor \rfloor \leq \lfloor n / m \rfloor \leq 2 \cdot \supernodefactor$, leading to the committee size claimed in the lemma statement.
    Moreover, the number of committees containing at least a $\beta$-proportion of bad supernodes is bounded by $\mu \cdot \lfloor N / \supernodefactor \rfloor = \threshbad / 2 \cdot \lfloor N / \supernodefactor \rfloor$, as required by the theorem statement.  

    It remains to consider a committee $C$ which has less than a $\beta$-proportion of bad supernodes and to prove that it is a good committee. We know that its good supernodes have less than a $\frac{1}{\omega' + \varepsilon/2}$-fraction of byzantine parties. In the worst case, bad supernodes can consist entirely of byzantine parties, but our choice of $\beta$ ensures that $C$ contains at most a $\left(\frac{1}{\eta + \varepsilon/4 } - \frac{1}{\omega' + \varepsilon/2} \right)$-fraction of bad supernodes. Then, as each supernode has the same number of parties, the total proportion of byzantine parties in $C$ is at most:
\begin{align*}
    \left(\frac{1}{\eta + \varepsilon/4 } - \frac{1}{\omega' + \varepsilon/2} \right) \cdot 1 &+ \left( 1 - \left(\frac{1}{\eta + \varepsilon/4 } - \frac{1}{\omega' + \varepsilon/2} \right) \right) \cdot \frac{1}{\omega' + \varepsilon/2} \\
    & < \frac{1}{\eta + \varepsilon/4 } - \frac{1}{\omega' + \varepsilon/2}  + \frac{1}{\omega' + \varepsilon/2}  \\
    & = \frac{1}{\eta + \varepsilon/4}.
\end{align*}
Thus $C$ is $\ba$-friendly. Moreover, the proportion of bad supernodes being at most $\beta \leq 1/\omega$ ensures that $C$ is indeed good.
\end{proof}

Next, we prove \Cref{cor:nodes-to-supernodes}, which we have used in $\reduceSupernodes^L$ to assign parties to supernodes.
\NodesToSuperNodes*
\begin{proof}
We apply \Cref{thm:big-one} with parameters $n$, $m = N$, $\alpha = \frac{1}{\omega' + \varepsilon}$, $\beta = \frac{1}{\omega' + \varepsilon / 2}$ and $\mu = \threshbad / 2 > 0$.
Since $\varepsilon > 0$, we have $\alpha<\beta$.

Then, \Cref{thm:big-one} gives a deterministic $\poly(n, 1/\varepsilon, \omega)$-time algorithm constructing an assignment graph $G=(L \cup R, E)$ where $|L| = n$ and $|R| = N$: the left vertices represent parties, and the right vertices represent supernodes.
By the properties of $G$, every party belongs to at most $D(\omega,\varepsilon)$ supernodes, and every supernode contains $D(\omega,\varepsilon) \cdot \left\lfloor \frac{n}{N}\right\rfloor$ parties, as required by the theorem statement.

It remains to show that the assignment respects the fraction $\threshbad / 2$ of byzantine supernodes.
Let $S\subseteq L$ be the set of byzantine parties: $|S| \leq n/(\omega'+\varepsilon) = \alpha  \cdot |L|$. Let $Q \subseteq R$ be the set of supernodes containing at least a $\beta = 1/(\omega'+\varepsilon/2)$ fraction of byzantine parties. By definition, these are exactly the byzantine supernodes. Applying \Cref{thm:big-one} to $S$, we get $|Q| < \mu \cdot |R|=\threshbad / 2 \cdot N$. Therefore, at most $\threshbad / 2 \cdot N$ supernodes are byzantine, as required.
\end{proof}

Applying \Cref{cor:nodes-to-supernodes} with $N = n$ results in a stronger version of \Cref{thm:initial-extractor}, which we have used in $\initSupernodes^L$ to assign parties to supernodes.
\InitialExtract*

Finally, we recall \Cref{lemma:unkown-committee-assigment}, which we have used in \Cref{section:unknown-L} to build committees responsible for computing valid values, so that we reduce the number of honest parties holding value $\bot$.

\UnknownLCommitteeAssignment*
\begin{proof}
We use \Cref{thm:big-one} with parameters $\alpha = \frac{1}{\omega+1+\varepsilon}$,
$\beta = \frac{1}{\omega+1}$, $m = n$, and $\mu' = \mu / 2$. 
\Cref{thm:big-one} gives a deterministic  $\poly(n, 1/\varepsilon, \omega)$-time algorithm constructing an assignment graph $G=(L \cup R, E)$ where $|L| = |R| = n$:
the left vertices represent parties, and the right vertices represent committees.
Each party is assigned to up to $D = D(\omega, \varepsilon)$ committees, and  each committee consists of 
$D \cdot \left\lfloor \frac{n}{m}\right\rfloor = D$ parties, as claimed in the theorem's statement.

\Cref{thm:big-one} ensures that, for any set of parties $S$ of size at most $\alpha \cdot \abs{L}$, fewer than $\mu' \cdot n$ of the $n$ committees have at least a $\beta$ fraction of parties from $S$. Then, as $t < n/(\omega + 1 +\varepsilon) = \alpha \cdot n$ and $\beta=1/(\omega+1)$, we get that fewer than $\mu' \cdot n$ committees contain at least a $1/(\omega+1)$ fraction of byzantine parties, and fewer than $\mu' \cdot n$ committees contain at least a $1/(\omega+1)$ fraction of honest parties with input $\bot$. Consequently, fewer than $2 \cdot \mu' \cdot n = \mu \cdot n$ of the committees are in either of the two cases.
\end{proof}
\section{Conclusions}

In this paper, we have explored whether $\ca$ on arbitrary convexity spaces can be achieved deterministically using only $o(L \cdot n^2)$ bits of communication, where $L$ is an upper bound on the honest inputs' bit-length. For spaces with dilation factor $1$, we have obtained communication complexity $\bigO(L \cdot n \log n + n^2 \kappa \log n)$, where $\kappa$ denotes a security parameter and $L$ denotes an upper bound on the honest inputs' bit-lengths. For spaces with finite dilation factor, we have obtained communication complexity $\bigO(L \cdot n^{1 + o(1)} + n^2 \kappa \log^{1 + o(1)} n)$.

As pointed out in \cite{PODC:GhiLiuWat25}, the standard communication-saving techniques for $\ba$ do not directly extend to $\ca$:
they crucially rely on \emph{compressing} parties' inputs into short encodings, whereas in $\ca$ we must preserve the \emph{geometric/convex structure} of the input space.
Our protocols therefore compress the \emph{system} rather than the \emph{values}. To do so, we have introduced \emph{extractor-based deterministic committees}. This tool has enabled us to obtain deterministic $\ca$ protocols with near-optimal communication complexity for sufficiently large $L$, asymptotically optimal round complexity, and near-optimal resilience.
Our results also have implications beyond $\ca$: in particular, our $\ca$ protocols immediately provide communication-efficient corollaries for parallel instances of $\ba$.

Our work leaves a couple of interesting open problems for further research:
\begin{itemize}
    \item \textbf{Reducing the gap further.} Given the lower bound that $\Omega(Ln)$ bits of communication are necessary for $\ba$, can one obtain deterministic $\ca$ for arbitrary convexity spaces with communication $\bigO(Ln+\poly(n,\kappa))$?

    \item \textbf{Sharper resilience.}
    Is it possible to achieve $\ca$ on arbitrary convexity spaces with communication complexity $o(Ln^2)$ with near-optimal resilience (with or without $\varepsilon$-slack) when no upper bound on the honest inputs' bit-length is known in advance?

    \item \textbf{Beyond $\ca$: general validity in $\ba$.}
    Can our extractor-based committee approach be extended to synchronous $\ba$ variants with more general validity conditions?
\end{itemize}

\section{Acknowledgements}

D.\ Ghinea was supported in part by the SNSF grant number 2000-1-243069.

\bibliography{bib/abbrev1,bib/crypto,bib/project}

@string{virtual =               "Virtual Event"}

@string{ieee =                  {IEEE}}

@string{springer =              "Springer"}

@string{dagstuhl =              "Schloss Dagstuhl"}

@string{jacm =                  "Journal of the {ACM}"}

@string{acm =                   "Association for Computing Machinery"}

@article{LaShPe82,
  title={The byzantine generals problem},
  author={Lamport, Leslie and Shostak, Robert and Pease, Marshall},
  journal={ACM Transactions on Programming Languages and Systems},
  volume={4},
  number={3},
  pages={382--401},
  year={1982}
}

@article{PSL80,
  title={Reaching agreement in the presence of faults},
  author={Pease, Marshall and Shostak, Robert and Lamport, Leslie},
  journal={Journal of the ACM (JACM)},
  volume={27},
  number={2},
  pages={228--234},
  year={1980},
  publisher={ACM}
}

@article{DolStr83,
  title={Authenticated algorithms for Byzantine agreement},
  author={Dolev, Danny and Strong, H. Raymond},
  journal={SIAM Journal on Computing},
  volume={12},
  number={4},
  pages={656--666},
  year={1983},
  publisher={SIAM}
}

@InProceedings{BenDoHo10,
author="Ben-Or, Michael
and Dolev, Danny
and Hoch, Ezra N.",
editor="Lynch, Nancy A.
and Shvartsman, Alexander A.",
title="Brief Announcement: Simple Gradecast Based Algorithms",
booktitle="Distributed Computing",
year="2010",
publisher="Springer Berlin Heidelberg",
address="Berlin, Heidelberg",
pages="194--197",
isbn="978-3-642-15763-9"
}

@article{JACM:DLPSW86,
	author = {Dolev, Danny and Lynch, Nancy A. and Pinter, Shlomit S. and Stark, Eugene W. and Weihl, William E.},
	title = {Reaching Approximate Agreement in the Presence of Faults},
	year = {1986},
	issue_date = {July 1986},
	publisher = {Association for Computing Machinery},
	address = {New York, NY, USA},
	volume = {33},
	number = {3},
	issn = {0004-5411},
	url = {https://doi.org/10.1145/5925.5931},
	doi = {10.1145/5925.5931},
	journal = {J. ACM},
	month = may,
	pages = {499–516},
	numpages = {18}
}

@InProceedings{OPODIS:AAD04,
	author="Abraham, Ittai
	and Amit, Yonatan
	and Dolev, Danny",
	editor="Higashino, Teruo",
	title="Optimal Resilience Asynchronous Approximate Agreement",
	booktitle="Principles of Distributed Systems",
	year="2005",
	publisher="Springer Berlin Heidelberg",
	address="Berlin, Heidelberg",
	pages="229--239",
	isbn="978-3-540-31584-1"
}

@InProceedings{DISC:NoRy19,
	author =	{Thomas Nowak and Joel Rybicki},
	title =	{{Byzantine Approximate Agreement on Graphs}},
	booktitle =	{33rd International Symposium on Distributed Computing (DISC 2019)},
	pages =	{29:1--29:17},
	series =	{Leibniz International Proceedings in Informatics (LIPIcs)},
	ISBN =	{978-3-95977-126-9},
	ISSN =	{1868-8969},
	year =	{2019},
	volume =	{146},
	editor =	{Jukka Suomela},
	publisher =	{Schloss Dagstuhl--Leibniz-Zentrum fuer Informatik},
	address =	{Dagstuhl, Germany},
	URL =		{http://drops.dagstuhl.de/opus/volltexte/2019/11336},
	URN =		{urn:nbn:de:0030-drops-113363},
	doi =		{10.4230/LIPIcs.DISC.2019.29}
}

@inproceedings{helly,
  title={Helly's theorem and its relatives},
  author={Danzer, Ludwig},
  booktitle={Proc. Symp. Pure Math.},
  volume={7},
  pages={101--180},
  year={1963},
  organization={Amer. Math. Soc.}
}

@article{DIST:MHVG15,
  title={Multidimensional agreement in Byzantine systems},
  author={Mendes, Hammurabi and Herlihy, Maurice and Vaidya, Nitin and Garg, Vijay K},
  journal={Distributed Computing},
  volume={28},
  number={6},
  pages={423--441},
  year={2015},
  publisher={Springer}
}

@article{Fekete90,
  title={Asymptotically optimal algorithms for approximate agreement},
  author={Fekete, Alan David},
  journal={Distributed Computing},
  volume={4},
  number={1},
  pages={9--29},
  year={1990},
  publisher={Springer}
}

@inproceedings{SIROCCO:Alistarh21,
author="Alistarh, Dan
and Ellen, Faith
and Rybicki, Joel",
editor="Jurdzi{\'{n}}ski, Tomasz
and Schmid, Stefan",
title="Wait-Free Approximate Agreement on Graphs",
booktitle="Structural Information and Communication Complexity",
year="2021",
publisher="Springer International Publishing",
address="Cham",
pages="87--105",
isbn="978-3-030-79527-6",
url="https://doi.org/10.1007/978-3-030-79527-6_6",
doi="10.1007/978-3-030-79527-6_6"
}

@inproceedings{SuVai16,
  title={Fault-tolerant multi-agent optimization: optimal iterative distributed algorithms},
  author={Su, Lili and Vaidya, Nitin H},
  booktitle={Proceedings of the 2016 ACM symposium on principles of distributed computing},
  pages={425--434},
  year={2016}
}

@InProceedings{DISC:NRSVX20,
  author =	{Kartik Nayak and Ling Ren and Elaine Shi and Nitin H. Vaidya and Zhuolun Xiang},
  title =	{{Improved Extension Protocols for Byzantine Broadcast and Agreement}},
  booktitle =	{34th International Symposium on Distributed Computing (DISC 2020)},
  pages =	{28:1--28:17},
  series =	{Leibniz International Proceedings in Informatics (LIPIcs)},
  ISBN =	{978-3-95977-168-9},
  ISSN =	{1868-8969},
  year =	{2020},
  volume =	{179},
  editor =	{Hagit Attiya},
  publisher =	{Schloss Dagstuhl--Leibniz-Zentrum f{\"u}r Informatik},
  address =	{Dagstuhl, Germany},
  URL =		{https://drops.dagstuhl.de/opus/volltexte/2020/13106},
  URN =		{urn:nbn:de:0030-drops-131064},
  doi =		{10.4230/LIPIcs.DISC.2020.28}
}

@inproceedings{DISC:MoRe21,
  title={Optimal Communication Complexity of Authenticated Byzantine Agreement},
  author={Momose, Atsuki and Ren, Ling},
  booktitle={35th International Symposium on Distributed Computing (DISC 2021)},
  year={2021},
  organization={Schloss Dagstuhl-Leibniz-Zentrum f{\"u}r Informatik}
}

@article{ReedSolomon,
  title={Polynomial codes over certain finite fields},
  author={Reed, Irving S and Solomon, Gustave},
  journal={Journal of the society for industrial and applied mathematics},
  volume={8},
  number={2},
  pages={300--304},
  year={1960},
  publisher={SIAM}
}

@inproceedings{MerkleTrees,
  title={A digital signature based on a conventional encryption function},
  author={Merkle, Ralph C},
  booktitle={Conference on the theory and application of cryptographic techniques},
  pages={369--378},
  year={1987},
  organization={Springer}
}

@inproceedings{CTRSA:Nguyen05,
  title={Accumulators from bilinear pairings and applications},
  author={Nguyen, Lan},
  booktitle={Topics in Cryptology--CT-RSA 2005: The Cryptographers’ Track at the RSA Conference 2005, San Francisco, CA, USA, February 14-18, 2005. Proceedings},
  pages={275--292},
  year={2005},
  organization={Springer}
}

@InProceedings{SPAA:GhLiWa23,
author = {Ghinea, Diana and Liu-Zhang, Chen-Da and Wattenhofer, Roger},
title = {Multidimensional Approximate Agreement with Asynchronous Fallback},
year = {2023},
isbn = {9781450395458},
publisher = {Association for Computing Machinery},
address = {New York, NY, USA},
url = {https://doi.org/10.1145/3558481.3591105},
doi = {10.1145/3558481.3591105},
booktitle = {Proceedings of the 35th ACM Symposium on Parallelism in Algorithms and Architectures},
pages = {141–151},
numpages = {11},
location = {Orlando, FL, USA},
series = {SPAA '23}
}

@inproceedings{eprint:ConvexWorld,
  author =	{Constantinescu, Andrei and Ghinea, Diana and Wattenhofer, Roger and Westermann, Floris},
  title =	{{Convex Consensus with Asynchronous Fallback}},
  booktitle =	{38th International Symposium on Distributed Computing (DISC 2024)},
  pages =	{15:1--15:23},
  series =	{Leibniz International Proceedings in Informatics (LIPIcs)},
  ISBN =	{978-3-95977-352-2},
  ISSN =	{1868-8969},
  year =	{2024},
  volume =	{319},
  publisher =	{Schloss Dagstuhl -- Leibniz-Zentrum f{\"u}r Informatik},
  address =	{Dagstuhl, Germany},
  URL =		{https://drops.dagstuhl.de/entities/document/10.4230/LIPIcs.DISC.2024.15},
  URN =		{urn:nbn:de:0030-drops-212411},
  doi =		{10.4230/LIPIcs.DISC.2024.15}
}

@inproceedings{IEEE:MelWat18,
  author={Melnyk, Darya and Wattenhofer, Roger},
  booktitle={2018 IEEE 37th Symposium on Reliable Distributed Systems (SRDS)}, 
  title={Byzantine Agreement with Interval Validity}, 
  year={2018},
  volume={},
  number={},
  pages={251-260},
  address={Salvador, Brazil},
  url          = {https://doi.org/10.1109/SRDS.2018.00036},
  publisher    = {{IEEE} Computer Society},
  doi={10.1109/SRDS.2018.00036}
}

@inproceedings{OPODIS:StolWat15,
  author =	{David Stolz and Roger Wattenhofer},
  title =	{{Byzantine Agreement with Median Validity}},
  booktitle =	{19th International Conference on Principles of Distributed Systems (OPODIS 2015)},
  pages =	{1--14},
  series =	{Leibniz International Proceedings in Informatics (LIPIcs)},
  ISBN =	{978-3-939897-98-9},
  ISSN =	{1868-8969},
  year =	{2016},
  volume =	{46},
  editor =	{Emmanuelle Anceaume and Christian Cachin and Maria Potop-Butucaru},
  publisher =	{Schloss Dagstuhl--Leibniz-Zentrum fuer Informatik},
  address =	{Dagstuhl, Germany},
  URL =		{http://drops.dagstuhl.de/opus/volltexte/2016/6591},
  URN =		{urn:nbn:de:0030-drops-65911},
  doi =		{10.4230/LIPIcs.OPODIS.2015.22}
}

@InProceedings{OPODIS:CGHWW23,
  author =	{Constantinescu, Andrei and Ghinea, Diana and Heimbach, Lioba and Wang, Zilin and Wattenhofer, Roger},
  title =	{{A Fair and Resilient Decentralized Clock Network for Transaction Ordering}},
  booktitle =	{27th International Conference on Principles of Distributed Systems (OPODIS 2023)},
  pages =	{8:1--8:20},
  series =	{Leibniz International Proceedings in Informatics (LIPIcs)},
  ISBN =	{978-3-95977-308-9},
  ISSN =	{1868-8969},
  year =	{2024},
  volume =	{286},
  editor =	{Bessani, Alysson and D\'{e}fago, Xavier and Nakamura, Junya and Wada, Koichi and Yamauchi, Yukiko},
  publisher =	{Schloss Dagstuhl -- Leibniz-Zentrum f{\"u}r Informatik},
  address =	{Dagstuhl, Germany},
  URL =		{https://drops.dagstuhl.de/entities/document/10.4230/LIPIcs.OPODIS.2023.8},
  URN =		{urn:nbn:de:0030-drops-194989},
  doi =		{10.4230/LIPIcs.OPODIS.2023.8}
}

@article{IPL:TurCoa84,
  title={Extending binary Byzantine agreement to multivalued Byzantine agreement},
  author={Turpin, Russell and Coan, Brian A},
  journal={Information Processing Letters},
  volume={18},
  number={2},
  pages={73--76},
  year={1984},
  publisher={Elsevier}
}

@inproceedings{SRDS:CacTes05,
  title={Asynchronous verifiable information dispersal},
  author={Cachin, Christian and Tessaro, Stefano},
  booktitle={24th IEEE Symposium on Reliable Distributed Systems (SRDS'05)},
  pages={191--201},
  year={2005},
  publisher    = {{IEEE} Computer Society},
  organization={IEEE},
  address = {Orlando, FL, {USA}},
  URL = {https://doi.org/10.1109/RELDIS.2005.9},
  doi={10.1109/RELDIS.2005.9}}

@article{CoaWel92,
  title={Modular construction of a Byzantine agreement protocol with optimal message bit complexity},
  author={Coan, Brian A and Welch, Jennifer L},
  journal={Information and Computation},
  volume={97},
  number={1},
  pages={61--85},
  year={1992},
  publisher={Elsevier}
}

@inproceedings{delphi24,
  author    = {A. Bandarupalli and A. Bhat and S. Bagchi and A. Kate and C.-D. Liu-Zhang and M. K. Reiter},
  title     = {{Delphi: Efficient Asynchronous Approximate Agreement for Distributed Oracles}},
  booktitle = {Proceedings of the 2024 54th Annual IEEE/IFIP International Conference on Dependable Systems and Networks (DSN)},
  pages     = {456--469},
  year      = {2024},
  publisher = {IEEE},
  address   = {Brisbane, Australia},
  doi       = {10.1109/DSN58291.2024.00051},
  url       = {https://doi.org/10.1109/DSN58291.2024.00051}
}

@inproceedings{federated21,
author = {El-Mhamdi, El-Mahdi and Farhadkhani, Sadegh and Guerraoui, Rachid and Guirguis, Arsany and Hoang, L\^{e}-Nguy\^{e}n and Rouault, S\'{e}bastien},
title = {Collaborative learning in the jungle (decentralized, byzantine, heterogeneous, asynchronous and nonconvex learning)},
year = {2021},
isbn = {9781713845393},
publisher = {Curran Associates Inc.},
address = {Red Hook, NY, USA},
booktitle = {Proceedings of the 35th International Conference on Neural Information Processing Systems},
articleno = {1918},
numpages = {14},
series = {NIPS '21}
}

@InProceedings{king09,
author="King, Valerie
and Saia, Jared",
editor="Keidar, Idit",
title="From Almost Everywhere to Everywhere: Byzantine Agreement with $\widetilde{O}(n^{3/2})$ Bits",
booktitle="Distributed Computing",
year="2009",
publisher="Springer Berlin Heidelberg",
address="Berlin, Heidelberg",
pages="464--478",
isbn="978-3-642-04355-0"
}

@InProceedings{Rog06,
author="Rogaway, Phillip",
editor="Nguyen, Phong Q.",
title="Formalizing Human Ignorance",
booktitle="Progress in Cryptology - VIETCRYPT 2006",
year="2006",
publisher="Springer Berlin Heidelberg",
address="Berlin, Heidelberg",
pages="211--228",
isbn="978-3-540-68800-6"
}

@inproceedings{DISC:CohKeiSpi20,
  title={Not a COINcidence: Sub-Quadratic Asynchronous Byzantine Agreement WHP},
  author={Cohen, Shir and Keidar, Idit and Spiegelman, Alexander},
  booktitle={34th International Symposium on Distributed Computing},
  year={2020}
}

@Inbook{UnauthenticatedBA,
author="Berman, Piotr
and Garay, Juan A.
and Perry, Kenneth J.",
editor="Baeza-Yates, Ricardo
and Manber, Udi",
title="Bit Optimal Distributed Consensus",
bookTitle="Computer Science: Research and Applications",
year="1992",
publisher="Springer US",
address="Boston, MA",
pages="313--321",
isbn="978-1-4615-3422-8",
doi="10.1007/978-1-4615-3422-8_27",
url="https://doi.org/10.1007/978-1-4615-3422-8_27"
}

@INPROCEEDINGS{sensor24,
  author={Bandarupalli, Akhil and Bhat, Adithya and Chaterji, Somali and Reiter, Michael K. and Kate, Aniket and Bagchi, Saurabh},
  booktitle={2024 IEEE 44th International Conference on Distributed Computing Systems (ICDCS)}, 
  title={SensorBFT: Fault-Tolerant Target Localization Using Voronoi Diagrams and Approximate Agreement}, 
  year={2024},
  volume={},
  number={},
  pages={186-197},
  doi={10.1109/ICDCS60910.2024.00026}}

@inproceedings{SPAA:CMMS25,
author = {Cambus, M\'{e}lanie and Melnyk, Darya and Milentijevi\'{c}, Tijana and Schmid, Stefan},
title = {{Approximate Agreement Algorithms for Byzantine Collaborative Learning}},
year = {2025},
isbn = {9798400712586},
publisher = {Association for Computing Machinery},
address = {New York, NY, USA},
url = {https://doi.org/10.1145/3694906.3743343},
doi = {10.1145/3694906.3743343},
booktitle = {Proceedings of the 37th ACM Symposium on Parallelism in Algorithms and Architectures},
pages = {89–100},
numpages = {12},
location = {Portland, OR, USA},
series = {SPAA '25}
}

@inproceedings{federated20,
author = {El-Mhamdi, El-Mahdi and Guerraoui, Rachid and Guirguis, Arsany and Hoang, L\^{e} Nguy\^{e}n and Rouault, S\'{e}bastien},
title = {Genuinely Distributed Byzantine Machine Learning},
year = {2020},
isbn = {9781450375825},
publisher = {Association for Computing Machinery},
address = {New York, NY, USA},
url = {https://doi.org/10.1145/3382734.3405695},
doi = {10.1145/3382734.3405695},
booktitle = {Proceedings of the 39th Symposium on Principles of Distributed Computing},
pages = {355–364},
numpages = {10},
location = {Virtual Event, Italy},
series = {PODC '20}
}

@article{pseudorandomness,
    author = {Vadhan, Salil P.},
    title = {Pseudorandomness},
    journal = {Foundations and Trends in Theoretical Computer Science},
    volume = {7},
    number = {1-3},
    pages = {1-336},
    year = {2012},
    month = {12},
    issn = {1551-305X},
    doi = {10.1561/0400000010},
}

@inbook{ShaExtractor,
author = {Ronen Shaltiel},
title = {Recent Developments in Explicit Constructions of Extractors},
booktitle = {Current Trends in Theoretical Computer Science},
chapter = {},
pages = {189-228},
doi = {10.1142/9789812562494_0013},
}

@article{DISC:ACFR19,
author = {Alc{\'{a}}ntara, Manuel and Casta{\~{n}}eda, Armando and Flores-Pe{\~{n}}aloza, David and Rajsbaum, Sergio},
doi = {10.1007/s00446-018-0345-3},
issn = {1432-0452},
journal = {Distributed Computing},
number = {3},
pages = {235--255},
title = {{The topology of look-compute-move robot wait-free algorithms with hard termination}},
url = {https://doi.org/10.1007/s00446-018-0345-3},
volume = {32},
year = {2019}
}

@inproceedings{PODC:GhLiWa22,
author = {Ghinea, Diana and Liu-Zhang, Chen-Da and Wattenhofer, Roger},
title = {Optimal Synchronous Approximate Agreement with Asynchronous Fallback},
year = {2022},
isbn = {9781450392624},
publisher = {Association for Computing Machinery},
address = {New York, NY, USA},
url = {https://doi.org/10.1145/3519270.3538442},
doi = {10.1145/3519270.3538442},
booktitle = {Proceedings of the 2022 ACM Symposium on Principles of Distributed Computing},
pages = {70–80},
numpages = {11},
location = {Salerno, Italy},
series = {PODC'22}
}

@inproceedings{PODC:GhMeMi26,
author = {Ghinea, Diana and Melnyk, Darya and  Milentijević, Tijana},
title = {{Network-Agnostic Multidimensional Approximate Agreement with Optimal Resilience}},
year = {2026},
publisher = {Association for Computing Machinery},
url = {https://doi.org/10.1145/3796701.3815960},
doi = {10.1145/3796701.3815960},
booktitle = {Proceedings of the 2026 ACM Symposium on Principles of Distributed Computing},
location = {Egham, United Kingdom},
series = {PODC'26}
}

@inproceedings{ARXIV:FuGhPa25,
      title={{Round and Resilience-Optimal Approximate Agreement on Trees and Block Graphs}}, 
      author={Marc Fuchs and Diana Ghinea and Zahra Parsaeian and Joel Rybicki},
    year = {2026},
    publisher = {Association for Computing Machinery},
    url = {https://doi.org/10.1145/3796701.3815967},
    doi = {10.1145/3796701.3815967},
    booktitle = {Proceedings of the 2026 ACM Symposium on Principles of Distributed Computing},
    location = {Egham, United Kingdom},
    series = {PODC'26}
}

@misc{SokDiana,
      author = {Diana Ghinea and Chen-Da Liu-Zhang},
      title = {{SoK}: Approximate Agreement},
      howpublished = {Cryptology {ePrint} Archive, Paper 2025/2339},
      year = {2025},
      url = {https://eprint.iacr.org/2025/2339}
}

@InProceedings{MoseArxivNew,
  author =	{Mizrahi Erbes, Mose and Wattenhofer, Roger},
  title =	{{Asynchronous Approximate Agreement with Quadratic Communication}},
  booktitle =	{29th International Conference on Principles of Distributed Systems (OPODIS 2025)},
  pages =	{16:1--16:26},
  year =	{2026},
  volume =	{361},
  publisher =	{Schloss Dagstuhl -- Leibniz-Zentrum f{\"u}r Informatik},
  address =	{Dagstuhl, Germany},
  doi =		{10.4230/LIPIcs.OPODIS.2025.16},
}

@InProceedings{DISC25:AdaptiveBA,
  author =	{Constantinescu, Andrei and Dufay, Marc and Paramonov, Anton and Wattenhofer, Roger},
  title =	{{Brief Announcement: From Few to Many Faults: Adaptive Byzantine Agreement with Optimal Communication}},
  booktitle =	{39th International Symposium on Distributed Computing (DISC 2025)},
  pages =	{52:1--52:8},
  year =	{2025},
  volume =	{356},
  publisher =	{Schloss Dagstuhl -- Leibniz-Zentrum f{\"u}r Informatik},
  address =	{Dagstuhl, Germany},
  doi =		{10.4230/LIPIcs.DISC.2025.52},
}

@article{MultiDimBA,
author = {Flamini, Andrea and Longo, Riccardo and Meneghetti, Alessio},
year = {2022},
month = {03},
pages = {1-19},
title = {Multidimensional Byzantine agreement in a synchronous setting},
volume = {35},
journal = {Applicable Algebra in Engineering, Communication and Computing},
doi = {10.1007/s00200-022-00548-5}
}

@article{dolev1985bounds,
  title={Bounds on information exchange for byzantine agreement},
  author={Dolev, Danny and Reischuk, R{\"u}diger},
  journal={Journal of the ACM (JACM)},
  volume={32},
  number={1},
  pages={191--204},
  year={1985},
  publisher={ACM New York, NY, USA}
}

\newpage
\appendix
\section*{Appendix}

\section{Dilation Factor of $\mathbb{R}^d$} \label{appendix:dilation-factor}

We present the proof of \Cref{theo:dilation-factor-euclidian}, restated below.
\DilationFactorR*

\begin{proof}
Let $d > 0$, and consider an encoding $\enc$ of $\mathbb{R}^d$ with input set $X$. We remark that $\mathbb{R}^d$ has Helly number $\omega = d+1$. Within this proof, by polytope, we mean convex bounded polytope.

We first define $X^\star$. Informally, $X^\star$ is the smallest set that contains $X$ and all intermediate values that can be created from $X$. To define it formally, given a multiset of points $S$, we define $\safe(S) = \cap_{A \subseteq S, |S \setminus A| < |S|/\omega} \hull{A}$. We remark that for any set $A$, $\hull{A}$ is a polytope and thus, as a finite intersection of polytopes, for any multiset $S$, $\safe(S)$ is also a polytope. For a polytope $P$, we define $v(P)$ as an arbitrary vertex of $P$ (for example the one with the smallest coordinates when ordered lexicographically, i.e., the vertex whose first coordinate is smallest, breaking ties by the second coordinate, then the third, and so on).

We define $X^\star$ recursively: we start with $X_0 = X$. Then, for any $k \geq 1$, $X_k$ contains the intermediate values that can be obtained, roughly, after $k$ steps of applying $\safe$, on values from $X_{k - 1}$. That is, on any multiset containing values from $X_{k - 1}$, with any multiplicity.
Formally, $X_k$ is defined as follows:
$$ X_{k} = X_{k-1} \cup \left(
\bigcup_{\substack{
    S \text{ non-empty finite multiset} \\
  \text{with values from } X_{k - 1}
}}
v(\safe(S)) \right).$$ 

We may then set $X^\star = \lim_{k \to \infty} X_k = \cup_{k=0}^\infty X_k$.

Then, using $X^\star,$ we define the extension $\enc^\star$ recursively: 
\begin{itemize}
    \item For $x \in X_0 = X$, $\enc^\star(x) = 1 || \enc(x)$. Note that this ensures that $\enc^\star$ is a proper extension of $\enc$.
    \item For $k \geq 1$ and $x \in X_k \setminus X_{k-1}$, we can find $S \subseteq X_{k-1}$ such that $x = v(\safe(S))$: this comes from the definition of $X_k$. 
    
    For a bitstring $b$, we define $\textsc{Prefix}(b) = 1_{|b|}0b$, which is the length of $b$ in unary followed by $0$ and then the value of $b$.
    For example, $\textsc{Prefix}(101) = 1110101$.
    $\textsc{Prefix}(b)$ allows to encode $b$ in a prefix-free manner. 

    Let $A \subseteq S$, such that $|S \setminus A| < |S|/\omega$. Let $m$ be the dimension of $\hull{A}$, i.e., the dimension of the smallest affine space which contains $A$. As $\hull{A}$ is a convex hull of dimension $m$, for each of its facets, the affine space supporting it is uniquely identified by $m$ distinct elements in $S$.
    
    We remark that a point $y \in \mathbb{R}^d$ is in $\hull{A}$ if and only if it is in the half-plane defined by each facet of $\hull{A}$. For a facet $F$ of $A$, we remark that the half-plane it defines can be described by $d - m + 1$ affine constraints, which can be computed deterministically.
    Therefore, as $\safe(S)$ is an intersection of convex hulls, an element of $\mathbb{R}^d$ is in $\safe(S)$ if and only if it satisfies the affine constraints defined by each convex hull.
    
    Then, as $x$ is a vertex of $\safe(S)$, it is an extremal point of it, hence we can find $d$ independent affine constraints active at $x$ among the affine constraints provided by each facet of each polytope in the intersection defining  $\safe(S)$. Let $c_1, c_2, ..., c_d$ be such constraints. We remark that $x$ is uniquely defined by these constraints. Moreover, given a constraint $c_i$ provided by a polytope $P_i$ of dimension $m_i$, $c_i$ is itself uniquely defined by: (i) the polytope's facet $F_i$, and (ii) the index of $c_i$ among the constraints provided by $F_i$: this index is denoted by $l_i$, and is between $0$ and $d-m_i$. Meanwhile, the facet $F_i$ is itself uniquely defined by $m_i$ elements $y_{i,1},..., y_{i, m_i}$ of $S$. Let $e_{i,1},...,e_{i, m_i}$ be the $\enc^\star$-encoding of $y_{i,1},..., y_{i, m_i}$. Let $m_{i,b}$, $l_{i,b}$ be the binary encoding of respectively $m_i$ and $l_i$ using exactly $\lceil \log d\rceil$ bits. We can now define the encoding of x:

    \begin{align*}
        \enc^\star(x) = 0\ ||\ &m_{1, b}\ ||\ l_{1,b}\ ||\ 
        \textsc{Prefix}(e_{1,1})\ ||\ \ldots\ ||\ \textsc{Prefix}(e_{1, m_1})\ || \\
        &m_{2, b}\ ||\ l_{2,b}\ ||\ 
        \textsc{Prefix}(e_{2,1})\ ||\ \ldots\ ||\ \textsc{Prefix}(e_{2, m_2})\ || \\
        &\ldots  \\
        &m_{d, b}\ ||\ l_{d,b}\ ||\ 
        \textsc{Prefix}(e_{d,1})\ ||\ \ldots\ ||\ \textsc{Prefix}(e_{d, m_d}).
    \end{align*}

    The prefix-free encoding of the values ensures that $\enc^\star$ is injective. For a bitstring $y$, remark that $|\textsc{Prefix}(y)| = 2|y| + 1$. Assuming all elements of $S$ have $\enc^\star$-encoding length at most $L \geq 1$, we get:
    \begin{align*}
        |\enc^\star(x)| &\leq d\cdot (2 \lceil \log d\rceil + d\cdot (2 L)) \\
        & \leq d^2 \cdot 4 L && \text{Because $\lceil \log d\rceil \leq d$}
    \end{align*}
\end{itemize}
This proves that $\mathbb{R}^d$ has dilation factor $4 d^2$.
\end{proof}

\section{Parallel Composition Proof} \label{section:parallel-composition}

We present the proof of \Cref{theo:parallel-is-easy}, restated below. 

\ParallelIsEasy*
\begin{proof}
    We define $\mathcal{C}' = 
    \mathcal{C}_1 \times \mathcal{C}_2 \times \ldots \times \mathcal{C}_q$. $\mathcal{C}'$ is the cartesian product of all previous spaces and as a consequence a finite space with $L$-bit inputs. We define the convexity operation on $\mathcal{C}'$ as follows: for $k \geq 1$ and $(v_i = (v_{i,1}, \ldots, v_{i,q}))_{i \in \{1\ldots k\}}$ $k$ elements in $\mathcal{C}'$:
    
    \begin{align*}
        \hull{v_i : i \in \{1\ldots k\}} = \hull{ v_{i, 1} : i \in \{1\ldots k\}} \times \ldots \times \hull{v_{i, q} : i \in \{1\ldots k\}}.
    \end{align*}

    Elements in a convex set in $\mathcal{C}'$ are elements where each component is in the convex set of the matching $\mathcal{C}_i$. As a consequence, $\mathcal{C}'$ can be used to represent all convexity relations on $\mathcal{C}_1, \mathcal{C}_2, \ldots, \mathcal{C}_q$. We now want to prove that it has Helly number at most $\omega_{\max}$. The reason for this low Helly number is because every convex set is well structured. To be more specific, because we defined a convex set as a cartesian product and cartesian products are stable by intersection, we can prove that for any convex subset $S \subseteq \mathcal{C}'$, there exists $S_1, S_2, \ldots, S_q \in \mathcal{C}_1, \mathcal{C}_2, \ldots, \mathcal{C}_q$ such that $S = S_1 \times S_2 \times \ldots \times S_q$.

    Let $k \geq \omega_{\max}$ and $A_1, A_2, \ldots, A_k$ be convex subsets of $\mathcal{C}'$ such that the intersection of $\omega_{\max}$ of these subsets is not empty. We want to prove that $A = \cap_{i=1}^k A_i \ne \emptyset$. As stated previously, each $A_i$ can be represented as a cartesian product $A_i = A_{i,1} \times \ldots \times A_{i,q}$. 
    Let $j \in \{1,\ldots,q\}$. By the definition of convexity on $\mathcal{C}'$ and the way we chose $A_1, A_2, \ldots, A_k$, the intersection of any $\omega_{\max}$ subsets among $A_{1, j}, \ldots, A_{k,j}$ is not empty. Because $\omega_{max} \geq \omega_j$, the intersection of any $\omega_{j}$ of these subsets is also not empty. Because they are convex subsets of $\mathcal{C}_j$, 
    $\bigcap_{i=1}^k A_{i,j} \ne \emptyset$.

    Because $A = \left( \bigcap_{i=1}^k A_{i,1} \right) \times \cdots \times \left( \bigcap_{i=1}^k A_{i,q} \right)$, $A$ is therefore the cartesian product of non-empty spaces, so it is non-empty itself. Therefore, $\mathcal{C}'$ has Helly number at most $\omega_{\max}$.
\end{proof}

\section{Multi-Participation Safety and Extra-Corruptions Safety} \label{section:ba:extra-properties}

As described in \Cref{section:general-protocol}, our protocols assume a $\ba$ protocol $\Pi_{\ba}$ and run it in a non-standard setting: among a group of \emph{virtual} parties, where one concrete party may simulate multiple virtual parties, and where the number of byzantine virtual parties may exceed the threshold against which $\Pi_{\ba}$ achieves security. We expect $\Pi_{\ba}$ to maintain its communication complexity and round complexity regardless of the number of byzantine virtual parties (Extra-Corruptions Safety), and to retain security whenever the group of virtual parties satisfies its resilience threshold (Multi-Participation Safety). The $\ba$ protocols we use naturally achieve these properties. In this section, we first describe how virtual parties run $\Pi_{\ba}$, then explain why this gives Multi-Participation Safety, and why Extra-Corruptions Safety also holds.

\paragraph{Mapping between parties and simulated parties.} We recall that the set of virtual parties $\PS_V := \{\party_1, \party_2, \ldots, \party_m\}$ and the mapping to real parties $M : \PS \rightarrow 2^{\PS_V}$ are known to all parties. For every virtual party $\party_i \in \PS_V$, we write $M^{-1}(\party_i)$ for the unique real party that simulates $\party_i$. As we will be sending party indices over the network, we also define $p(\party_i, M(\party))$ as the position of $\party_i$ in the list $M(\party)$ sorted increasingly by the indices $i$. In our protocols, each party simulates a constant number of virtual parties at a time; hence $p(\party_i, M(\party))$ can be encoded using $\bigO(1)$ bits for every $\party \in \PS$ and every $\party_i \in M(\party)$.

\paragraph{Communication over simulated parties.} Whenever virtual party $\party_i \in M(\party)$ needs to send a message $\msg$ to virtual party $\party_j$, party $\party$ simulates this locally if $\party_j \in M(\party)$. Otherwise, it sends the message $(\msg, p(\party_i, M(\party)), p(\party_j, M(\party')))$ to $\party' = M^{-1}(\party_j)$. When $\party'$ receives this message from $\party$, it forwards $\msg$ to the virtual party $\party_j$ that it simulates, treating it as a message from $\party_i$. This incurs only an additional $\bigO(1)$ bits of overhead per message.

\paragraph{Signatures for simulated parties.} For authenticated $\ba$, we need to account for the fact that the PKI is set up for the real parties, and not for the simulated parties. For the protocol of \cite{DISC:MoRe21}, it suffices to modify the signing procedure as follows: if virtual party $\party_i$ needs to sign message $\msg$, then $\party = M^{-1}(\party_i)$ signs $(\msg, i)$ using its own secret key. When virtual party $\party_j$ needs to verify the signature of simulated party $\party_i$ on message $\msg$, $\party' = M^{-1}(\party_j)$ verifies the signature on $(\msg, i)$ using the public key of $\party = M^{-1}(\party_i)$. Including the virtual identity $i$ in the signed message prevents a signature generated on behalf of one virtual party from being reused on behalf of another.

\paragraph{Multi-Participation Safety.} The communication and signature simulations above make each honest virtual party behave as a distinct party in a standard execution of $\Pi_{\ba}$: each honest virtual party maintains its own local state, sends the messages prescribed by $\Pi_{\ba}$ for its virtual identity, receives messages labeled with the corresponding virtual sender and receiver identities, and, in authenticated settings, verifies signatures with respect to the simulated sender's virtual identity. The inclusion of the virtual identity in signed messages ensures that signatures produced on behalf of one honest simulated identity cannot be transferred to another simulated identity. Simulated identities that belong to byzantine real parties may, but these are precisely the virtual parties treated as byzantine in the simulated execution. Then, the simulated execution can be interpreted as a standard execution of $\Pi_{\ba}$ among the virtual parties $\PS_V$, where a virtual party is byzantine exactly when it is simulated by a byzantine real party. Consequently, if the set of virtual parties satisfies the resilience threshold of $\Pi_{\ba}$, then the usual security guarantees of $\Pi_{\ba}$ apply to the simulated execution. This gives Multi-Participation Safety.

\paragraph{Extra-Corruptions Safety: round complexity.} We may assume without loss of generality that any $\ba$ protocol $\Pi_{\ba}$ with a known upper bound $R$ on its round complexity maintains this bound even in a setting with more byzantine virtual parties than $\Pi_{\ba}$ can tolerate: once $R$ rounds have passed, every honest virtual party simply outputs $\bot$ if it has not already terminated.

\paragraph{Extra-Corruptions Safety: communication complexity.} Assuming that a $\ba$ protocol $\Pi_{\ba}$ maintains its communication upper bound even in a setting with more corrupted parties than it can tolerate is not without loss of generality. However, this property is satisfied by many protocols whose message pattern is fixed in advance and for which an upper bound on the number of bits sent by each honest party in each step can be fixed in advance. This is the case, for instance, in the protocol of \cite{UnauthenticatedBA} for unauthenticated settings, and in the protocol of \cite{DISC:MoRe21} for authenticated settings: in each step, each honest party sends a bounded number of messages, each containing one value of bounded size, since we run $\ba$ on inputs of known length, and, in the authenticated case, up to $\bigO(n)$ signatures of size $\bigO(\kappa)$ each. The simulation above adds only $\bigO(1)$ bits to each message, and therefore preserves the same asymptotic communication bounds.

\section{Authenticated $\ba$: Communication Complexity}
\label{sec:appendix:ba communication}
The $\Pi_{\ba}$ protocol of \cite{DISC:MoRe21} achieves communication complexity $\bitcost_{L,n}(\Pi_{\ba}) = \bigO(\kappa n^2)$ for $L \in \bigO(\kappa)$. We make $L$ explicit in their analysis, and show that  $\bitcost_{L,n}(\Pi_{\ba}) = \bigO((L + \kappa) \cdot n^2)$.

$\Pi_{\ba}$ (presented in \cite[Figure 1]{DISC:MoRe21}) assumes that the input values come from an arbitrary input space $V$. It relies on a weaker variant of $\ba$, namely Graded Byzantine Agreement ($\gba$) -- parties initially hold inputs from $V$, and, instead of agreeing on a value $v \in V$, each party obtains a value $v \in V$ and a grade of confidence $g \in \{0, 1\}$. The Validity property for $\gba$ requires that, if all honest parties hold input $v$, then every honest party outputs $(v, 1)$. Agreement is replaced with Consistency: if an honest party outputs $(v, 1)$, then no honest party outputs $(v', g')$ with $v' \neq v$. 
Assuming a $\gba$ protocol $\Pi_{\gba}$, $\Pi_{\ba}$ becomes a recursive construction:
if the set of parties is small (at most a constant $M$ satisfying $2^{r - 1} M \leq n \leq 2^r$), the parties rely on an inefficient $\ba$ protocol. Otherwise, the parties first run $\Pi_{\gba}$. This is followed by a $\Pi_{\ba}$ execution on the first half of the set of parties -- they join with the value obtained in $\Pi_{\gba}$, and report the output obtained to the entire set. Note that $\Pi_{\ba}$ is only guaranteed to succeed if the first half of the set of parties contains an honest majority. Each party $\party$ checks if a majority of the values reported are the same -- if this is the case, and $\party$ did not obtain a different value with grade $1$ in $\Pi_{\gba}$, it adopts this value. This way, if the first half of the set of parties contains an honest majority, $\ba$ is already achieved. If this was not the case, then the second half of the set of parties contains an honest majority. The protocol continues with another execution of  $\Pi_{\gba}$ on the values adopted, a $\Pi_{\ba}$ execution on the second half of the parties, and another check. If $\ba$ was already achieved in the first part, it is maintained, otherwise it is now achieved.

The work of \cite{DISC:MoRe21} proposes two constructions of $\Pi_{\gba}$. The first construction, described in \cite[Section 4.1]{DISC:MoRe21}, leads to communication complexity $\bigO(\kappa n^2)$ in $\Pi_{\ba}$ for inputs of $\bigO(\kappa)$ bits even up to $t < n/2$ corruptions, but relies on threshold signatures. The second construction, described in \cite[Section 4.2]{DISC:MoRe21}, leads to the same communication complexity for $t < (1/2 - \varepsilon) \cdot n$ for any predefined constant $\varepsilon > 0$, but without threshold signatures.  
Recall that we run $\Pi_{\ba}$ in settings with virtual parties, where multiple virtual parties may be simulated by the same party, hence we cannot rely on threshold signatures, and therefore we rely on the latter $\Pi_{\gba}$ construction.

While \cite{DISC:MoRe21} states that $\Pi_{\gba}$ has communication complexity $\bitcost_{\bigO(1), s} = \bigO(\kappa s^2)$, it is easy to make the inputs' length $L$ explicit in their analysis.
Throughout the execution of $\Pi_{\gba}$ in a setting of $s$ parties, each party sends  $\bigO(s)$ messages, and $\bigO(d)$ certificates, where $d \in \bigO(1/\varepsilon) = \bigO(1)$, as $\varepsilon$ is a constant. Each message consists of an (input) value and a signature, and each certificate consists of $\bigO(s)$ signed (input) values. As signatures consist of $\bigO(\kappa)$ bits and we run $\Pi_{\ba}$, and therefore $\Pi_{\gba}$ in settings where an upper bound $L$ on the honest inputs' length is known, each message has length $\bigO(L + \kappa)$, and each certificate has size $\bigO(s \cdot (L + \kappa))$. This leads to $\bigO(s \cdot (L + \kappa)) + \bigO(d \cdot s \cdot (L + \kappa)) = \bigO(s \cdot (L + \kappa))$ bits per party. Then, the communication complexity becomes $\bitcost_{L, s}(\Pi_{\gba}) := \bigO((L + \kappa) \cdot s^2)$. 

We may then make $L$ explicit in the analysis of $\Pi_{\ba}$ as well, by following \cite[Lemma 10]{DISC:MoRe21}: in an instance over $s \leq M$ parties, one can use an inefficient $\ba$ protocol with communication complexity $\bigO((L + \kappa) \cdot s^3)$ (as we use this protocol in settings where $L$ is known, we may obtain this via $\plainbc$ \cite{DolStr83}). As $s = \bigO(1)$ in this case, the communication complexity is $C(s) = \bigO(L + \kappa)$. In an instance over $s > M$ parties, the communication complexity is defined recursively: 
\begin{align*}
    C(s) &= \bitcost_{L, \lfloor s/2 \rfloor}(\Pi_{\gba}) + C(\lfloor s/2 \rfloor) + \bigO((L + \kappa) \cdot s^2)  + C(\lceil s/2 \rceil) +  \bitcost_{L, \lceil s/2 \rceil}(\Pi_{\gba}) \\
    &= \bigO((L + \kappa) \cdot s^2) + C(\lfloor s/2 \rfloor) + C( \lceil s/2 \rceil).
\end{align*}

As $2^{r - 1} \cdot M \leq n \leq 2^r$, it follows that:
\begin{align*}
    C(n) \leq 2^r \cdot \bigO(L + \kappa) + 
    \sum_{i = 0}^r  2^i \cdot \bigO((L + \kappa) \cdot (n/2^i)^2) = \bigO((L + \kappa) \cdot  n^2),
\end{align*}
hence $\bitcost_{L,n}(\Pi_{\ba}) = C(n) = \bigO((L + \kappa) \cdot  n^2)$, as claimed.

\section{Supersending} \label{section:supersending}

We discuss the protocol described in \Cref{lemma:supersending}, restated below, enabling efficient communication among supernodes, or, in general, groups of virtual parties. This protocol relies on techniques from prior works on extension protocols \cite{DISC:NRSVX20,AC:BLLN22, PODC:GhiLiuWat25}. We recall that $\supersend^L$ assumes a group of virtual parties $A$ needs to send a message to a group of virtual parties $B$.

\SupersendLemma*

We recall that $\ba$-friendly refers to settings where $\ba$ can be achieved. As described in \Cref{section:general-protocol}, we make use of $\ba$ protocols that still retain security in settings of virtual parties, and still maintain their communication and round complexity guarantees even in settings where the number of corruptions exceeds the security corruption threshold. The protocol behind \Cref{lemma:supersending} assumes a $\ba$ protocol $\Pi_{\ba}$, and we recall that the term \emph{$\ba$-friendly} refers to settings where the group of virtual parties respects the resilience thresholds of the assumed $\ba$ protocol. In unauthenticated settings, we instantiate this with that of \Cref{thm:magic-ba-no-pki}. In authenticated settings, we make use of the protocol of \Cref{thm:magic-ba-pki}.
We will not mention this throughout the section unless relevant, for simplicity of presentation --  we will simply refer to $\Pi_{\ba}$. In addition, when referring to $\bitcost_{\ell, n}(\Pi_{\ba})$ and $\roundcost_{n}(\Pi_{\ba})$, we take the highest cost between the two protocols. Hence, we write $\bitcost_{\ell, n}(\Pi_{\ba}) = \bigO((\ell + \kappa) \cdot n^2)$, and $\roundcost_{n}(\Pi_{\ba}) = \bigO(n)$.

Hence, we need to enable the virtual parties in $A$ to distribute $m$ efficiently to the parties in $B$ so that, (1) if $B$ is $\ba$-friendly, all parties in $B$ obtain the same message $m'$, and (2) if in addition $A$ contains an honest (virtual) majority, the parties in $B$ obtain $m = m'$. Similarly to prior works, we achieve this using Reed-Solomon ($\rs$) codes \cite{ReedSolomon} and Merkle Trees ($\mt$) \cite{MerkleTrees}. We describe these below.

\paragraph{Reed-Solomon codes.} $\rs$ codes enable us to split a message $m$ into $n_B$ shares (codewords) so that: the message can be uniquely reconstructed from any subset of $n_B - t_B$ of the $n_B$ shares. 

We assume standard $\rs$ codes with parameters $(n_B, n_B - t_B)$, with $t_B = \lfloor n_B / 2 \rfloor$. This provides us with a deterministic algorithm $\rs.\encode(m)$, which takes an $L$-bit message $m$ as input and converts it into $n_B$ shares $(\share_1, 
\ldots, \share_{n_B})$ of $\bigO(L / (n_B - t_B)) = \bigO(L / n_B)$ bits each. The shares $\share_i$ are elements of a Galois Field $\mathbb{F} = GF(2^a)$ with $n_B \leq 2^a - 1$.

To reconstruct the original message, $\rs$ codes provide a decoding algorithm, $\rs.\decode$, which takes as input $n_B - t_B$ (indexed) shares $(i, \share_i)$ and returns the original message $m$. Any $n_B - t_B$ of the $n_B$ shares uniquely determine the original message $m$.

\paragraph{Merkle trees.} 
To enable the parties to detect corrupted codewords, and also to compress messages, we rely on a cryptographic accumulator \cite{CTRSA:Nguyen05}. 

Essentially, accumulators convert a set (in our case, the $n_B$ codewords) into a $\kappa$-bit value and provide witnesses confirming the accumulated set's contents. For this task, we use Merkle Trees ($\mt$) \cite{MerkleTrees}, which do not require a trusted dealer. 

An $\mt$ is a balanced binary tree that enables us to compress a list of values into a $\kappa$-bit encoding, and to efficiently verify that a value belongs to the compressed list.
Given a list $S = \{ \share_1, \ldots, \share_{n_B} \}$, the $\mt$ is built bottom-up, using the collision-resistant hash function $H_{\kappa}$: starting with $n_B$ leaves, where the $i$-th leaf stores $H_{\kappa}(\share_i)$. Each non-leaf node stores $H_{\kappa}(h_{\leftt} \concat h_{\rightt})$, where $h_{\leftt}$ and $h_{\rightt}$ are the hashes stored by the node's left and right child, respectively.
This way, the hash stored by the root represents the encoding of $S$.

Given the root's hash $z$, one can prove that $\share_i$ belongs to the compressed list using a witness $w_i$ of $\bigO(\kappa \cdot \log n_B)$ bits. The witness $w_i$ contains the hashes needed to verify the path from the $i$-th leaf to the root.

We use $\mt.\buildd(S)$ to denote the (deterministic) algorithm that creates the $\mt$ for the given list $S$ and returns the hash stored by the root $z$ and the witnesses $w_1, w_2, \ldots w_{n_B}$.

Afterwards, $\mt.\verifyy(z, i, \share_i, w_i)$ returns $\true$ if $w_i$ proves that $H_{\kappa}(\share_i)$ is indeed stored on the $i$-th leaf of the $\mt$ with root hash $z$ and $\false$ otherwise.

Note that collision-resistance implies that it is computationally infeasible to produce distinct lists with the same root, which prevents the adversary from producing witnesses for values of its own choice. 
Equivalently, once a root $z$ is fixed, there is at most one value $\share_i$ for each index $i$ that admits a valid witness under $z$. Thus, for a fixed pair $(z,i)$, all valid tuples $(z,i,\share_i,w_i)$ contain the same share $\share_i$.

\paragraph{Protocol.} The supersending protocol will be divided into multiple steps.

In the first step, each virtual party $\party$ in $A$ splits the message $m$ into $n_B$ shares using $\rs.\encode$. Afterwards, $\party$ computes the encoding $z$ as the root of the $\mt$ tree compressing the $n_B$ shares, along with the witnesses $w_i$ for every virtual party $\party_i$. $\party$ then sends $(z, i, \share_i, w_i)$ to each party $\party_i \in B$.

Upon receiving the shares, every virtual party $\party_i$ in $B$ checks if $\mt.\verifyy(z, i, \share_i, w_i) = \true$ holds: if this is not the case, $\party_i$ ignores the tuple received. If $\mt.\verifyy(z, i, \share_i, w_i) = \true$, we say that $(z, i, \share_i, w_i)$ is a \emph{valid tuple}.

In addition, $\party_i$ verifies the length of the tuple received -- this can be done as we assume a publicly known upper bound $L$ on the bit length of the message to be sent.

This step is implemented in subprotocol $\sendShares^L$, which we present below.
\begin{dianabox}{$\party$ $\sendShares^L$ to $B$}
	\algoHead{Code for virtual party $\party$ holding message $m$}
	\begin{algorithmic}[1]
        \State $\share_1, \share_2, \ldots, \share_{n_B} := \rs.\encode(m)$; $z, w_1, w_2, \ldots, w_{n_B} := \mt.\buildd(\{\share_1, \share_2, \ldots, \share_{n_B}\})$.
        \State Send $(z, i, \share_i, w_i)$ to every virtual party  $\party_i$ in $B$.
	\end{algorithmic}

    \algoHead{Code for (virtual) party $\party_i \in B$:}
	\begin{algorithmic}[1]
     \setcounter{ALG@line}{2} % start at line 3
        \State If you have received a tuple $(z, i, \share_i, w_i)$ consisting of $O(L/ n_B + \kappa \log n_B)$ bits from $P$ such that $\mt.\verifyy(z, i, \share_i, w_i) = \true$, output $(z, i, \share_i, w_i)$. Otherwise, output $\bot$.
	\end{algorithmic}
\end{dianabox}
We describe the guarantees of $\sendShares^L$ in the lemma below.
\begin{lemma}\label{lemma:send-shares}
    $\sendShares^L$ has round complexity $\roundcost(\sendShares^L) = 1$, and communication complexity $\bitcost_{L, n_B}(\sendShares^L) := \bigO(L + \kappa n_B \log n_B)$.
    Moreover:
    \begin{itemize}
        \item Every honest virtual party $\party_i \in B$ outputs either $\bot$ or a valid tuple $(z, i, \share_i, w_i)$ consisting of $O(L/ n_B + \kappa \log n_B)$ bits.
        \item If $\party$ is honest, every honest virtual party $\party_i \in B$ outputs a valid tuple $(z, i, \share_i, w_i)$. In addition, if $S$ is a set containing more than $n_B/2$ of the honest parties' indexed shares $(i, \share_i)$, $\rs.\decode(S) = m$.
    \end{itemize}
\end{lemma}
\begin{proof}
    The round complexity follows trivially from the protocol's description. For communication complexity, we only need to consider the case where $\party$ is honest: 
    $\party$ sends a tuple $(z,i,  \share_i, w_i)$ to every virtual party $\party_i \in B$. Each such tuple consists of 
    $\bigO(\kappa + L / n_B + \kappa \log n_B)$ bits, hence the sender $\party$ sends $\bigO(L + \kappa n_B \log n_B)$ bits.

    The parties' output format follows from the protocol description.

    If $\party$ is honest, then it has computed the shares $\share_i$, the encoding $z$ and the witnesses $w_i$ correctly, and has sent the corresponding tuple to each virtual party $\party_i$ in $B$. Hence, each honest party $\party_i$ holds a valid tuple. Next, if $S$ is a set containing more than $n_B/2$ of the honest virtual parties' indexed shares, these uniquely determine the original message $m$, hence $\rs.\decode(S) = m$.
\end{proof}

If $A$ contains byzantine virtual parties, after $\sendShares^L$ is executed for each virtual party in $A$, each virtual party in $B$ may hold multiple  distinct tuples -- for distinct encodings $z$. We first need to decide the message behind which $z$ should be reconstructed, and afterwards enable the virtual parties to attempt to reconstruct it.

If $A$ contains an honest majority proposing the same message, we want to make sure that the parties agree to decode the encoding $z$ corresponding to that message. Hence, each virtual party $\party_i$ in $B$ considers the tuple  $(z, i, \share_i, w_i)$ that it has received most often in the $\sendShares^L$ invocations. This way, if $A$ contains a majority of honest virtual parties that joined $\sendShares^L$ with the same message $m$, each honest virtual party in $B$ chooses the tuples corresponding to $m$. Afterwards, each party in $B$ joins $\Pi_{\ba}$ with the encoding $z$ considered above as input, and hence the parties agree on an encoding $z^\star$.

At this point, prior works \cite{DISC:NRSVX20,AC:BLLN22, PODC:GhiLiuWat25} run a subsequent $\Pi_{\ba}$ invocation to verify whether some honest (virtual) party has proposed $z = z^\star$ -- this would imply that $z$ can be reconstructed. We lose this implication, as the honest virtual parties in $B$ might not hold sufficiently many shares for $z$.

For now, we simply attempt reconstruction: parties holding $z = z^\star$ distribute their shares, and then the parties attempt to reconstruct $m^\star := \rs.\decode(S)$, where $S$ is the set of valid tuples received for $z$. It is possible that some honest virtual parties reconstruct $m^\star$, while others do not.

We note that it is possible that the encoding $z^\star$ agreed upon comes in fact from byzantine parties in $A$: byzantine parties can select some arbitrary shares instead of running $\rs.\encode(\cdot)$, and construct their Merkle trees using these arbitrary shares.
Then, the $\mt.\verifyy(\cdot)$ checks may pass. However, $\mt.\verifyy(\cdot)$ only ensures that the shares received come from the committed vector of shares, which does not imply that these shares come from $\rs.\encode(\cdot)$.
Therefore, the parties that decode a message $m^\star$ run one more check before deciding to output $m^\star$: they compute $\rs.\encode(m^\star)$ and afterwards they rebuild the corresponding Merkle root. If the resulting root matches $z^\star$, it is safe to output $m^\star$.

We present the code of $\reconstructFromShares^L$ below, and the guarantees are stated by \Cref{lemma:reconstruct-from-shares}. For the final reconstruction, it will be useful that the parties also output the encoding $z^\star$ agreed upon.

\begin{dianabox}{$\reconstructFromShares^L$}
	\algoHead{Code for virtual party $\party_i \in B$ holding a multiset $M$ of shares $(z, i, \share_i, w_i)$}
	\begin{algorithmic}[1]
        \State If $M$ is empty, set $z$ to a fixed default $\kappa$-bit string and set $\tau := \bot$. Otherwise, let $\tau=(z,i,\share_i,w_i)$ be the tuple that appears most often in $M$ (breaking ties arbitrarily), and let $z$ be its encoding.
        \State Join $\Pi_{\ba}$ with input $z$ and obtain output $z^\star$.
        \State If $z^\star = z$ and $\tau \neq \bot$, send $(i,\share_i,w_i)$ to every virtual party $\party \in B$.
        \State Discard any tuples $(j, \share_j, w_j)$ where $\mt.\verifyy(z^{\star}, j, \share_j, w_j) = \false$.
        \State Let $S := \{(j,\share_j)\ :\ (j,\share_j,w_j)\text{ was not discarded}\}$.
        \State If $\abs{S} \leq  n_B/2$, output $(z^\star, \bot)$ and terminate.
        \State Otherwise, let $m^{\star} := \rs.\decode(S)$
        \State If $m^\star = \bot$, output $(z^\star,\bot)$ and terminate.
        \State $\share'_1, \share'_2, \ldots, \share'_{n_B} := \rs.\encode(m^\star)$; $z', \ldots := \mt.\buildd(\{\share'_1, \share'_2, \ldots, \share'_{n_B}\})$.
        \State Output $(z^\star,m^\star)$ if $z^\star = z'$ and $m^\star$ consists of up to $L$ bits, and output $(z^\star,\bot)$ otherwise.
	\end{algorithmic}
\end{dianabox}

\begin{lemma}\label{lemma:reconstruct-from-shares}
    Assume that every honest virtual party in $B$ joins $\reconstructFromShares^L$ with a multiset of valid tuples of $\bigO(L / n_B + \kappa \log n_B)$ bits each.
    Then, every honest virtual party completes the execution of $\reconstructFromShares^L$ with an output $(z^\star, m^\star)$, where $z^\star$ is a $\kappa$-bit string and $m^\star$ is either an $L$-bit string or $\bot$, within $\roundcost_{n_B}(\reconstructFromShares^L) = \bigO(n_B)$ rounds, and with communication complexity $\bitcost_{L, n_B}(\reconstructFromShares^L) = \bigO(L n_B + \kappa n_B^2 \log n_B)$.

    Moreover, if $B$ is $\ba$-friendly, then 
    all honest parties output the same encoding $z^\star$, and either all honest parties output $m^\star = \bot$ or there exists an $L$-bit message $m$ such that
    (i) $z^\star$ is the encoding of $m$ and 
    (ii) every honest virtual party in $B$ outputs $m^\star \in \{m,\bot\}$.

    If, in addition, there is an $L$-bit message $m$ such that all honest parties in $B$ hold tuples appearing most often for the same encoding $z$, where $z$ is the Merkle root obtained from $\rs.\encode(m)$, then all honest parties in $B$ output $m^\star := m$.
\end{lemma}
\begin{proof}
    We first discuss the round complexity guarantees, regardless of the number of byzantine virtual parties involved. The honest virtual parties in $B$ start the first execution of $\Pi_{\ba}$, and, as $\Pi_{\ba}$ is Extra-Corruptions-Safe, regardless of whether an output is obtained by round $\roundcost_{n_B}(\Pi_{\ba})$ or not, the execution of $\Pi_{\ba}$ is halted. This is followed by one more round of communication, hence all honest virtual parties complete the execution of $\reconstructFromShares^L$ at the end of round $1 + \roundcost_{n_B}(\Pi_{\ba}) = \bigO(n_B)$. 

    Next, we discuss the communication complexity, also regardless of the number of byzantine virtual parties involved. The honest virtual parties in $B$ run a $\Pi_{\ba}$ invocation on $O(\kappa)$-sized inputs. As $\Pi_{\ba}$ is Extra-Corruptions-Safe, this leads to  $\bitcost_{\kappa, n_B}(\Pi_{\ba})$ bits of communication. Next, as each tuple consists of $\bigO(\kappa + L / n_B + \kappa \log n_B)$ bits, the tuples' distribution takes up to $\bigO(L n_B + \kappa n_B^2 \log n_B)$ bits in total. Consequently, honest virtual parties in $B$ send up to $\bigO(L n_B + \kappa n_B^2 \log n_B) + \bitcost_{\kappa, n_B}(\Pi_{\ba}) = \bigO(L n_B + \kappa n_B^2 \log n_B)$ bits.

    Now, assume that $B$ is $\ba$-friendly. As $\Pi_{\ba}$ achieves $\ba$ in $B$, the honest virtual parties obtain the same encoding $z^\star$ in the first execution of $\Pi_{\ba}$, hence they output the same encoding $z^\star$ in $\reconstructFromShares^L$. Then, the virtual parties attempt to decode $z^\star$.

    Assume that two honest virtual parties output $m^\star_1 \neq \bot$ and $m^\star_2 \neq \bot$. By the check in the last line of the protocol, the Merkle roots of $\rs.\encode(m^\star_1)$ and $\rs.\encode(m^\star_2)$ are both equal to $z^\star$. By collision-resistance of the Merkle tree, these two encoded share vectors are identical: $\rs.\encode(m^\star_1)=\rs.\encode(m^\star_2)$.
    Since any $n_B-t_B$ coordinates of a $\rs$ codeword uniquely determine the encoded message, two messages whose $\rs$ encodings are identical must be equal: $m^\star_1=m^\star_2$.

    It remains to consider the case where $B$ is $\ba$-friendly, and all honest virtual parties in $B$ hold tuples appearing most often for the same encoding $z$, where $z$ is the Merkle root obtained from $\rs.\encode(m)$ for some $L$-bit message $m$. We show that all honest virtual parties in $B$ output $m^\star := m$.
    In this case, all honest virtual parties join $\Pi_{\ba}$ with the same input $z$. As $\Pi_{\ba}$ achieves $\ba$, the Validity condition ensures that the honest virtual parties obtain $z^\star := z$. Then, the honest virtual parties distribute their tuples. As $B$ is $\ba$-friendly, there is an honest majority in $B$, and hence each honest virtual party in $B$ receives more than $n_B/2$ tuples valid for $z^\star$. 
    All tuples accepted at this point are valid with respect to the same root $z^\star$. Since the Merkle tree admits at most one valid opening for each pair $(z^\star,i)$, every accepted tuple with index $i$ contains the same share as the corresponding tuple from $\rs.\encode(m)$. Hence the set of accepted indexed shares is a subset of the codeword $\rs.\encode(m)$, and because it has size greater than $n_B/2$, $\rs.\decode$ returns $m$.
    Moreover, the final consistency check succeeds, since re-encoding $m$ and rebuilding the Merkle tree yields the same root $z^\star$. Therefore all honest virtual parties in $B$ output $m^\star := m$.
\end{proof}

Whenever $B$ is $\ba$-friendly, $\reconstructFromShares^L$ ensures that there is a message $m$ such that all honest virtual parties in $B$ have obtained $m$ or $\bot$.
We make use of $\Pi_{\ba}$ again, enabling the honest virtual parties to decide whether they should simply output $\bot$ or, if there are honest virtual parties in $B$ holding $m$, reconstruct $m$.

If the decision is to reconstruct $m$, the virtual parties in $B$ holding $m$ run $\sendShares^L$ to distribute shares for $m$ to all virtual parties in $B$, and then the virtual parties in $B$ run $\reconstructFromShares^L$ to reconstruct $m$ using the shares obtained from the virtual parties in $B$: after this step, all honest virtual parties in $B$ output the same message $m$.

We present the code of $\supersend^L$ below.

\begin{dianabox}{$\supersend^L$}
	\algoHead{Code for virtual party $\party \in A$ holding message $m$}
	\begin{algorithmic}[1]
        \State Distribute $m$ to the parties in $B$ using $\sendShares^L$.
	\end{algorithmic}

	\algoHead{Code for virtual party $\party_i \in B$}
	\begin{algorithmic}[1]
        \State Join the $\sendShares^L$ execution for every virtual party in $A$.
        \State Run $\reconstructFromShares^L$ using the tuples obtained in $\sendShares^L$, and obtain output $(z^\star, m^\star)$.
        \State Join $\Pi_{\ba}$ with input $1$ if $m^\star \neq \bot$ and $0$ otherwise. If $\Pi_{\ba}$ did not return $1$, output $\bot$ and terminate.
        \State If you have obtained $m^\star \neq \bot$, distribute $m^\star$ to all virtual parties in $B$ via $\sendShares^L$.
        \State Let $M^\star$ be the multiset of tuples received in Line 4 whose Merkle root is $z^\star$ and whose witnesses verify with respect to $z^\star$. Run $\reconstructFromShares^L$ on $M^\star$, and obtain output $(z^\star, m^\star)$.
        \State Output $m^\star$.
	\end{algorithmic}
\end{dianabox}

We are now ready to present the proof of \Cref{lemma:supersending}.

\begin{proof}[Proof of \Cref{lemma:supersending}]
We first discuss the round complexity and the communication complexity, regardless of the number of byzantine virtual parties involved.

The honest virtual parties in $A \cup B$ first execute $\sendShares^L$. According to \Cref{lemma:send-shares}, this ensures termination within $\roundcost_{n_B}(\sendShares^L)$ rounds, and with communication complexity 
$n_A \cdot \bitcost_{L, n_B}(\sendShares^L)$.

Afterwards, the honest virtual parties in $B$ execute $\reconstructFromShares^L$. By \Cref{lemma:reconstruct-from-shares}, this ensures termination within $\roundcost_{n_B}(\reconstructFromShares^L)$ rounds, and with communication complexity $\bitcost_{L, n_B}(\reconstructFromShares^L)$.

This is followed by one execution of $\Pi_{\ba}$, which ends after $\roundcost_{n_B}(\Pi_{\ba})$ rounds, with communication complexity $\bitcost_{1, n_B}(\Pi_{\ba})$.

Finally, the virtual parties in $B$ execute $\sendShares^L$ and $\reconstructFromShares^L$ once again, which leads to $\roundcost_{n_B}(\sendShares^L) + \roundcost_{n_B}(\reconstructFromShares^L)$ additional rounds of communication, and to $n_B \cdot \bitcost_{L, n_B}(\sendShares^L) + \bitcost_{L, n_B}(\reconstructFromShares^L)$ additional bits of communication.

In total, the round complexity can be written as follows:
\begin{align*}
2 \cdot \roundcost_{n_B}(\sendShares^L) + 2 \cdot \roundcost_{n_B}(\reconstructFromShares^L) + \roundcost_{n_B}(\Pi_{\ba}).
\end{align*}

For communication complexity, we obtain: 
\begin{align*}
    (n_A + n_B) \cdot \bitcost_{L, n_B}(\sendShares^L) + 2 \cdot \bitcost_{L, n_B}(\reconstructFromShares^L) + \bitcost_{1, n_B}(\Pi_{\ba}).
\end{align*}

We may then replace the building blocks' round complexities and communication complexities in these expressions with those ensured by the results describing them: \Cref{lemma:send-shares} for $\sendShares^L$, \Cref{lemma:reconstruct-from-shares} for $\reconstructFromShares^L$, and \Cref{thm:magic-ba-no-pki} / \Cref{thm:magic-ba-pki} for $\Pi_{\ba}$. This way, the final round complexity becomes $ \bigO(n_B)$, and the communication complexity becomes $\bigO(L \cdot (n_A + n_B) + \kappa n_B \log n_B \cdot (n_A + n_B))$, as claimed in the statement.

In the remainder of the proof, we assume that $B$ is $\ba$-friendly.

We show that the honest virtual parties in $B$ obtain the same output in $\supersend^L$. 
We first note that the $\Pi_{\ba}$ invocation in line 3 succeeds: with output $0$ or $1$.

If $\Pi_{\ba}$ outputs $0$, all honest parties output $\bot$ in $\supersend^L$, hence the claim follows. Otherwise, if $\Pi_{\ba}$ returns $1$, there is an honest virtual party $\party$ that has joined $\Pi_{\ba}$ with input $1$ due to $\ba$'s Validity condition. Then, $\party$ has obtained $m^\star \neq \bot$ in $\reconstructFromShares^L$. By \Cref{lemma:reconstruct-from-shares}, it follows that every honest party either obtains $m^\star \neq \bot$ or $\bot$ in $\reconstructFromShares^L$. Moreover, \Cref{lemma:reconstruct-from-shares} ensures that the honest parties have obtained the same $z^\star$ that encodes $m^\star$.

Then, the honest parties that obtained $m^\star \neq \bot$ distribute shares for $m^\star$ via $\sendShares^L$: \Cref{lemma:send-shares} ensures that, in each of the $\sendShares^L$ invocations of line 4 having honest senders, the honest parties obtain valid shares for $z^\star$.
As $\rs.\encode$ and $\mt.\buildd$ are deterministic, all honest senders produce the same tuple for each index $i$.
That is, every honest party $\party_i$ receives at least one valid tuple for
$z^\star$ in Line 4; moreover, since $\party_i$ only considers tuples for its own index $i$, collision-resistance of the Merkle tree implies that all valid tuples for $(z^\star,i)$ are identical.

Afterwards, the honest parties join the second $\reconstructFromShares^L$ invocation using their shares that are valid for $z^\star$.
Each honest party defines $M^\star$ as the multiset of tuples whose Merkle root is $z^\star$ and whose witnesses verify with respect to $z^\star$. Since at least one honest party obtained $m^\star \neq \bot$ in the first reconstruction, at least one honest sender executes Line 4. Therefore every honest party $\party_i$ receives a valid tuple for its own index $i$ under root $z^\star$. Moreover, because all tuples in $M^\star$ have Merkle root $z^\star$, every honest party that invokes $\reconstructFromShares^L$ on $M^\star$ chooses a tuple with encoding $z^\star$ as its most frequent tuple.

Moreover, by the first invocation of \Cref{lemma:reconstruct-from-shares}, the encoding $z^\star$ is the Merkle root obtained from $\rs.\encode(m^\star)$.
Then, \Cref{lemma:reconstruct-from-shares} ensures that all honest parties obtain $m^\star$ in the second $\reconstructFromShares^L$ invocation, and therefore output $m^\star$ in $\supersend^L$. Note that $m^\star$ consists of up to $L$ bits due to \Cref{lemma:reconstruct-from-shares}.

Finally, we assume that, in addition to $B$ being $\ba$-friendly, $A$ contains more than $n_A / 2$ honest virtual parties that hold the same $L$-bit message $m$.
In this case, $\sendShares^L$ ensures that every honest virtual  party in $B$ joins the first invocation of $\reconstructFromShares^L$ with a majority of identical tuples valid with respect to the encoding of $m$. Then, \Cref{lemma:reconstruct-from-shares} ensures that all honest virtual parties obtain $m^\star = m$ in the first execution of $\reconstructFromShares^L$. Afterwards, the honest virtual parties in $B$ distribute $m$ in $\supersend^L$, hence \Cref{lemma:send-shares}
ensures that the honest virtual parties obtain tuples for $z^\star$ that enable them to reconstruct $m^\star$ once again in $\reconstructFromShares^L$, and hence output $m^\star$.
\end{proof}

\section{CA within a Supernode} \label{appendix:ca-within-supernode}

We present our protocol $\supernodeCA^L$, described by \Cref{lemma:ca-within-supernode}, restated below.
\CAWithinSupernodes*

As discussed in \Cref{subsection:init-supernodes}, we first provide the honest parties with an identical view, i.e., a multiset $\msgset$, containing the honest inputs. To achieve this, each party sends its input value $v_{\inputt}$ to all parties. Afterwards, for every party $\party$, the parties join a $\ba$ protocol $\Pi_{\ba}$ -- with the value received from $\party$ if it consists of up to $L$ bits, and with a default value otherwise. If the setting is authenticated, we instantiate $\Pi_{\ba}$ with the protocol described by \Cref{thm:magic-ba-pki}. Otherwise, we use the $\ba$ protocol described by \Cref{thm:magic-ba-no-pki}. This way, if the supernode $\supernodeCA^L$ runs in is non-byzantine, the parties have indeed obtained an identical multiset $\msgset$ containing all honest inputs.
Then, each party computes the \emph{safe area} $S := \safe_{k}(\msgset)$ with $k := \abs{\msgset} - (n_S - t_S)$, and outputs a value from $S$ (if any) by making a deterministic decision.

\begin{dianabox}{$\supernodeCA^L$}
	\algoHead{Code for party $\party \in S$  with input $v_{\inputt}$}
	\begin{algorithmic}[1]
    \State Send $v_{\inputt}$ to all parties.
    \State Join a $\Pi_{\ba}$ invocation for every party $\party'$: with input $v'$ if you have received an $L$-bit value $v'$ from $\party'$, and with a default $L$-bit input from $V_{\inputt}$ otherwise.
    \State Let $\msgset$ be the multiset of $L$-bit outputs from $\Pi_{\ba}$, $t_S = \lceil n_S / \omega \rceil$ - 1, $k := \abs{\msgset} - (n_S - t_S)$. 
    \State If $k < 0$ or $\safe_k(\msgset) = \emptyset$, output $\bot$.
    \State Otherwise, let $v$ be the value with the smallest binary encoding in $\safe_k(\msgset)$, breaking ties deterministically. If $v$ is a $(\delta \cdot L)$-bit value, output $v$; otherwise output $\bot$.
	\end{algorithmic}
\end{dianabox}

We present the proof of \Cref{lemma:ca-within-supernode} below.

\begin{proof}[Proof of \Cref{lemma:ca-within-supernode}]
    $\supernodeCA^L$ makes $n_S$ parallel invocations of $\Pi_{\ba}$ on $L$-bit inputs. Regardless of whether we instantiate the $\Pi_{\ba}$ protocol with that of \Cref{thm:magic-ba-no-pki} or that of \Cref{thm:magic-ba-pki}, we may write $\roundcost_{n_S}(\Pi_{\ba}) = \bigO(n_S)$ and $\bitcost_{L,n_S}(\Pi_{\ba}) = \bigO( (L + \kappa) \cdot n_S^2)$. 
    Then, due to the Extra-Corruptions Safety property of $\Pi_{\ba}$, regardless of whether $S$ is a byzantine supernode or not, $\supernodeCA^L$ achieves round complexity $\roundcost_{n_S}(\Pi_{\ba}) = \bigO(n_S)$  and communication complexity $n_S \cdot \bitcost_{L,n_S}(\Pi_{\ba}) = \bigO( (L + \kappa) \cdot n_S^3)$, as claimed in the lemma's statement.

    In addition, by the protocol's description, if an honest party outputs $v \neq \bot$, then $v$ is a $(\delta \cdot L)$-bit value.

    If $S$ is a non-byzantine supernode, it is $\ba$-friendly, and hence every $\Pi_{\ba}$ invocation provides the honest parties with the same $L$-bit value. Moreover, each $\Pi_{\ba}$ invocation corresponding to an honest party $\party$ provides the honest parties with the input $v_{\inputt}$ of $\party$: this is because all honest parties receive  $v_{\inputt}$ from $\party$, hence they join $\Pi_{\ba}$ with input $v_{\inputt}$, and the Validity property of  $\Pi_{\ba}$ ensures they agree on $v_{\inputt}$.

    Consequently, the honest parties have obtained the same multiset $\msgset$ of up to $n_S$ values out of which at least $\abs{\msgset} - t_S$ come from honest parties. 
    As $t_S < n_S / \omega$, \Cref{lemma:safe-area} ensures that the honest parties obtain a non-empty safe area that is included in the honest inputs' convex hull, and hence they choose a value $v$ in the honest inputs' convex hull. As the space has dilation factor $\delta$, $v$ has bit-length up to $\delta \cdot L$, hence all honest parties output $v$.
\end{proof}

\section{Unknown Inputs' Bit-Lengths: Missing Details}\label{appendix:unknown-L-missing-details}

We present in detail the subprotocols of $\Pi_{\ca}$, described in \Cref{section:unknown-L}. Afterwards, we present the guarantees of $\Pi_{\ca}$.

\subsection{Naive Approximation of $L$} \label{subsection:naive-approximation}
If the honest parties agree on $L \leq n^2$, we rely on a naive approximation approach: each party joins with input $L_{\inputt} :=$ its input's bit-length. The parties compare their inputs' bit-lengths with all relevant powers of two using $\Pi_{\ba}$ (in parallel) and return the most suitable one.

\begin{dianabox}{$\exponentialSearch$}
	\algoHead{Code for party $\party$ with input $L_{\inputt} \in \{1, \ldots, n^2\}$}
	\begin{algorithmic}[1]
    \State \textbf{In parallel} for $i=0\ldots \lceil \log n^2 \rceil$
        \State  \hspace{0.5cm} Join instance $i$ of $\Pi_{\ba}$ with input $b_{\inputt} := 1$ if $L_{\inputt} \leq 2^i$ and $0$ otherwise.
    \State Let $j$ be the lowest instance index $i$ for which $\Pi_{\ba}$ returned $1$.
    \State Output $\widetilde{L} := 2^j$
	\end{algorithmic}
\end{dianabox}

\begin{lemma}\label{theo:expo-seach-theo}
    Assume that $t < n / 3$, and that every honest party joins $\exponentialSearch$ with an input $L_{\inputt} \in \{1, \ldots, n^2\}$, and let $L_{\min}$ and $L_{\max}$ be the lowest and respectively the maximum honest inputs $L_{\inputt}$. Then, the honest parties obtain a value $\widetilde{L} \in \naturalnumbers$ satisfying $L_{\min} \leq \widetilde{L} < 2 \cdot L_{\max}$.
    
    $\exponentialSearch$ runs in $\bigO(n)$ rounds, and has communication complexity $\bigO(n^2 \cdot \log n)$.
\end{lemma}
\begin{proof}
    We first note that $\exponentialSearch$ runs $\lfloor \log n^2 \rfloor = \bigO(\log n)$ parallel $\Pi_{\ba}$ invocations. 
    According to \Cref{thm:magic-ba-no-pki}, $\Pi_{\ba}$ has communication complexity $\bigO(n^2)$ and round complexity $\bigO(n)$. This immediately implies the communication complexity and the round complexity of $\exponentialSearch$ claimed in the lemma's statement.

    Due to $\Pi_{\ba}$ achieving Agreement, we have the guarantee that the honest parties obtain the same output in each instance of $\Pi_{\ba}$, consequently obtain the same $j$ and the same $\widetilde{L}$. Moreover, as  $L_{\max} \leq n^2$, the highest-index instance of $\Pi_{\ba}$ (with $i = \lceil \log n^2 \rceil$) returns $1$ due to its Validity, hence the value $j$ in the protocol is properly defined. 
    Moreover, as instance $j$ of $\Pi_{\ba}$ returns $1$, then the Validity condition of $\Pi_{\ba}$ ensures that there is at least one honest party that has joined with input $1$, hence $L_{\min} \leq 2^j = \widetilde{L}$.
    
    It remains to show that $\widetilde{L} \leq 2 \cdot L_{\max}$. 
    If $j = 0$,
    $\widetilde{L} = 2^{j} = 1 < 2 L_{max}$ because $L_{max} \geq 1$.
    Otherwise, if $j > 0$, the $\Pi_{\ba}$ instance with index $j - 1$ returned $0$. The Validity condition of $\Pi_{\ba}$ implies that at least one honest party has joined instance $j - 1$ with input $0$, hence $L_{\max} > 2^{j - 1}$. Hence, $\widetilde{L} = 2^j = 2 \cdot 2^{j - 1} < 2 \cdot L_{\max}$.
 \end{proof}

\subsection{High-Communication $\ca$} \label{subsection:high-communication-CA}

For the high-communication $\ca$ protocol, we use the protocol of \cite{OPODIS:StolWat15} with minor adjustments described in \cite{PODC:GhiLiuWat25}. 
\begin{restatable}[Theorem 4 of \cite{OPODIS:StolWat15}; Theorem 3 of \cite{PODC:GhiLiuWat25}]{theorem}{LargeBitCostCa}
\label{theorem:large-bitcost-ca}
    There is a $\ca$ protocol $\largebitcostca$ for $\naturalnumbers$ resilient against $t < n / 3$ corruptions, with communication complexity $\bitcost_L(\largebitcostca) = 
      O(L \cdot n^3),
    $
    and round complexity $\roundcost_{L}(\largebitcostca) =  O(n)$.
\end{restatable}

\subsection{$\ca$ despite Missing Inputs} \label{subsection:ca-bot} 

We assume a committee of $n_C$ parties, where up to $t_C$ may be byzantine, and up to $t_C$ may have input $\bot$. In addition, we assume that an upper bound $\widetilde{L}$ on the honest inputs' bit-length is known a priori.
In this setting, 
we describe a protocol $\botCommitteeCA$ that achieves $\ca$ whenever $t_C < n_C / (\omega + 1)$. To ensure efficient communication, we make use of the supersending primitive described in \Cref{lemma:supersending}, relying on the $\ba$ protocol of \Cref{thm:magic-ba-no-pki}. This provides the honest parties with an identical view: this is a multiset $\msgset$ containing $n_C - 2t_C + k \leq n_C$ (non-$\bot$) values out of which at least $n_C - 2t_C$ come from honest parties, and up to $t_C$ come from byzantine parties.
Hence, up to $\min(k, t_C)$ of the values may be outside the honest inputs' convex hull. Each honest party computes the safe area $S := \safe_{\min(k, t_C)}(\msgset)$ and outputs a value in $S$.

To prove that safe areas are non-empty, we make use of a stronger variant of \cref{lemma:safe-area}, which follows from \cite[Lemma 15]{eprint:ConvexWorld} and the safe area definition. We include the proof in \Cref{appendix:safe-area-proofs}. 
\begin{restatable}{lemma}{BotSafeArea} \label{lemma:bot-friendly-safe-area}
    Let $\msgset$ be a multiset of $n - 2t \leq n - 2t + k \leq n$ values from a convexity space with Helly number $\omega$. $\safe_{\min(k, t)}(\msgset) \subseteq \hull{\mathcal{H}}$ for any multiset $\mathcal{H} \subseteq \msgset$ of size $n - 2t + k - \min(k, t)$. Moreover, if $t < n / (\omega + 1)$, $\safe_{\min(k, t)}(\msgset) \neq \emptyset$.
\end{restatable}

We present the code of this protocol, and \Cref{theo:convex-with-bot} describes its guarantees. The protocol's code assumes $t_C := \lceil n_C / (\omega + 1) \rceil - 1$.
\begin{dianabox}{$\botCommitteeCA$}
	\algoHead{Code for party $\party$ with input $v_{\inputt}$}
	\begin{algorithmic}[1]
    \State Send $v_{\inputt}$ to all parties via $\supersend^{\widetilde{L}}$.
    \State Let $\msgset$ be the multiset of non-$\bot$ $\widetilde{L}$-bit values received, and $k := \abs{\msgset} - (n_C - 2t_C)$.
    \State If $k < 0$ or $\safe_{\min(k, t_C)}(\msgset) = \emptyset$, output $\bot$.
    \State Otherwise, let $v \in \safe_{\min(k, t_C)}$ with the smallest bit-length, breaking ties deterministically. 
    \State If $v$ consists of more than $\delta \cdot \widetilde{{L}}$ bits, output $\bot$. Otherwise, output $v$.
	\end{algorithmic}
\end{dianabox}

\begin{lemma}\label{theo:convex-with-bot}
    Let $\mathcal{C}$ be a convexity space with Helly number $\omega \geq 2$ and dilation factor $\delta$, and consider a set of $n_C$ parties.
    Assume that every honest party joins $\botCommitteeCA$ either with an $\widetilde{L}$-bit input in $\mathcal{C}$ or $\bot$.
    
    Then, every honest party completes the execution of $\botCommitteeCA$ with either a $(\delta \cdot \widetilde{L})$-bit output in $\mathcal{C}$ or $\bot$, within $\roundcost_{n_C}(\botCommitteeCA)  = O(n_C)$ rounds and with communication complexity $\bitcost_{\widetilde{L}, n}(\botCommitteeCA) = \bigO(\widetilde{L} n_C^2 + \kappa n_C^3 \log n_C)$.

    Moreover, if up to $t_C < n_C / 3$ of the parties are byzantine, the honest parties obtain the same output. If, in addition, up to $t_C < n_C / (\omega + 1)$ of the parties are byzantine and up to $t_C$ of the honest parties join with input $\bot$, the honest parties output the same $(\delta \cdot \widetilde{L})$-bit valid value.
\end{lemma}
\begin{proof}
    Regardless of the number of byzantine parties involved, as $\widetilde{L}$ is known, \Cref{lemma:supersending} ensures that $\botCommitteeCA$ has round complexity $\roundcost_{1, n_C}(\supersend^{\widetilde{L}}) = \bigO(n_C)$ and communication complexity 
    $n_C \cdot \bitcost_{\widetilde{L}, 1, n_C} (\supersend^{\widetilde{L}}) = \bigO(\widetilde{L} n_C^2 + \kappa n_C^3 \log n_C)$.

    Next, we assume that $t_C < n_C/3$ of the parties are byzantine, and we describe the honest parties' view after the supersending step. As this is a $\ba$-friendly setting for the $\Pi_{\ba}$ protocol of \Cref{thm:magic-ba-no-pki}, \Cref{lemma:supersending} ensures that the honest parties obtain the same $\widetilde{L}$-bit message in each invocation of $\supersend^{\widetilde{L}}$.
    Hence, honest parties obtain the same multiset $\msgset$, and therefore the same output.

    Afterwards, assume that up to $t_C < n_C / (\omega + 1)$ of the parties are byzantine, and at most $t_C$ of the honest parties hold input $\bot$, while the remaining honest parties hold valid inputs. \Cref{lemma:supersending} ensures that, in each supersending invocation having as sender an honest party with $\widetilde{L}$-bit message $m$, the honest parties obtain message $m$. Hence, if $\msgset$ contains $n_C - 2t_C + k$ values, with $0 \leq k \leq 2t_C$, up to $t_C$ of these may come from byzantine parties, and at least $n_C - 2t_C$ of these are guaranteed to come from honest parties. This means that up to $\min(k, t_C)$ of the values in $\msgset$ come from byzantine parties, and the remaining $n_C - 2t_C + k - \min(k, t_C)$ values are valid. \Cref{lemma:bot-friendly-safe-area} ensures that $\safe_{\min(k, t_C)}(\msgset)$ is included in the convex hull of the at least $n_C - 2t_C + k - \min(k, t_C)$ valid values in $\msgset$. Moreover, if $t_C < n_C / (\omega + 1)$, \Cref{lemma:bot-friendly-safe-area} ensures that $\safe_{\min(k, t_C)}(\msgset) \neq \emptyset$. This implies that all honest parties obtain a valid value $v$. As the input space has dilation factor $\delta$, $v$ consists of up to $\delta \cdot \widetilde{L}$ bits. 
\end{proof}

\subsection{Putting It All Together} \label{subsection:unknown-L-all-done}

Now that all subprotocols have been defined, we provide the formal guarantees of $\Pi_{\ca}$. We first establish that the honest parties agree on a reasonable upper bound $\widetilde{L}$.

\begin{lemma}\label{lemma:unk-l-proof}
    Let $L$ denote the maximum honest input bit-length. Then, the honest parties agree on $\widetilde{L} \leq 2L$ such that at most $t$ honest parties have inputs of more than $\widetilde{L}$ bits.
    Moreover, this step is achieved within $\bigO(n)$ rounds and with communication complexity $\bigO(L n + n^2 \log n)$.
\end{lemma}
\begin{proof}
    In the initial exchange, every honest party defines $L'$ as the $(n - t)$-th lowest bit-length received and there are up to $t$ byzantine parties involved. Then, each honest party obtains a value $L' \leq L$ such that at most $t$ honest parties hold inputs of more than $L'$ bits.
    We denote the lowest and respectively highest honest values $L'$ by $L'_{\min}$ and $L'_{\max}$. 

    Note that this initial exchange step consists of $1$ round of communication, and has communication complexity $\bigO(n^2 \log L)$. We remark that, if $L \geq n^2$, then $n^2 \log L = \bigO(nL)$, and, if $L < n^2$, then $n^2 \log L = \bigO(n^2 \log n)$. In both cases, the communication complexity for this step is $\bigO(nL + n^2 \log n)$.

    Next, we establish that the honest parties obtain $\widetilde{L}$ such that $L'_{\min} \leq \widetilde{L} \leq 2 \cdot L'_{\max}$. We consider two cases, based on the bit returned by $\Pi_{\ba}$. According to  \Cref{thm:magic-ba-no-pki}, $\Pi_{\ba}$ adds $\bigO(n)$ rounds and $\bigO(n^2)$ communication.
    
    If $\Pi_{\ba}$ returns $b_{\outputt} = 1$, we first note that at least one honest party has joined with input $b_{\outputt} = 1$ due to the Validity condition of $\Pi_{\ba}$, hence $L'_{\max} \geq n^2$. In turn, this means that $L \geq n^2$ as well. In this case, the parties run $\largebitcostca$. According to \Cref{theorem:large-bitcost-ca}, the parties agree on an integer value in $\left[ \lceil L'_{\min} / n^2 \rceil, \lceil L'_{\max} / n^2 \rceil \right]$. Then, the honest parties obtain a value $\widetilde{L} \in \left[ \lceil L'_{\min} / n^2 \rceil \cdot n^2, \lceil L'_{\max} / n^2  \rceil \cdot n^2 \right]$.  Note that $\lceil L'_{\min} / n^2 \rceil \cdot n^2 \geq L_{\min}'$ and  $\lceil L'_{\max} / n^2  \rceil \cdot n^2 \leq L'_{\max} + n^2 \leq 2 \cdot L'_{\max}$. Consequently, $L'_{\min} \leq \widetilde{L} \leq 2 \cdot L'_{\max}$.
    Regarding communication complexity in this case, \Cref{theorem:large-bitcost-ca} ensures that $\largebitcostca$ adds $\bigO(n)$ rounds and $\bigO(\log (L/n^2) \cdot n^3)$ bits of communication, hence $\bigO(Ln)$ bits of communication.

    If $\Pi_{\ba}$ returns $b_{\outputt} = 0$, due to the Validity condition of $\Pi_{\ba}$, there is at least one honest party that has joined with input $b_{\inputt} = 0$, and therefore $L'_{\min} \leq n^2$. In this case, the parties agree on $\widetilde{L}$ via $\exponentialSearch$. Then, \Cref{theo:expo-seach-theo} ensures that the honest parties agree on an integer $\widetilde{L}$ such that $L'_{\min} \leq \widetilde{L} \leq 2 \cdot L'_{\max}$, and that this step adds $O(n)$ rounds, and $O(n^2 \log n)$ bits of communication.
    
    Finally, as $L'_{\min} \leq \widetilde{L}$ and $L'_{\min}$ comes with the property that there are at most $t$ honest parties holding inputs of length greater than $L'_{\min}$, we may conclude that there are at most $t$ honest parties holding inputs of length higher than $\widetilde{L}$.

    Summing up the round complexity and the communication complexity of each intermediate step, we obtain round complexity $\bigO(n)$ and communication complexity $\bigO(L n + n^2 \log n)$, as claimed in the lemma's statement. 
\end{proof}

We may now conclude the section with the theorem describing $\Pi_{\ca}$. 
\Cref{theo:abstract-for-1-unknown-L} follows from \Cref{theo:ca-unknown-resiliency} by instantiating $\Pi_{\ca}^{L}$ with the unauthenticated protocol described in \Cref{theo:abstract-for-1-fixed-L}, and \Cref{theo:abstract-for-finite-unknown-L} follows from \Cref{theo:ca-unknown-resiliency} by instantiating $\Pi_{\ca}^{L}$ with the unauthenticated protocol described in \Cref{theo:abstract-for-finite-fixed-L}.

\begin{theorem}\label{theo:ca-unknown-resiliency}
    Consider an arbitrary constant $\varepsilon > 0$ and a convexity space $\mathcal{C}$ with constant Helly number $\omega \geq 2$ and finite dilation factor $\delta$. 
    Assume a protocol $\Pi_{\ca, \varepsilon}^L$ that achieves $\ca$ whenever: up to $t < n / (\max(3,\omega) + \varepsilon)$ of the $n$ parties involved are corrupted and the honest parties hold $L$-bit inputs.
    
    Then, $\Pi_{\ca}$ achieves $\ca$ on $\mathcal{C}$ even when up to $t < n/(\omega + 1 + \varepsilon)$ of the $n$ parties involved are byzantine, with  round complexity $\bigO(n) + \roundcost_{n}(\Pi_{\ca}^{\delta \cdot \widetilde{L}})$ and communication complexity $\bigO(L n + n^2\cdot \log n) + \bitcost_{\delta \cdot \widetilde{L}, n}(\Pi_{\ca}^{\delta \cdot \widetilde{L}})$, where $\widetilde{L} \leq 2 \cdot L$.
\end{theorem}
\begin{proof}
    Let $L$ denote the maximum among the honest inputs' lengths.
    By \Cref{lemma:unk-l-proof}, the honest parties agree on an integer $\widetilde{L}$ satisfying $ \widetilde{L} \leq 2 \cdot L$ within $\bigO(n)$ rounds, and with communication complexity $\bigO(L n + n^2 \log n)$. Moreover, $\widetilde{L}$ comes with the guarantee that at most $t$ of the honest parties hold inputs of more than $\widetilde{L}$ bits. 

    By \Cref{lemma:unkown-committee-assigment}, each committee $C_i$ has size $O(1)$. Then, running $\botCommitteeCA$ within a single committee has communication complexity $\bigO(L + \kappa)$ and round complexity $\bigO(1)$ according to \Cref{theo:convex-with-bot}. As $\botCommitteeCA$ is executed within the $n$ committees, this leads to a total of $\bigO(nL + n\kappa)$ bits of communication and $\bigO(1)$ rounds.

    \Cref{theo:convex-with-bot} ensures that each honest party obtains either $\bot$ or a $(\delta \cdot \widetilde{L})$-bit value for each committee it belongs to. Then, in the next step, each committee $C_i$ sends a $(\delta \cdot \widetilde{L})$-bit value to party $\party_i$. Note that $\delta \cdot \widetilde{L} \in \bigO(L)$ since $\delta \in \bigO(1)$ and \Cref{lemma:unk-l-proof} ensures that $\tilde{L} \leq 2 \cdot L$. As every party belongs to $\bigO(1)$ committees according to \Cref{lemma:unkown-committee-assigment}, we obtain that every party sends $\bigO(1)$ values, hence $\bigO(L)$ bits. Over all parties, this sums up to $\bigO(Ln)$ bits. This step takes one round of communication.

    Finally, the parties run $\Pi_{\ca}^{\delta \cdot \widetilde{L}}$, setting the corruption threshold $t' < n / (\max(\omega, 3) + \varepsilon / 2)$. In the following, we show that this threshold is respected, i.e., there are at most $t'$ parties that are either byzantine or honest but unable to join $\Pi_{\ca}^{\delta \cdot \widetilde{L}}$ with a valid input.
    
    As the committee assignment is made using \cref{lemma:unkown-committee-assigment},
    we obtain that only a $\mu$-fraction of the committees may not satisfy the constraints stated in \cref{theo:convex-with-bot} for  $\botCommitteeCA$ to return a valid $(\delta \cdot \widetilde{L})$-bit value.
    \Cref{lemma:unkown-committee-assigment} ensures that each of the remaining $(1 - \mu) \cdot n$ committees has less than a $1/(\omega + 1)$-fraction of byzantine  parties. In addition, as $\omega \geq 2$, this implies that the fraction of byzantine parties in each of these committees is less than $1/3$.
    Hence, each such committee $C_i$ satisfies the constraints of \Cref{theo:convex-with-bot}, and therefore each honest party $\party_i$ corresponding to such a committee $C_i$ obtains the same $(\delta \cdot \widetilde{L})$-bit valid value from all honest parties in $C_i$. Then, since $C_i$ contains an honest majority, $\party_i$ accepts this value as its new input.
    
    This implies that there is at most a $\mu$-fraction of parties that are honest and do not receive a $(\delta \cdot \widetilde{L})$-bit valid value. Then, the number of parties $t'$ that are byzantine or honest but without a valid value can be upper bounded as follows:
    \begin{align*}
    t' \leq t + \mu n <  \frac{n}{\omega + 1 + \varepsilon} + \frac{n}{\omega + 1 +\varepsilon / 2} - \frac{n}{\omega + 1 + \varepsilon}  = \frac{n}{\omega + 1 + \varepsilon/2} \leq \frac{n}{\max(\omega, 3) + \varepsilon / 2},
    \end{align*}
    where the last step comes from the fact that, as $\omega \geq 2$, we obtain $\omega + 1 + \varepsilon/2 \geq \max(\omega, 3) + \varepsilon / 2$. 

    Then, we are running $\Pi_{\ca}^{\delta \cdot \widetilde{L}}$ in a setting where all but fewer than $t' < n / (\max(\omega, 3) + \varepsilon / 2)$ of the parties involved are honest and join with valid $(\delta \cdot \widetilde{L})$-bit inputs, as required. 
    Applying \Cref{theorem:ca-L-final}, this implies that 
    $\Pi_{\ca}^{\delta \cdot \widetilde{L}}$ provides all honest parties with a value $v$ in the convex hull of these valid $(\delta \cdot \widetilde{L})$-bit inputs, hence also in the original honest inputs' convex hull. Consequently, $\Pi_{\ca}$ achieves $\ca$. In addition, we obtain that $\Pi_{\ca}$ has a total round complexity  $\bigO(n) + \roundcost_{n}(\Pi_{\ca}^{\delta \cdot \widetilde{L}})$ and a total communication complexity of $\bigO(L n + n^2\cdot \log n) +  \bitcost_{{\delta \cdot \widetilde{L}}, n}(\Pi_{\ca}^{\delta \cdot \widetilde{L}})$. Using the costs stated in \Cref{theorem:ca-L-final}, we obtain the round complexity and communication complexity claimed in our theorem's statement.
\end{proof}

\section{Safe Area Proofs} \label{appendix:safe-area-proofs}

We describe the proofs of \Cref{lemma:safe-area} and \Cref{lemma:bot-friendly-safe-area}, relying on a result from \cite{eprint:ConvexWorld}, stated below.
\begin{lemma}[Lemma 15 of \cite{eprint:ConvexWorld}]\label{lemma:original-safe-area} 
    Assume $n' > \max(\omega \cdot t_s, \omega \cdot t_a + t_s)$, and that $\msgset$ is a multiset of $n - t_s + k$ values, where $0 \leq k \leq t_s$. Then, $\safe_{\max(k, t_a)} (\msgset) \neq \emptyset$.
\end{lemma}

\NonEmptySafeArea*
\begin{proof}
    By \Cref{definition:safe-area}, $\safe_k(\msgset)$ is computed as the intersection of the convex hulls of all subsets of size $n - t_s$ of $\msgset$. Then, as $\mathcal{H}$ is such a subset, it immediately follows that  $\safe_k(\msgset) \subseteq \hull{\mathcal{H}}$.

    Applying \Cref{lemma:original-safe-area} with $n' := n$, $t_s := t$, $t_a := 0$, we obtain that $\safe_k(\msgset) \neq \emptyset$.
\end{proof}

\BotSafeArea*
\begin{proof}
    We first show that $\safe_{\min(k, t)}(\msgset) \subseteq \hull{\mathcal{H}}$: according to \Cref{definition:safe-area}, $\safe_{\min(k, t)}(\msgset)$ is computed as the intersection of the convex hulls of all subsets of size $n - 2t + k - \min(k, t)$ of $\msgset$. Then, as $\mathcal{H}$ is such a subset, the claim follows immediately. In the remainder of the proof, we show that $\safe_{\min(k, t)}(\msgset) \neq \emptyset$.
    
    First assume that $\abs{\msgset} = n - 2t + k$ such that $k \leq t$, hence $\min(k, t) = k$. Then, we apply \Cref{lemma:original-safe-area} for $n' = n - t$, $t_s := t$ and $t_a := 0$. Hence, $\msgset$ is a multiset of $n' - t_s + k$ values, with $0 \leq k \leq t = t_s$.
    As $n > (\omega + 1) \cdot t$, hence $n' = (n - t) > \max(\omega \cdot t_s, \omega \cdot 0 + t_s)$, we obtain that $\safe_{\max(k, t_a)}(\msgset) = \safe_{k}(\msgset) \neq \emptyset$.

    Next, assume that $\abs{\msgset} = n - 2t + k$ with $k > t$, hence $\min(k, t) = t$. In this case, we apply \Cref{lemma:original-safe-area} for $n' = n$, $t_s := t$, $t_a := t$, hence $\msgset$ is a multiset of $n' - t_s + k'$ values, with $k' = k - t $ and, as $t < k \leq 2t$, $0 < k' \leq t$. As $n > (\omega + 1) \cdot t = t_s$, we obtain that $n' > \max(\omega \cdot t_s, \omega \cdot t_a + t_s)$, hence $\safe_{\max(k', t_a)}(\msgset) = \safe_t(\msgset) \neq \emptyset$.
\end{proof}
\end{document}